\documentclass[runningheads,citeauthoryear]{llncs}

\newcommand{\ifAnon}[2]{#2}

\usepackage[T1]{fontenc}
\usepackage{graphicx} %
\usepackage{amsmath}  %
\usepackage{amssymb}  %
\usepackage{booktabs} %
\usepackage{subcaption} %
\usepackage{adjustbox}
\usepackage{color}
\usepackage{colortbl}
\usepackage{stmaryrd}
\usepackage{xcolor}
\usepackage{csquotes}
\usepackage{xfrac}
\usepackage{xspace}
\usepackage{colortbl}
\usepackage{mathtools}
\usepackage{stmaryrd}
\usepackage[
	pdfencoding=auto,
	psdextra,
	colorlinks=true,
	linkcolor=blue!50!black,
	citecolor=blue!50!black,
	urlcolor=blue!50!black,
	bookmarks=true,
	bookmarksopen=true,
	bookmarksnumbered=true,
	pdfstartview=FitH,
	bookmarksdepth=2 %
]{hyperref}
\usepackage[nameinlink]{cleveref}
\renewcommand{\cref}{\Cref}
\usepackage{thm-restate}
\usepackage{multirow}
\usepackage{braket}
\usepackage{siunitx}
\usepackage{tikz}
\usetikzlibrary{positioning}
\usetikzlibrary{shapes.geometric, arrows}
\usepackage{pgfplots}
\pgfplotsset{compat=newest}
\usepgfplotslibrary{groupplots}
\usepackage[normalem]{ulem}
\usepackage[most]{tcolorbox}
\usepackage{nicefrac}
\usepackage{dirtytalk}
\usepackage{transparent}
\usepackage{wrapfig}
\usepackage[inline]{enumitem}
\usepackage{listings}
\usepackage{array}

\usepackage[
  backend=biber,
  style=lncs
]{biblatex}
\DeclareSourcemap{
  \maps[datatype=bibtex]{
    \map{
      \step[fieldsource=url, match=\regexp{\\_}, replace=_]
    }
    \map{
      \step[fieldsource=doi, match=\regexp{\\_}, replace=_]
    }
  }
}

\providecommand{\Description}[1]{}
\newcommand{\qedhere}{\hfill\ensuremath{\square}}

\makeatletter
\providecommand{\leftsquigarrow}{%
  \mathrel{\mathpalette\reflect@squig\relax}%
}
\newcommand{\reflect@squig}[2]{%
  \reflectbox{$\m@th#1\rightsquigarrow$}%
}
\makeatother

\makeatletter
\let\c@lemma=\c@theorem
\let\c@definition=\c@theorem
\let\c@example=\c@theorem
\makeatother

\newcommand{\morespace}[1]{~{}#1{}~}
\newcommand{\qmorespace}[1]{\quad{}#1{}\quad}
\newcommand{\qqmorespace}[1]{\qquad{}#1{}\qquad}

\newcommand{\eeq}{\morespace{=}}

\newcommand{\lleq}{\morespace{\leq}}

\newcommand{\eexpleq}{\morespace{\expleq}}

\newcommand{\qiff}{\qmorespace{\textnormal{iff}}}

\newcommand{\qimplies}{\qmorespace{\textnormal{implies}}}

\newcommand{\qqand}{\qqmorespace{\textnormal{and}}}

\definecolor{hintgray}{RGB}{92,92,92}

\definecolor{col1}{RGB}{254,195,10}
\definecolor{col2}{RGB}{251,0,6}
\definecolor{col3}{RGB}{252,11,135}
\definecolor{col4}{RGB}{17,139,2}
\definecolor{col5}{RGB}{14,131,254}
\definecolor{col6}{RGB}{82,0,135}

\definecolor{brLightGreen}{HTML}{addd8e}
\definecolor{brLightGray}{HTML}{f0f0f0}

\definecolor{orange}{RGB}{230,159,0}
\definecolor{skyblue}{RGB}{86,180,233}
\definecolor{brBlue}{RGB}{0,114,178}
\definecolor{bluishgreen}{RGB}{0,158,115}
\definecolor{vermillion}{RGB}{213,94,0}
\definecolor{reddishpurple}{RGB}{204,121,167}

\colorlet{heyvlColor}{vermillion!60!black}
\definecolor{prepostColor}{RGB}{0,69,107}
\colorlet{stmtColor}{bluishgreen!50!black}

\newcommand{\gutter}[1]{\text{\color{gray}\footnotesize{}{#1}}\quad}
\newcommand{\lineNumber}[1]{\text{\color{gray}\footnotesize{}{(l.~#1)}}}

\newcommand{\highlightCode}[2]{\hspace{-\fboxsep}\colorbox{#1}{#2}}

\newcommand{\Sup}{\reflectbox{\textnormal{\textsf{\fontfamily{phv}\selectfont S}}}\hspace{.2ex}}
\newcommand{\Inf}{\raisebox{.6\depth}{\rotatebox{-30}{\textnormal{\textsf{\fontfamily{phv}\selectfont \reflectbox{J}}}}\hspace{-.1ex}}}

\newcommand{\quant}[2]{#1.~#2}
\newcommand{\squant}[2]{\ensuremath{\Sup \quant{#1}{#2}}}
\newcommand{\iquant}[2]{\ensuremath{\Inf \quant{#1}{#2}}}

\newcommand{\true}{\mathsf{true}}
\newcommand{\false}{\mathsf{false}}

\newcommand{\impl}{\rightarrow}
\newcommand{\coimpl}{\leftsquigarrow}

\newcommand{\colheylo}[1]{{\color{prepostColor}#1}}

\newcommand{\expa}{\colheylo{\ensuremath{X}}}
\newcommand{\expb}{\colheylo{\ensuremath{Y}}}
\newcommand{\expc}{\colheylo{\ensuremath{Z}}}

\newcommand{\expaP}{\colheylo{\ensuremath{X'}}}
\newcommand{\expbP}{\colheylo{\ensuremath{Y'}}}

\newcommand{\expleq}{\ensuremath{\preceq}}
\newcommand{\expgeq}{\ensuremath{\succeq}}

\newcommand{\expvalidate}[1]{\ensuremath{\triangle\!\left( #1 \right)}}
\newcommand{\expcovalidate}[1]{\ensuremath{\triangledown\!\left( #1 \right)}}

\newcommand{\substBy}[2]{[{#1} \mapsto {#2}]}

\newcommand{\iverson}[1]{\left[{#1}\right]}

\newcommand{\ite}[3]{\mathrm{ite}({#1}, {#2}, {#3})}

\newcommand{\lam}[2]{\lambda {#1}.~{#2}}

\newcommand{\ifThenElse}[3]{
	\begin{cases}{#2}, & \text{if } {#1}\\{#3}, & \text{otherwise}
	\end{cases}
}

\newcommand{\pexp}[2]{#1 \cdot \langle #2 \rangle }
\newcommand{\pexpand}{+}
\newcommand{\sfsymbol}[1]{\textsf{\upshape {#1}}}
\newcommand{\Vars}{\sfsymbol{Vars}\xspace}

\newcommand{\typeof}[2]{\ensuremath{#1 \colon {#2}}}

\newcommand{\pGCL}{\sfsymbol{pGCL}\xspace}
\newcommand{\HeyVL}{\sfsymbol{HeyVL}\xspace}
\newcommand{\HeyLo}{\sfsymbol{HeyLo}\xspace}

\newcommand{\typevar}{\ensuremath{\tau}}

\newcommand{\aexpr}{\ensuremath{a}} %
\newcommand{\bexpr}{\ensuremath{b}} %
\newcommand{\procname}{P} %

\newcommand{\termvar}{\ensuremath{t}}

\newcommand{\Bools}{\mathbb{B}}
\newcommand{\Nats}{\mathbb{N}}
\newcommand{\Ints}{\mathbb{Z}}

\newcommand{\PosReals}{\mathbb{R}_{\geq 0}}
\newcommand{\PosRealsInf}{\mathbb{R}_{\geq 0}^{\infty}}

\newcommand{\States}{\mathsf{States}}
\newcommand{\State}{\sigma}

\newcommand{\symStmt}[1]{\ensuremath{{\color{stmtColor}{#1}}}}
\newcommand{\symSkip}{\symStmt{\texttt{skip}}}
\newcommand{\symAssign}{\coloneqq}
\newcommand{\symRasgn}{\colonapprox}
\newcommand{\symSemi}{\symStmt{\texttt{;}~}}
\newcommand{\symIf}{\symStmt{\texttt{if}}}
\newcommand{\symElse}{\symStmt{\texttt{else}}}
\newcommand{\symAssert}{\symStmt{\texttt{assert}}}
\newcommand{\symAssume}{\symStmt{\texttt{assume}}}

\newcommand{\symDemonic}{\symStmt{\symIf~(\sqcap)}}

\newcommand{\symAngelic}{\symStmt{\symIf~(\sqcup)}}

\newcommand{\symHavoc}{\symStmt{\texttt{havoc}}}
\newcommand{\symWhile}{\symStmt{\texttt{while}}}

\newcommand{\symObserve}{\symStmt{\texttt{observe}}}
\newcommand{\symTick}{\symStmt{\texttt{reward}}}
\newcommand{\symProc}{\symStmt{\texttt{proc}}}
\newcommand{\symcoProc}{\symStmt{\texttt{coproc}}}

\newcommand{\symPost}{\symStmt{\texttt{post}}}
\newcommand{\symValidate}{\symStmt{\texttt{validate}}}
\newcommand{\symVar}{\symStmt{\texttt{var}}}
\newcommand{\symAnnotate}{\symStmt{\texttt{@}}}

\newcommand{\blockStart}{\ensuremath{\{}}
\newcommand{\blockEnd}{\ensuremath{\}}}

\newcommand{\stmt}{{\color{heyvlColor}\ensuremath{S}}}
\newcommand{\stmtP}{{\color{heyvlColor}\ensuremath{S'}}}
\newcommand{\stmtC}{{\color{heyvlColor}\ensuremath{S_C}}}

\newcommand{\stmtOne}{{\color{heyvlColor}\ensuremath{S_1}}}
\newcommand{\stmtOneP}{{\color{heyvlColor}\ensuremath{S_1'}}}
\newcommand{\stmtTwo}{{\color{heyvlColor}\ensuremath{S_2}}}
\newcommand{\stmtTwoP}{{\color{heyvlColor}\ensuremath{S_2'}}}
\newcommand{\stmtN}[1]{{\color{heyvlColor}\ensuremath{S_{#1}}}}
\newcommand{\stmtI}{{\stmtN{i}}}
\newcommand{\stmtITransformed}{{\color{heyvlColor}\ensuremath{\tilde{\stmtI}}}}
\newcommand{\stmtTransformed}{{\color{heyvlColor}\ensuremath{\tilde{\stmt}}}}

\newcommand{\stmtSkip}{\symSkip}
\newcommand{\stmtAsgn}[2]{\ensuremath{{#1} \symAssign {#2}}}

\newcommand{\stmtDeclInit}[3]{\symVar~\ensuremath{{#1}\colon {#2} \symRasgn {#3}}}

\newcommand{\stmtRasgn}[2]{\ensuremath{{#1} \symRasgn {#2}}}
\newcommand{\stmtSeq}[2]{\ensuremath{{#1}\symSemi {#2}}}
\newcommand{\stmtIf}[3]{\ensuremath{\symIf~({#1})~\{ {#2} \}~\symElse~\{ {#3} \}}}
\newcommand{\stmtIfStart}[1]{\ensuremath{\symIf~({#1})~\{ }}
\newcommand{\stmtProb}[3]{\ensuremath{\{ {#2} \} ~[{#1}]~ \{ {#3} \}}}
\newcommand{\stmtWhile}[2]{\ensuremath{\symWhile~({#1})~\{ {#2} \}}}
\newcommand{\headerWhile}[1]{\ensuremath{\symWhile~({#1})}}

\newcommand{\stmtHavoc}[1]{\ensuremath{\symHavoc~{#1}}}
\newcommand{\stmtAssert}[1]{\ensuremath{\symAssert~{#1}}}
\newcommand{\stmtAssume}[1]{\ensuremath{\symAssume~{#1}}}

\newcommand{\stmtSpec}[3]{{#1} : [{#2},~ {#3}]}

\newcommand{\stmtAnnotate}[3]{\ensuremath{\symAnnotate{}{\symStmt{\texttt{#1}}}({#2})~{#3}}}

\newcommand{\stmtDemonicStart}{\ensuremath{\symDemonic~\{ }}

\newcommand{\stmtElseStart}{\ensuremath{\symElse~\{ }}
\newcommand{\stmtAngelicStart}{\ensuremath{\symAngelic~\{ }}

\newcommand{\stmtDemonic}[2]{\ensuremath{\stmtDemonicStart{#1} \}~\stmtElseStart {#2} \}}}

\newcommand{\stmtAngelic}[2]{\ensuremath{\stmtAngelicStart{#1} \}~\stmtElseStart {#2} \}}}

\newcommand{\stmtObserve}[1]{\ensuremath{\symObserve~{#1}}}
\newcommand{\stmtTick}[1]{\ensuremath{\symTick~{#1}}}
\newcommand{\stmtValidate}{\ensuremath{\symValidate}}

\newcommand{\Havoc}[1]{\ensuremath{\symHavoc~{#1}}}
\newcommand{\Assert}[1]{\ensuremath{\symAssert~{#1}}}
\newcommand{\Assume}[1]{\ensuremath{\symAssume~{#1}}}

\newcommand{\Validate}{\ensuremath{\stmtValidate}}

\newcommand{\symUp}{\symStmt{\texttt{co}}}

\newcommand{\coHavoc}[1]{\ensuremath{\symUp\symHavoc~{#1}}}
\newcommand{\coAssert}[1]{\ensuremath{\symUp\symAssert~{#1}}}
\newcommand{\coAssume}[1]{\ensuremath{\symUp\symAssume~{#1}}}

\newcommand{\coValidate}{\ensuremath{\symUp\stmtValidate}}

\newcommand{\exprTrue}{\ensuremath{\texttt{true}}}
\newcommand{\exprFalse}{\ensuremath{\texttt{false}}}
\newcommand{\exprFlip}[1]{\ensuremath{\symStmt{\texttt{flip}}({#1})}}

\newcommand{\eval}[1]{\llbracket{#1}\rrbracket}
\newcommand{\evalState}[1]{\eval{#1}(\State)}
\newcommand{\evalStateSubstBy}[3]{\eval{#1}(\State\substBy{#2}{#3})} %

\newcommand{\exprIte}[3]{\mathrm{ite}({#1},~{#2},~{#3})}

\newcommand{\symWp}{\sfsymbol{wp}}
\newcommand{\symVc}{\sfsymbol{vp}}

\renewcommand{\wp}[1]{\symWp\llbracket{#1}\rrbracket}

\newcommand{\vc}[1]{\symVc\llbracket{#1}\rrbracket}

\newcommand{\symWlp}{\sfsymbol{wlp}}
\newcommand{\wlp}[1]{\symWlp\llbracket{#1}\rrbracket}

\newcommand{\symCwp}{\sfsymbol{cwp}}
\newcommand{\cwp}[1]{\symCwp\llbracket{#1}\rrbracket}

\newcommand{\Expectations}{\mathbb{E}}

\newcommand{\embed}[1]{\scalebox{0.85}{\textnormal{\textsf{?}}}({#1})}

\newcommand{\lowerTriple}[4][]{\langle {\colheylo{#2}} \rangle_{\expleq}^{#1}~{#3}~\langle {\colheylo{#4}} \rangle}
\newcommand{\upperTriple}[4][]{\langle {\colheylo{#2}} \rangle_{\expgeq}^{#1}~{#3}~\langle {\colheylo{#4}} \rangle}

\newcommand{\intersem}[1]{\color{hintgray}{//~#1}}

\newcommand{\slice}{\ensuremath{\stmtP}}

\newcommand*\circled[1]{\tikz[baseline=(char.base)]{
		\node[shape=circle,draw,inner sep=0.5pt] (char) {#1};}}

\newcommand{\stmtNo}[1]{\textnormal{\circled{#1}}}

\newcommand{\varEnabled}[1]{enabled_{#1}}

\newcommand{\ParkIndCondPre}{\ensuremath{\mathsf{PI}_{pre}}}
\newcommand{\ParkIndCondInd}{\ensuremath{\mathsf{PI}_{ind}}}
\newcommand{\ParkIndCondPost}{\ensuremath{\mathsf{PI}_{post}}}

\newcommand{\sliced}[1]{\text{\sout{\transparent{0.5}{\ensuremath{#1}}}}}

\newcommand{\sliceMethod}[1]{\texttt{#1}\xspace}
\newlength{\scatterplotsize}
\newlength{\numberruntimeplotheight}
\definecolor{barplotcolora}{HTML}{1b9e77}
\definecolor{barplotcolorb}{HTML}{d95f02}
\definecolor{barplotcolorc}{HTML}{7570b3}
\definecolor{barplotcolord}{HTML}{e7298a}

\tikzset{slicingMethod/.code={%
			\ifthenelse{\equal{#1}{first_cex}}{\tikzset{red, only marks, mark=x, mark size=1.5pt, line width=1pt}}{}%
			\ifthenelse{\equal{#1}{optimal_cex}}{\tikzset{blue, only marks, mark=+, mark size=1.5pt, line width=1pt}}{}%

			\ifthenelse{\equal{#1}{core}}{\tikzset{red, only marks, mark=*, mark size=1.5pt, line width=1pt}}{}%
			\ifthenelse{\equal{#1}{mus}}{\tikzset{blue, only marks, mark=+, mark size=1.5pt, line width=1pt}}{}%
			\ifthenelse{\equal{#1}{sus}}{\tikzset{green, only marks, mark=x, mark size=1.5pt, line width=1pt}}{}%
			\ifthenelse{\equal{#1}{exists_forall}}{\tikzset{orange, only marks, mark=o, mark size=1.5pt, line width=1pt}}{}%
		}}

\tikzset{slicingMethodBar/.code={%
			\ifthenelse{\equal{#1}{first_cex}}{\tikzset{red, only marks, mark=x, mark size=1.5pt, line width=1pt}}{}%
			\ifthenelse{\equal{#1}{optimal_cex}}{\tikzset{blue, only marks, mark=+, mark size=1.5pt, line width=1pt}}{}%

			\ifthenelse{\equal{#1}{core}}{\tikzset{barplotcolora,fill=barplotcolora}}{}%
			\ifthenelse{\equal{#1}{mus}}{\tikzset{barplotcolorb,fill=barplotcolorb}}{}%
			\ifthenelse{\equal{#1}{sus}}{\tikzset{barplotcolorc,fill=barplotcolorc}}{}%
			\ifthenelse{\equal{#1}{exists_forall}}{\tikzset{barplotcolord,fill=barplotcolord}}{}%
		}}
\newcommand{\numberruntimeplot}[7]{%
	\begin{tikzpicture}
		\begin{axis}[
				width=\scatterplotsize,
				height=\numberruntimeplotheight,
				ymin=0,
				ymax=#7,
				axis y line=left,
				ytick={0,5,10,15,20},
				extra y ticks = {25},
				extra y tick labels = {OOR},
				xmin=1,
				xmax=80000,
				xtick={10,100,1000},
				axis x line=bottom,
				xmode=log,
				extra x ticks = {1, 30000},
				extra x tick labels = {$\leq 1$, OOR},
				ymode=normal,
				ylabel= #4,
				xlabel= #5,
				xlabel style={font=\scriptsize},
				ylabel style={yshift=-6pt, font=\scriptsize},
				xticklabel style={font=\scriptsize},
				yticklabel style={font=\scriptsize},
				legend columns=\legendcols,
				legend style={\legendstyle, font=\scriptsize,
						nodes={scale=1, transform shape},inner sep=2pt},
				legend cell align={left},
				table/col sep=comma,
			]
			\foreach \slicingMethod in {#2}{%
					\edef\loopbody{
						\noexpand\addplot[slicingMethod=\slicingMethod ] table [
								x expr=\noexpand\thisrow{\slicingMethod_slice_time} > 0 ? \noexpand\thisrow{\slicingMethod_slice_time} : 1,
								y expr=\noexpand\thisrow{\slicingMethod_slice_size} > 10000 ? 25 : \noexpand\thisrow{\slicingMethod_slice_size},
								col sep=comma] {#1};
					}
					\loopbody
				}
			\draw[densely dotted] (axis cs: 10,0) -- (axis cs: 10,#7);
			\draw[densely dotted] (axis cs: 100,0) -- (axis cs: 100,#7);
			\draw[densely dotted] (axis cs: 1000,0) -- (axis cs: 1000,#7);
			\draw[densely dotted] (axis cs: 30000,0) -- (axis cs: 30000,#7);
			\draw[densely dotted] (axis cs: 0,25) -- (axis cs: 30000,25);
			\legend{#3}
		\end{axis}
	\end{tikzpicture}%
}

\newcommand{\relativenumberruntimeplot}[7]{%
	\begin{tikzpicture}
		\begin{axis}[
				width=\scatterplotsize,
				height=\numberruntimeplotheight,
				ymin=0,
				ymax=1.3,
				axis y line=left,
				ytick={0,0.2,...,1.0},
				yticklabel={\pgfmathparse{\tick*100}\pgfmathprintnumber{\pgfmathresult}\%},
				extra y ticks = {1.2},
				extra y tick labels = {OOR},
				xmin=1,
				xmax=80000,
				xtick={10,100,1000},
				axis x line=bottom,
				xmode=log,
				extra x ticks = {1, 30000},
				extra x tick labels = {$\leq 1$, OOR},
				ymode=normal,
				ylabel= #4,
				xlabel= #5,
				xlabel style={font=\scriptsize},
				ylabel style={yshift=-6pt, font=\scriptsize},
				xticklabel style={font=\scriptsize},
				yticklabel style={font=\scriptsize},
				legend columns=\legendcols,
				legend style={\legendstyle, font=\scriptsize,
						nodes={scale=1, transform shape},inner sep=2pt},
				legend cell align={left},
				table/col sep=comma,
			]
			\foreach \slicingMethod in {#2}{%
					\edef\loopbody{
						\noexpand\addplot[slicingMethod=\slicingMethod ] table [
								x expr=\noexpand\thisrow{\slicingMethod_slice_time} > 0 ? \noexpand\thisrow{\slicingMethod_slice_time} : 1,
								y expr=\noexpand\thisrow{\slicingMethod_slice_size} > 10000 ? 1.2 : (\noexpand\thisrow{\slicingMethod_slice_size}/\noexpand\thisrow{original_size}),
								col sep=comma] {#1};
					}
					\loopbody
				}
			\draw[densely dotted] (axis cs: 10,0) -- (axis cs: 10,#7);
			\draw[densely dotted] (axis cs: 100,0) -- (axis cs: 100,#7);
			\draw[densely dotted] (axis cs: 1000,0) -- (axis cs: 1000,#7);
			\draw[densely dotted] (axis cs: 30000,0) -- (axis cs: 30000,#7);
			\draw[densely dotted] (axis cs: 0,1.2) -- (axis cs: 30000,1.2);
			\legend{#3}
		\end{axis}
	\end{tikzpicture}%
}

\newcommand{\relativeslicesizeruntimeplot}[7]{%

	\begin{tikzpicture}
		\begin{groupplot}[group style = {
						group size = 1 by 2,
						vertical sep = 15},
				width = \linewidth,
				tick label style={font=\scriptsize},
				typeset ticklabels with strut,
				enlarge y limits=false,
				legend columns=\legendcols,
				legend style={\legendstyle, font=\scriptsize,
						nodes={scale=1, transform shape},inner sep=2pt},
				legend cell align={left},
			]
			\nextgroupplot[
				xmin=0, xmax=21,
				ymin=0, ymax=80000,
				height = 0.31\linewidth,
				ymode = log,
				axis y line* = left,
				grid=both,
				ybar = .04cm,
				bar width = 1.5pt,
				enlarge x limits = {value = 1},
				xticklabel=\empty,
				ytick={0.1,10,100, 1000, 10000, 30000},
				yticklabels={0,10, $10^2$, $10^3$, $10^4$, TO},
				xtick distance=1,
			]

			\def\slicingMethod{core}
			\edef\loopbody{
				\noexpand\addplot[slicingMethodBar=\slicingMethod ] table [
						x expr=\noexpand\thisrow{id},
						y expr=\noexpand\thisrow{\slicingMethod_slice_time} > 0 ? \noexpand\thisrow{\slicingMethod_slice_time} : 1,
						col sep=comma] {#1};
			}
			\loopbody

			\def\slicingMethod{mus}
			\edef\loopbody{
				\noexpand\addplot[slicingMethodBar=\slicingMethod ] table [
						x expr=\noexpand\thisrow{id},
						y expr=\noexpand\thisrow{\slicingMethod_slice_time} > 0 ? \noexpand\thisrow{\slicingMethod_slice_time} : 1,
						col sep=comma] {#1};
			}
			\loopbody

			\def\slicingMethod{sus}
			\edef\loopbody{
				\noexpand\addplot[slicingMethodBar=\slicingMethod ] table [
						x expr=\noexpand\thisrow{id},
						y expr=\noexpand\thisrow{\slicingMethod_slice_time} > 0 ? \noexpand\thisrow{\slicingMethod_slice_time} : 1,
						col sep=comma] {#1};
			}
			\loopbody

			\def\slicingMethod{exists_forall}
			\edef\loopbody{
				\noexpand\addplot[slicingMethodBar=\slicingMethod ] table [
						x expr=\noexpand\thisrow{id},
						y expr=\noexpand\thisrow{\slicingMethod_slice_time} > 0 ? \noexpand\thisrow{\slicingMethod_slice_time} : 1,
						col sep=comma] {#1};
			}
			\loopbody

			\nextgroupplot[
				xmin=0, xmax=21,
				ymin=-1.3, ymax=0,
				width = \linewidth, height = 0.25\linewidth,
				axis y line* = left,
				grid=both,
				ybar = .04cm,
				bar width = 1.5pt,
				enlarge x limits = {value = 1},
				ytick={0, -0.5, -1, -1.2},
				yticklabels={0\%, 50\%, 100\%, TO},
				xticklabel style={rotate=-45, anchor=west, font=\scriptsize\ttfamily},
				xtick distance=1,
				xticklabels={0,,navarro20\_page31,
						gehr18\_1,
						navarro20\_ex4\_3,
						navarro20\_ex4\_5,
						navarro20\_ex5\_8,
						navarro20\_figure7,
						navarro20\_figure8,
						2drwalk,
						bayesian\_network,
						prspeed,
						rdspeed,
						rdwalk,
						sprdwalk,
						unif\_gen4\_wlp,
						burglar\_alarm\_wp,
						park\_1,
						park\_5,
						park\_6,
						algorithm\_r,
						bn\_passified,}
			]

			\def\slicingMethod{core}
			\edef\loopbody{
				\noexpand\addplot[slicingMethodBar=\slicingMethod ] table [
						x expr=\noexpand\thisrow{id},
						y expr=( \noexpand\thisrow{\slicingMethod_slice_size} > 10000 ? -1.2 : ((-(\noexpand\thisrow{\slicingMethod_slice_size})/(\noexpand\thisrow{original_size})))),
						col sep=comma] {#1};
			}
			\loopbody

			\def\slicingMethod{mus}
			\edef\loopbody{
				\noexpand\addplot[slicingMethodBar=\slicingMethod ] table [
						x expr=\noexpand\thisrow{id},
						y expr=( \noexpand\thisrow{\slicingMethod_slice_size} > 10000 ? -1.2 : ((-(\noexpand\thisrow{\slicingMethod_slice_size})/(\noexpand\thisrow{original_size})))),
						col sep=comma] {#1};
			}
			\loopbody
			\def\slicingMethod{sus}
			\edef\loopbody{
				\noexpand\addplot[slicingMethodBar=\slicingMethod ] table [
						x expr=\noexpand\thisrow{id},
						y expr=( \noexpand\thisrow{\slicingMethod_slice_size} > 10000 ? -1.2 : ((-(\noexpand\thisrow{\slicingMethod_slice_size})/(\noexpand\thisrow{original_size})))),
						col sep=comma] {#1};
			}
			\loopbody
			\def\slicingMethod{exists_forall}
			\edef\loopbody{
				\noexpand\addplot[slicingMethodBar=\slicingMethod ] table [
						x expr=\noexpand\thisrow{id},
						y expr=( \noexpand\thisrow{\slicingMethod_slice_size} > 10000 ? -1.2 : ((-(\noexpand\thisrow{\slicingMethod_slice_size})/(\noexpand\thisrow{original_size})))),
						col sep=comma] {#1};
			}
			\loopbody

		\end{groupplot}
	\end{tikzpicture}
}

\newcommand{\relativeslicesizeruntimeplotscatter}[7]{%
	\begin{tikzpicture}
		\begin{axis}[
				width=\scatterplotsize,
				height=\numberruntimeplotheight,
				ymin=0,
				ymax=1.3,
				axis y line=left,
				ytick={0,0.2,...,1.0},
				yticklabel={\pgfmathparse{\tick*100}\pgfmathprintnumber{\pgfmathresult}\%},
				extra y ticks = {1.2},
				extra y tick labels = {OOR},
				xmin=1,
				xmax=80000,
				xtick={10,100,1000},
				axis x line=bottom,
				xmode=log,
				extra x ticks = {1, 30000},
				extra x tick labels = {$\leq 1$, OOR},
				ymode=normal,
				ylabel= #4,
				xlabel= #5,
				xlabel style={font=\scriptsize},
				ylabel style={yshift=-6pt, font=\scriptsize},
				xticklabel style={font=\scriptsize},
				yticklabel style={font=\scriptsize},
				legend columns=\legendcols,
				legend style={\legendstyle, font=\scriptsize,
						nodes={scale=1, transform shape},inner sep=2pt},
				legend cell align={left},
				table/col sep=comma,
			]
			\foreach \slicingMethod in {#2}{%
					\edef\loopbody{
						\noexpand\addplot[slicingMethod=\slicingMethod ] table [
								x expr=\noexpand\thisrow{\slicingMethod_slice_time} > 0 ? \noexpand\thisrow{\slicingMethod_slice_time} : 1,
								y expr=\noexpand\thisrow{sus_slice_size} > 10000 ? -10 : ( \noexpand\thisrow{\slicingMethod_slice_size} > 10000 ? 1.2 : ((\noexpand\thisrow{original_size}-\noexpand\thisrow{sus_slice_size} < 1) ? 0 : ((\noexpand\thisrow{original_size}-\noexpand\thisrow{\slicingMethod_slice_size})/(\noexpand\thisrow{original_size}-\noexpand\thisrow{sus_slice_size})))),
								col sep=comma] {#1};
					}
					\loopbody
				}
			\draw[densely dotted] (axis cs: 10,0) -- (axis cs: 10,#7);
			\draw[densely dotted] (axis cs: 100,0) -- (axis cs: 100,#7);
			\draw[densely dotted] (axis cs: 1000,0) -- (axis cs: 1000,#7);
			\draw[densely dotted] (axis cs: 30000,0) -- (axis cs: 30000,#7);
			\draw[densely dotted] (axis cs: 0,1.2) -- (axis cs: 30000,1.2);
			\legend{#3}
		\end{axis}
	\end{tikzpicture}%
}

\newcommand{\tikzxmark}{%
	\tikz[scale=0.23] {
		\draw[line width=0.7,line cap=round] (0,0) to [bend left=6] (1,1);
		\draw[line width=0.7,line cap=round] (0.2,0.95) to [bend right=3] (0.8,0.05);
	}}
\newcommand{\tikzcmark}{%
	\tikz[scale=0.23] {
		\draw[line width=0.7,line cap=round] (0.25,0) to [bend left=10] (1,1);
		\draw[line width=0.8,line cap=round] (0,0.35) to [bend right=1] (0.23,0);
	}}

\lstdefinelanguage{HeyVL}{
	morekeywords={var,while,if,else,@invariant,@slice_verify,@slice_error,Int,Bool,EUReal,UInt,proc,coproc,pre,post,flip,coassume,assume,assert,coassert,\cup,cohavoc,havoc},
	morecomment=[l]{//},
	sensitive=true
}

\definecolor{darkred}{rgb}{0.8, 0.0, 0.0}
\newcommand{\marknewtext}[1]{#1}

\begin{document}

\title{Error Localization, Certificates, and Hints for Probabilistic Program Verification via Slicing (Extended Version)}
\titlerunning{Diagnostics for Probabilistic Program Verification}

\ifAnon{
	\author{Anonymous\\Authors}%
	\authorrunning{Anonymous}
	\institute{\phantom{ }\\\phantom{ }}
}{
	\author{Philipp Schröer\inst{1}\orcidID{0000-0002-4329-530X} \and
		Darion Haase\inst{1}\orcidID{0000-0001-5664-6773} \and
		Joost-Pieter Katoen\inst{1}\orcidID{0000-0002-6143-1926}
	}%
	\authorrunning{P. Schröer et al.}
	\institute{RWTH Aachen University, Germany \\
		\email{\{phisch, darion.haase, katoen\}@cs.rwth-aachen.de}}
}

\maketitle

\begin{abstract}
    This paper focuses on effective user diagnostics generated during
    the deductive verification of probabilistic programs. Our key principle
    is based on providing slices for (1) error reporting,
    (2) proof simplification, and (3) preserving successful verification results.
    By formally defining these different notions on HeyVL, an existing quantitative intermediate
    verification language (IVL), our concepts (and implementation) can be
    used to obtain diagnostics for a range of probabilistic
    programming languages. Slicing for error reporting is a novel notion of error localization for quantitative assertions.
    We demonstrate slicing-based diagnostics on a variety of
    proof rules such as quantitative versions of the specification statement
    and invariant-based loop rules, and formally prove the correctness of
    specialized error messages and verification hints.

    We implemented our user diagnostics into the deductive verifier Caesar.
    Our novel implementation --- called \emph{Brutus} --- can search for slices which do or do not verify, corresponding to each of the three diagnostic notions.
    For error reporting (1), it exploits a binary search-based algorithm that minimizes error-witnessing slices.
    To solve for slices that verify (2 and 3), we empirically compare different algorithms based on unsatisfiable cores, minimal unsatisfiable subset enumeration, and a direct SMT encoding of the slicing problem.
    Our empirical evaluation of Brutus on existing and new benchmarks shows that we can find slices that
    are both small and informative.
\end{abstract}

\section{Introduction}

\subsection{Verification of Probabilistic Programs}

\paragraph{Probabilistic Programs.}
Probabilistic programs are programs that contain random assignments and may contain conditioning.
They have been used to reason about randomized algorithms or models of real-world systems that deal with uncertainty.
Recently, probabilistic programs have gained popularity in the context of machine learning as a way to formally reason about the behavior of AI models~\cite{DBLP:journals/jmlr/BinghamCJOPKSSH19,DBLP:conf/nips/TranHMSVR18,JSSv076i01}.

\paragraph{Deductive Verification.}
Typically, probabilistic programs are analyzed using advanced simulation techniques such as Markov Chain Monte Carlo~\cite{DBLP:journals/ml/AndrieuFDJ03}.
This provides only statistical guarantees and has certain limitations, e.g., possible non-termination of the simulation method for programs with an infinite expected runtime.
To obtain hard guarantees, formal verification of probabilistic programs has received quite some attention~\cite{DBLP:journals/pacmpl/MoosbruggerSBK22,DBLP:journals/pacmpl/MajumdarS25, DBLP:journals/pacmpl/ChatterjeeGMZ24, DBLP:conf/cav/ChakarovS13, DBLP:journals/jacm/KaminskiKMO18}.
This paper considers deductive verification of discrete probabilistic programs based on weakest precondition reasoning à la Dijkstra.
The program is associated with a specification, consisting of (quantitative) pre- and postconditions, with the verification task to establish that the program fulfills the specification.

\begin{wrapfigure}[8]{r}{.46\textwidth}
    \vspace*{-2.9em}
    \begin{spreadlines}{0ex}
        $
            \begin{aligned}
                \gutter{1}    & \texttt{let}~s0choice = \texttt{flip}~0.5~\texttt{in}          \\
                \gutter{2}    & \texttt{if}~s0choice~\texttt{then}~\texttt{true}~\texttt{else} \\
                \gutter{3}    & \texttt{let}~drop = \texttt{flip}~0.0005~\texttt{in}           \\
                \gutter{4}    & \quad \neg drop                                                \\
                \gutter{Goal} & \Pr(\mathtt{result} = \texttt{true}) \geq 99\%
            \end{aligned}
        $
    \end{spreadlines}
    \caption{Dice program $\procname$ modelling the packet delivery reliability problem.}
    \label{fig:heyvl-encoding-task}
\end{wrapfigure}

As an example, \Cref{fig:heyvl-encoding-task} shows a Dice~\cite{DBLP:journals/pacmpl/HoltzenBM20} program $\procname$ representing a network packet delivery problem, adapted from~\cite[Fig.~12]{DBLP:conf/pldi/GehrMTVWV18}.
The goal is to analyze the delivery reliability of a packet sent from host $s_0$ to $s_3$ in a network with diamond topology.
Host $s_0$ uses probabilistic load balancing to forward the packet to $s_1$ or $s_2$ with equal probability, modeled through a fair coin flip for $s0choice$ in $\procname$.
Host $s_1$ has a perfectly reliable channel to $s_3$, while the channel from $s_2$ to $s_3$ has a probability of $0.05\%$ to drop the message.
In $\procname$, a biased coin is flipped to determine whether the message gets dropped, if it is transmitted over $s_2$.
Finally, the program returns whether the transmission to the destination $s_3$ is successful.
The objective for verification: checking that the probability of a successful transmission is at least $99\%$.

\paragraph{Automation.}
Reasoning about probabilistic programs is challenging because of possible undesirable behaviors that only occur with low (or even zero) probabilities.
Classical program verification would falsely reject probability zero failures, while many analyses of probabilistic programs often allow a certain leeway, e.g.\ by allowing errors with probability zero.
In addition to verifying that a program has the correct output \emph{in expectation}, properties such as expected runtimes, expected energy consumption, or expected privacy leakage can be of interest too.
A shift to quantitative verification approaches can account for these properties.
Instead of using Boolean-valued conditions, the specification is expressed through quantitative expectations.

To formally reason about probabilistic programs, we use a quantitative intermediate verification language (IVL) named \HeyVL{}, the basis of the deductive verifier Caesar~\cite{DBLP:journals/pacmpl/SchroerBKKM23}.
An IVL allows to decouple the verification process from different possible input languages, and is a common approach in classical deductive program verification, as done by e.g.\ Dafny with Boogie~\cite{DBLP:conf/lpar/Leino10,leino2008boogie}, or Prusti with Viper~\cite{DBLP:conf/nfm/AstrauskasBFGMM22,DBLP:conf/vmcai/0001SS16}.
\HeyVL{} is expressive enough to encode a wide range of probabilistic programs and their properties.
It features quantitative verification statements that generalize classical ones, i.e. $\symAssume$ and $\symAssert$ statements, to represent quantitative properties.
In addition, an extendable catalogue of proof rules enables different ways to reason about loops.

\HeyVL{}'s semantics is based on weakest pre-expectation transformers, i.e.\ basically non-negative real (or $\infty$)-valued random variables~\cite{DBLP:series/mcs/McIverM05,DBLP:phd/dnb/Kaminski19}.
The intuition is that $\symAssume$ statements introduce proof obligations on the expected value of expressions provided in $\symAssert$ statements.
Verification queries about lower and upper bounds result in a verification condition, typically an inequality between expectations, that needs to be proven.

\begin{wrapfigure}[17]{r}{.35\textwidth}
    \vspace*{-2.4em}
    \begin{spreadlines}{0ex}
        $
            \begin{aligned}
                \gutter{1}  & \highlightCode{col2!20}{\stmtAssume{0.99}}                          \\
                \gutter{2}  & \stmtRasgn{s0choice}{\exprFlip{0.5}}                                \\
                \gutter{3}  & \symDemonic~\{                                                      \\
                \gutter{4}  & \quad \highlightCode{brBlue!20}{\stmtAssume{\embed{s0choice}}}      \\
                \gutter{5}  & \quad \stmtAsgn{delivered}{\exprTrue}                               \\
                \gutter{6}  & \}~\symElse~\{                                                      \\
                \gutter{7}  & \quad \highlightCode{brBlue!20}{\stmtAssume{\embed{\neg s0choice}}} \\
                \gutter{8}  & \quad \stmtRasgn{drop}{\exprFlip{0.0005}}                           \\
                \gutter{9}  & \quad \stmtAsgn{delivered}{\neg drop}                               \\
                \gutter{10} & \}                                                                  \\
                \gutter{11} & \highlightCode{col2!20}{\stmtAssert{\iverson{delivered}}}
            \end{aligned}
        $
    \end{spreadlines}
    \caption{\HeyVL{} encoding of the verification query for $\procname$ from~\Cref{fig:heyvl-encoding-task}.}
    \label{fig:heyvl-encoding-encoding}
\end{wrapfigure}

For verification with Caesar, the program is first encoded into \HeyVL.
The identified pre- and post-expectation are added with initial \symAssume{}, resp.\ final \symAssert{} statements.
The program is further annotated with information on the proof rules used to reason about, e.g., invariants for loops.

The encoding of the network delivery problem from~\Cref{fig:heyvl-encoding-task} is given in~\Cref{fig:heyvl-encoding-encoding}.
The auxiliary variable $delivered$ represents the output value of $\procname$.
The specification is encoded as a verification condition, highlighted in red.
It states that the probability of the final value of $delivered$ meets the lower bound of $99\%$.
Iverson brackets $\iverson{\cdots}$ are used to map a Boolean result to the real values $0$ or $1$.
The \texttt{if}-\texttt{then}-\texttt{else} statement is encoded using a demonic choice $\symDemonic$ and two assumptions, highlighted in blue.
Boolean values are embedded with $\embed{\cdots}$ to $0$ or $\infty$, with the effect that the $\symAssume$-statements effectively discard any proof obligations on the expected value of the branch if the condition is $\exprFalse$, and keeping it unaffected otherwise.

\subsection{Focus of this Paper}

Deductive program verification is an iterative process where the user writes and refines (1) the specification, (2) the proof and proof rules, or (3) even the program under consideration.
Understanding the behavior of probabilistic programs often requires intricate reasoning about the probabilities of different events and even simple programs can require mathematical reasoning about series, exponentials, or limits.
Ideally, the verification process is guided by feedback from the verifier to assist with this complexity.

While existing tools provide automated ways to confirm that a proof is correct, there is a lack of assistance during the verification phase.
For instance, Caesar translates the generated verification condition into an SMT-formula that is discharged to an SMT solver.
In case the query fails, the solver only provides limited information by a counterexample, which corresponds to an initial state.
In the presence of quantitative control flow, it is unclear how this information can result in insightful diagnostic feedback to the user.

This paper is focused on aiding users in the verification of probabilistic programs through appropriate diagnostic messages.
More specifically, we tackle the following three questions:
\begin{enumerate}
    \item \emph{Error Localization:} If a program's quantitative verification fails, how can we identify and localize errors to provide meaningful error messages?
    \item \emph{Verification Certificates:} If a program's verification succeeds, how can we reduce and simplify proofs to convince the user through understandable, small certificates?
    \item \emph{Hints:} If a program's verification succeeds, can we distill the program's core that makes the specification true?
\end{enumerate}

Our central contribution is to
\begin{enumerate*}[label={(\alph*)}]
    \item formally define appropriate concepts to capture these issues,
    \item develop algorithms to compute these diagnostics, and
    \item implement and experimentally validate these.
\end{enumerate*}
The development of our framework in a quantitative \emph{intermediate} verification language enables us to apply our techniques to a range of probabilistic programming languages, such as Dice~\cite{DBLP:journals/pacmpl/HoltzenBM20} (with loops~\cite{DBLP:conf/fscd/Torres-RuizP0Z24}) or \pGCL~(\cite{DBLP:series/mcs/McIverM05}).

\subsection{Approach}

\paragraph{Slicing.}
Our approach is based on program slicing~\cite{DBLP:journals/tse/Weiser84,DBLP:journals/sigsoft/XuQZWC05}.
Slicing refers to methods that, through removal of statements, extract a subprogram from a program while preserving certain properties of interest.
We introduce three kinds of slices: \emph{error-witnessing}, \emph{verification-witnessing}, and \emph{verification-preserving slices}.
Depending on the presence or absence of statements in the slice, appropriate diagnostics can be deduced and presented to the user.
Let us explain the different types of slices and how they are used to produce useful diagnostics through a series of examples.

\begin{wrapfigure}[6]{r}{.38\textwidth}
    \vspace*{-4em}
    \begin{spreadlines}{0ex}
        $
            \begin{aligned}
                \gutter{1}    & \stmtProb{\nicefrac{1}{2}}{\stmtAsgn{b0}{0}}{\stmtAsgn{b0}{1}}          \\
                \gutter{2}    & \stmtProb{\nicefrac{1}{2}}{\sliced{\stmtAsgn{b1}{0}}}{\stmtAsgn{b1}{1}} \\
                \gutter{3}    & \stmtAsgn{r}{b0 + 2 * b1}                                               \\
                \gutter{Goal} & \Pr(r \geq 2) \not\leq \nicefrac{1}{3}
            \end{aligned}
        $
    \end{spreadlines}
    \caption{Uniform 2-bit integer sampling using two coins.}
    \label{fig:intro-error-preserving}
\end{wrapfigure}

\paragraph{Error Localization.}
\Cref{fig:intro-error-preserving}, including the striked-out code, shows a program that samples a uniform two-bit integer $r$ by flipping two fair coins $b0$ and $b1$.
Our goal is to show an upper bound of $\nicefrac{1}{3}$ on the probability of $r \geq 2$.
Obviously, this bound is too tight.
Which parts of the program cause the verification to fail?
We observe that the two traces sampling $b0 = 0, b1 = 1$ and $b0 = 1, b1 = 1$ result in $r \geq 2$ and together have probability $\frac{1}{2} \cdot \frac{1}{2} + \frac{1}{2} \cdot \frac{1}{2} = \frac{1}{2} > \frac{1}{3}$.
Importantly, this counterexample makes no use of the assignment $\stmtAsgn{b1}{0}$.

Leaving out this assignment (strike-out) results in an \emph{error-witnessing slice}: any error in the slice is also present in the original program.
This allows to localize the verification failure to a subprogram by isolating the statements that introduce responsible proof obligations.

\paragraph{Verification Certificates.}
Reconsider the network packet delivery problem from \Cref{fig:heyvl-encoding-task,fig:heyvl-encoding-encoding}, for which we want to successfully verify that the probability to correctly receive a message is at least $99\%$.
While we can consider the combined distribution of the (two) paths to the target host via $s_1$ and $s_2$, notice that we can also check the probability bound independently for each path.
\begin{wrapfigure}[9]{r}{.4\textwidth}
    \vspace*{-0.2em}
    \begin{spreadlines}{0ex}
        $
            \begin{aligned}
                 & \texttt{fun}~send(s0choice:\colon~\texttt{bool})~\{                  \\
                 & \quad \texttt{if}~s0choice~\texttt{then}~\texttt{true}~\texttt{else} \\
                 & \quad \texttt{let}~drop = \texttt{flip}~0.0005~\texttt{in}           \\
                 & \quad \quad \neg drop                                                \\
                 & \}
            \end{aligned}
        $
    \end{spreadlines}
    \caption{A Dice program representing the verification-witnessing slice of \cref{fig:heyvl-encoding-encoding}.}
    \label{fig:heyvl-encoding-task-slice}
\end{wrapfigure}

If we have a lower bound for every path, the least of those bounds gives a lower bound on the combined probability, thus completing the proof.
This corresponds to replacing the probabilistic choice between $s_1$ and $s_2$ by a non-deterministic choice.
In the \HeyVL-encoding (\Cref{fig:heyvl-encoding-encoding}) this corresponds to the removal of the blue $\symAssume$-statements (and the probabilistic assignment of $s0choice$).

The resulting program is a \emph{verification-witnessing} slice: if the slice verifies, then the original program verifies as well.
These slices contain all the necessary information to conclude that the original program can be verified.
They can be employed to simplify proofs and as verification certificates.
For the present example, the slice can be represented by making $s0choice$ an input parameter~(\Cref{fig:heyvl-encoding-task-slice}), as Dice~\cite{DBLP:journals/pacmpl/HoltzenBM20} does not include nondeterminism.

\begin{wrapfigure}[10]{r}{.44\textwidth}
    \vspace*{-2.3em}
    \begin{spreadlines}{0ex}
        $
            \begin{aligned}
                \gutter{1}    & \stmtAsgn{c}{0}                                               \\
                \gutter{2}    & \stmtAsgn{i}{1}                                               \\
                \gutter{3}    & \sliced{\symWhile}~\symIf~i \leq n~\{                         \\
                \gutter{4}    & \quad \stmtProb{\nicefrac{1}{n}}{\stmtAsgn{c}{i}}{\stmtSkip}~ \\
                \gutter{5}    & \quad \stmtAsgn{i}{i + 1}                                     \\
                \gutter{6}    & \}                                                            \\
                \gutter{Goal} & \Pr(c = X) \leq \nicefrac{1}{n} \text{ for } X \in \Nats
            \end{aligned}
        $
    \end{spreadlines}
    \caption{\marknewtext{Faulty reservoir sampling.}}
    \label{fig:intro-verification-preserving}
\end{wrapfigure}

\paragraph{Hints.}
\Cref{fig:intro-verification-preserving} is a buggy version of \emph{Algorithm R}~\cite{DBLP:journals/toms/Vitter85} to uniformly sample an element $c$ from a stream of size $n$ via reservoir sampling.
In the correct algorithm, the probabilistic choice in \lineNumber{4} is weighted with $\nicefrac{1}{i}$ instead of $\nicefrac{1}{n}$.
Nevertheless, the (correct) upper bound of $\nicefrac{1}{n}$ for the probability of terminating with $c = X$ can be verified for all $X \in \Nats$.
In fact, the loop itself is not required to establish the specification.
It can be replaced by a single loop iteration.

The resulting program is an (amorphous) \emph{verification-preserving slice}: if the original program verifies, then the slice verifies as well.
In general, we can use such slices to give hints, such as suggesting a new program that still satisfies the specification but is easier to grasp, more efficient, or alerting the programmer of the fact that the specification does not adequately capture the intended behavior of the program.

\paragraph{Implementation.}
\begin{figure}[t]
    \centering
    \begin{tikzpicture}[node distance=0.8cm, label distance=20mm, auto, font=\footnotesize,
            process/.style={rectangle, text centered, align=center, draw=black, fill=brBlue!30},
            io/.style={trapezium, trapezium left angle=70, trapezium right angle=110, text centered, align=center, draw=black, fill=orange!30},
            decision/.style={diamond, aspect=3, text centered, align=center, draw=black, fill=bluishgreen!30},
            arrow/.style={thick,->,>=stealth}
        ]
        \node (inputProgram) [io] {Probabilistic Program $P$ in Dice, pGCL, \dots};
        \node (inputProperty) [io, right=of inputProgram] {Property/Specification $\varphi$};
        \node (input) [process, below=0.4cm of inputProgram, xshift=1.3cm] {HeyVL Program $S = \stmtAssume{\cdots} \symSemi \textrm{encode}(P) \symSemi \stmtAssert{\cdots}$};
        \node (firstVerifies) [decision, below=0.3cm of input] {Does $S$ verify?};

        \node (selectDiscordant) [process, below left=0.5cm and 0.2cm of firstVerifies] {Select assert-like stmts.};
        \node (sliceErrors) [process, below of=selectDiscordant] {Slice};
        \node (errorSlice) [io, below = 0.3cm of sliceErrors,align=center] {Error-witnessing slice};

        \node (selectConcordant) [process, right=0.6cm of selectDiscordant] {Select assume-like stmts.};
        \node (sliceWitness) [process, below of=selectConcordant] {Slice};
        \node (witnessSlice) [io, below = 0.3cm of sliceWitness,align=center] {Ver.-witnessing slice};

        \node (selectAny) [process, right=0.6cm of selectConcordant] {Select any stmts.};
        \node (slicePreserving) [process, below of=selectAny] {Slice};
        \node (preservingSlice) [io, below = 0.3cm of slicePreserving,align=center] {Ver.-preserving slice};

        \draw [arrow] (inputProgram) -- (input);
        \draw [arrow] (inputProperty) -- (input);
        \draw [arrow] (input) -- (firstVerifies);
        \draw [arrow] (firstVerifies) -- node[anchor=east,yshift=2pt] {no} (selectDiscordant);
        \draw [arrow] (selectDiscordant) -- (sliceErrors);
        \draw [arrow] (sliceErrors) -- (errorSlice);

        \draw [arrow] (firstVerifies) -- node[anchor=west, xshift=5pt, yshift=2pt] {yes} (selectConcordant);
        \draw [arrow] (selectConcordant) -- (sliceWitness);
        \draw [arrow] (sliceWitness) -- (witnessSlice);

        \draw [arrow] (firstVerifies) edge[bend left=15] node[anchor=west, xshift=12pt, yshift=2pt] {yes} (selectAny);
        \draw [arrow] (selectAny) -- (slicePreserving);
        \draw [arrow] (slicePreserving) -- (preservingSlice);
    \end{tikzpicture}
    \Description{Flowchart of the verification process showing the steps to obtain error-witnessing, verification-witnessing, and verification-preserving slices.}
    \caption{Flowchart of the verification process leading to the various kinds of slices.}
    \label{fig:slicing-flowchart}
\end{figure}
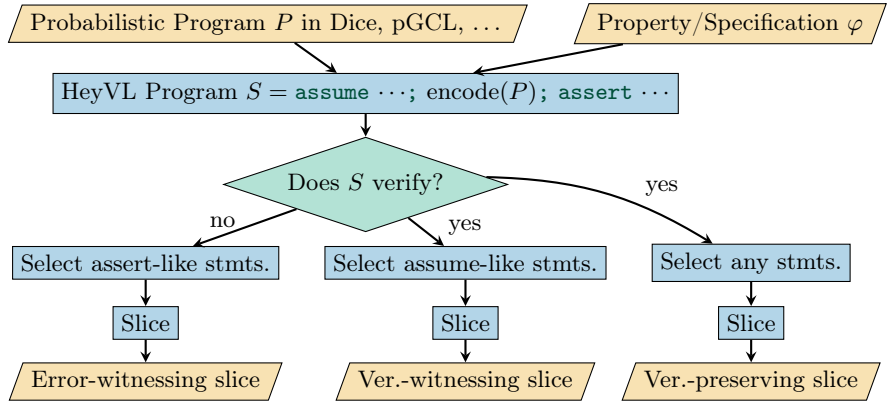

\Cref{fig:slicing-flowchart} shows a flowchart of how and when we obtain the different kinds of slices.
After encoding a probabilistic program and a specification into \HeyVL, the resulting program is verified and sliced.
If the program does not verify, we slice to obtain an error-witnessing slice.
We do this by selecting assert-like statements to slice.
If the program verifies, we can either aim for verification-witnessing or verification-preserving slices.
For verification-witnessing slices, we select assume-like statements and search for a slice that still verifies.
For verification-preserving slices, we can select any statements and search for a slice that still verifies.
These steps ensure that we obtain meaningful slices that do not vacuously satisfy the definitions, such as a verification-witnessing slice of a program that does not verify in the first place.

\subsection{Contributions and Outline}
In~\cref{sec:background-heyvl}, we explain the necessary background for the syntax and semantics of \HeyVL{}, which supports reasoning about lower and upper bounds of expected values of non-deterministic, probabilistic programs that include recursion and loops.
To summarize, the main contributions of this work are:
\begin{itemize}
    \item \Cref{sec:slicing-error-reporting}: Formal notion of \emph{error-witnessing slices} for error localization and simple sufficient conditions for soundness (removal of \emph{assert-like statements}).
    \item \Cref{sec:slicing-proof-simplification}: Formal notion of \emph{verification-witnessing slices} for verification certificates and simple sufficient conditions for their soundness (removal of \emph{assume-like statements}).
    \item \Cref{sec:slicing-verification-preserving}: Formal notion of \emph{verification-preserving slices} for hints and program optimization.
    \item \Cref{sec:slicing_on_encodings}: Theorems of correctness of error localization and hints based on slicing for specific proof rule encodings as case studies.
    \item \Cref{sec:implementation}: The first implementation of specification-based slicing for probabilistic programs which supports all of the above notions of slicing, integrated into the \emph{Caesar} verifier. All outputs of the slicing methods are integrated into its Visual Studio Code plugin and are visualized as error messages or warnings.
    \item \Cref{sec:implementation}: We evaluate the implementation on a set of representative benchmarks and show that we can quickly obtain user diagnostics for interactive feedback during the verification process.
\end{itemize}
Finally, \Cref{sec:related-work} surveys related work, putting existing notions into the context of our kinds of slices.
We refute claims in the literature that finding optimal slices requires both forward and backward reasoning and show that our approach can handle these cases.
Proofs of all results can be found in~\Cref{appendix:proofs}.

\section{HeyVL -- A Quantitative Intermediate Verification Language}\label{sec:background-heyvl}

In this section, we briefly explain the syntax and semantics of \HeyVL{}, an intermediate language to encode quantitative verification problems.
\HeyVL{} is based on \HeyLo{}, a syntax for expectation-based reasoning.
\HeyLo{} is used as \HeyVL{}'s assertion language, and is used to define \HeyVL{}'s formal semantics.
A more detailed exposition of both can be found in~\cite{DBLP:journals/pacmpl/SchroerBKKM23}.

\subsection{Expectation-Based Reasoning}

Reasoning about probabilistic programs generalizes Boolean to \emph{expectation-based} reasoning~\cite{DBLP:series/mcs/McIverM05}, i.e.\ about quantities such as expected values assigned to every program state.
In the following, we introduce the basic notions of logical reasoning about expectations.

\paragraph{Variables and Program States.}
Given a countably infinite set of typed variables $\Vars = \{x,y,\ldots\}$, we denote by $\typeof{x}{\typevar}$ that $x$ has type $\typevar$.
Typical types are Booleans $\Bools = \{\true,\false\}$, integers $\Ints$, or extended unsigned reals $\PosRealsInf = \PosReals \cup \{\infty\}$.
A \emph{program state} $\State$ maps every variable $\typeof{x}{\typevar}$ to a value $\State(x) \in \typevar$.

\paragraph{Expectations.}
Generalizing logical predicates, \emph{expectations} map program states to values in $\PosRealsInf$.
Expectations form a complete lattice $(\Expectations,\,\expleq)$ with $\expa \in \Expectations$ for all $\expa \colon \States \to \PosRealsInf$ and where $\expa \eexpleq \expb$ iff for all $\State\in\States$: $\expa(\State) \lleq \expb(\State)$.
\Cref{fig:heylo-semantics} shows the syntax for expectations $\expc$ and their semantics~\cite[Section 2]{DBLP:journals/pacmpl/SchroerBKKM23}.

\begin{table}[t]
	\Description{\HeyLo{} semantics.}
	\caption{Operators on expectations. $\inf$ and $\sup$ are taken on $\PosRealsInf$. $\expa\substBy{x}{v}$ denotes the expectation with the value for $x$ replaced by $v$, i.e. $\expa\substBy{x}{v} = \lam{\State}{\expa(\State\substBy{x}{v})}$, with $\sigma\substBy{x}{v}(y) = v$ if $x=y$ and $\sigma\substBy{x}{v}(y) = \sigma(y)$ else.}
	\label{fig:heylo-semantics}
	\renewcommand{\arraystretch}{1.5}%
	\newcommand{\hint}[1]{\footnotesize{}#1}%
	\renewcommand{\eval}[1]{#1}%
	\begin{minipage}{1\textwidth}
		\begin{center}
			\adjustbox{width=\textwidth}{%
				\begin{tabular}
					{@{\qquad}l@{\quad}l@{\quad}|@{\quad}l@{\quad}l@{\quad}}
					\toprule
					$\expc$                                & $\evalState{\expc}$                                                                & $\expc$                                & $\evalState{\expc}$                                                           \\
					\midrule
					$\aexpr$                               & $\evalState{a}\vphantom{\ifThenElse{\evalState{\bexpr} = \true}{\infty}{0}}$       & $\embed{\bexpr}$                       & $\ifThenElse{\evalState{\bexpr} = \true}{\infty}{0}$                          \\
					$\expa \sqcap \expb$                   & $\min \Set{\eval{\expa}(\sigma),~ \eval{\expb}(\sigma) }$                          & $\expa \sqcup \expb$                   & $\max \Set{ \eval{\expa}(\sigma),~ \eval{\expb}(\sigma) }$                    \\
					$\iquant{\typeof{x}{\typevar}}{\expa}$ & $\inf \Set{ \evalStateSubstBy{\expa}{x}{v} | v \in \typevar }$                     & $\squant{\typeof{x}{\typevar}}{\expa}$ & $\sup \Set{ \evalStateSubstBy{\expa}{x}{v} | v \in \typevar }$                \\
					$\expa \impl \expb$                    & $\ifThenElse{\evalState{\expa} \leq \evalState{\expb}}{\infty}{\evalState{\expb}}$ & $\expa \coimpl \expb$                  & $\ifThenElse{\evalState{\expa} \geq \evalState{\expb}}{0}{\evalState{\expb}}$ \\[1em]
					\bottomrule
				\end{tabular}
			}
		\end{center}
	\end{minipage}%
\end{table}

\paragraph{Basic expectations.}
We use $r$ to represent the expectation $\lam{\State}{r}$ for some $r \in \PosRealsInf$.
For arithmetic expressions $a$ of type $\PosRealsInf$, we simply write $a$ to mean $\lam{\State}{a(\State)}$, e.g. $2 \cdot x + 5$ represents the expectation $\lam{\State}{2 \cdot \State(x) + 5}$.
The \emph{embedding operator} $\embed{\cdot}$ embeds a Boolean expression $b$ into expectations: $\embed{b}(\State)$ maps to $\infty$ if $b$ evaluates to $\exprTrue$ in $\State$ and to $0$ otherwise.
The minimum and maximum between two expectations is denoted by the symbols $\sqcap$ and $\sqcup$, respectively.

\paragraph{Quantifiers.}
The \emph{infimum quantifier} $\Inf $ and the \emph{supremum quantifier} $\Sup$ are the quantitative analogues of the universal $\forall$ and the existential $\exists$ quantifier from predicate logic.
Intuitively, the $\Inf$ quantifier minimizes a quantity, just like the $\forall$ quantifier minimizes a predicate's truth value.
Further, we can embed $\forall$ as follows: $(\iquant{\typeof{x}{\typevar}}{\embed{\bexpr}})(\sigma) = \infty$ if and only if $\sigma \models \forall \quant{\typeof{x}{\typevar}}{\bexpr}$.
Dually, $\Sup$ maximizes a quantity, and we have $(\squant{\typeof{x}{\typevar}}{\embed{\bexpr}})(\sigma) = \infty$ if and only if $\sigma \models \exists \quant{\typeof{x}{\typevar}}{\bexpr}$.

\paragraph{Implications and Coimplications.}
The \emph{implication} $\impl$ generalizes the Boolean implication to expectations.
For a state $\State$, the implication $\expa \impl \expb$ evaluates to $\infty$ if $\expa(\State) \leq \expb(\State)$, and to $\expb(\State)$ otherwise.
Its dual \emph{coimplication} $\coimpl$ generalizes the converse nonimplication, defined for propositions $P$ and $Q$ as $\neg(P \leftarrow Q)$.
On expectations, $\expa \coimpl \expb$ evaluates to $0$ if $\expa(\State) \geq \expb(\State)$, and to $\expb(\State)$ otherwise.

The implications have two important applications.
The implication $\expa \impl \expb$ can be used to encode a comparison, as $\expa \impl \expb = \infty$ iff $\expa \expleq \expb$.
They can also simplify proof obligations: The inequality $\expa \expleq (\expb \impl \expc)$ can be transformed into the equivalent $(\expa \sqcap \expb) \expleq \expc$ by the \emph{adjointness property}.\footnote{With this implication, expectations form a \emph{Heyting algebra}. The name ``\HeyLo{}'' stands for a logic over Heyting algebras.}
This tells us that introducing assumptions can only make the proof obligations simpler.
The coimplication satisfies the dual properties and is used for upper bounds reasoning.

\paragraph{Validations.}
In the verification of probabilistic programs, one often needs to \say{cast} quantities $\expa$ into a quality.
For every state $\State$, the \emph{validation} $\expvalidate{\expa}$ evaluates to $\expvalidate{\expa}(\State) = \infty$ if and only if $\expa(\State) = \infty$.
Otherwise, $\expvalidate{\expa}(\State) = 0$.
Thus, it \say{pulls down} every value different from the maximal value $\infty$.
The \emph{covalidation} $\expcovalidate{\expa}$ is dual and \say{pulls up} every value that differs from the minimal value $0$.
Turning quantitative implications into qualitative comparisons is now $\expvalidate{\expa \impl \expb}$, and we have $\expvalidate{\expa \impl \expb}(\State) = \infty$ iff $\expa(\State) \leq \expb(\State)$.
The dual is $\expcovalidate{\expa \coimpl \expb}$, which evaluates to $0$ iff $\expa(\State) \geq \expb(\State)$.

\subsection{The Intermediate Verification Language HeyVL}

\HeyLo expressions denote quantitative properties about probabilistic programs.
The intermediate verification language \HeyVL{} encodes quantitative verification problems for probabilistic programs.
It extends standard (probabilistic) constructs such as assignments, probabilistic sampling, and sequencing with quantitative \emph{verification statements} that are used to transform and approximate the expected values.
With these verification statements, \HeyVL{} generalizes classical (Boolean) IVLs such as Boogie~\cite{leino2008boogie}.

\paragraph{Syntax of Statements.}
The syntax of \HeyVL{} statements $\stmt$ is given in the first and third columns in \cref{fig:heyvl-semantics}.
Here, $x \in \Vars$ is a variable of type $\typevar$, $a$ is an arithmetic expression, $\expa$ and $\expb$ are expectation expressions, and $\mu$ is a \emph{distribution expression} of type $\tau$ representing a finite-support probability distribution.
Let $\mu \eeq \pexp{p_1}{\termvar_1} \pexpand \ldots \pexpand \pexp{p_n}{\termvar_n}$ denote a probability distribution with probability expressions $p_1,\ldots,p_n$ for each value $\termvar_1,\ldots,\termvar_n$, respectively.
We use $\exprFlip{p}$ to denote a Bernoulli distribution with probability expression $p$, i.e. $\exprFlip{p} \eeq \pexp{p}{\true} \pexpand \pexp{(1-p)}{\false}$.
Let us briefly explain each statement.
The \emph{probabilistic assignment} $\stmtDeclInit{x}{\typevar}{\mu}$ assigns to the variable $x$ of type $\typevar$ a value sampled from the distribution $\mu$.
When $x$ is already declared, we just write $\stmtRasgn{x}{\mu}$ for probabilistic assignments and $\stmtAsgn{x}{v}$ for deterministic assignments (with the Dirac distribution $\mu = \delta_v = \pexp{1}{v}$).
The statement $\stmtTick{a}$ adds a reward of $a$ to the expected value of $\expa$.
It can be used to model time progression or a more general resource consumption of the program.
The statement $\stmtSeq{\stmtOne}{\stmtTwo}$ represents a sequence of programs $\stmtOne$ and $\stmtTwo$.
The statement $\stmtDemonic{\stmtOne}{\stmtTwo}$ is a demonic \emph{nondeterministic choice} between $\stmtOne$ and $\stmtTwo$.
Nondeterminism is resolved by minimizing the expected value.
The statement $\stmtAssert{\expb}$ quantitatively generalizes assertions from classical IVLs.
Similarly, $\stmtAssume{\expb}$ generalizes assumptions.
The $\stmtHavoc{x}$ statement forgets the current value of $x$ and introduces a minimizing nondeterministic branching for every possible value of $x$.
The $\stmtValidate$ statement introduces a quantitative validation.

A distinguishing feature of \HeyVL{} compared to IVLs for non-probabilistic programs is that its verification-related statements have dual versions.
There is a dual maximizing (aka: angelic) nondeterministic choice $\symAngelic$, as well as $\symUp\symAssert$, $\symUp\symAssume$, $\symUp\symHavoc$, and $\symUp\symValidate$ statements.
Whereas the non-$\symUp$ statements are used to reason about lower bounds (\say{the expected value is \emph{at least} some value}), the $\symUp$ statements allow reasoning about upper bounds (\say{the expected value is \emph{at most} some value}).

\begin{table}[t]
	\caption{Syntax and semantics of \HeyVL{} statements. Here $\mu = \pexp{p_1}{\termvar_1} \pexpand \ldots \pexpand \pexp{p_n}{\termvar_n}$ and $\expa\substBy{x}{\termvar_i}$ is the formula obtained from substituting every occurrence of $x$ in $\expa$ by $\termvar_i$ in a capture-avoiding manner.}
	\label{fig:heyvl-semantics}
	\begin{minipage}{1\textwidth}
		\begin{center}
			\adjustbox{width=\textwidth}{%
				\renewcommand*{\arraystretch}{1.5}%
				\begin{tabular}
					{l@{\quad}l@{\quad}|@{\quad}l@{\quad}l}
					\toprule
					$\stmt$                                            & $\vc{\stmt}(\expa)$                                            & $\stmt$                          & $\vc{\stmt}(\expa)$                                            \\
					\midrule
					\multirow{2}{*}{$\stmtDeclInit{x}{\typevar}{\mu}$} & $p_1 \cdot \expa\substBy{x}{\termvar_1}$                       & \stmtTick{\aexpr}                & $\expa + \aexpr$                                               \\
					                                                   & \quad$+ \ldots + p_n \cdot \expa\substBy{x}{\termvar_n}$       & \stmtSeq{\stmtOne}{\stmtTwo}     & $\vc{\stmtOne}\bigl(\vc{\stmtTwo}(\expa)\bigr)$                \\
					\multirow{2}{*}{\shortstack{$\symDemonic~\blockStart~\stmtOne~\blockEnd$ \\ $\symElse~\blockStart~\stmtTwo~\blockEnd$}}                   & \multirow{2}{*}{$\vc{\stmtOne}(\expa) \sqcap \vc{\stmtTwo}(\expa)$} & \multirow{2}{*}{\shortstack{$\symAngelic~\blockStart~\stmtOne~\blockEnd$ \\ $\symElse~\blockStart~\stmtTwo~\blockEnd$}}  & \multirow{2}{*}{$\vc{\stmtOne}(\expa) \sqcup \vc{\stmtTwo}(\expa)$} \\[1.2em]
					\Assert{\expb}                                     & $\expb \sqcap \expa$                                           & \coAssert{\expb}                 & $\expb \sqcup \expa$                                           \\
					\Assume{\expb}                                     & $\expb \rightarrow \expa$                                      & \coAssume{\expb}                 & $\expb \coimpl \expa$                                          \\
					\Havoc{x}                                          & $\iquant{x}{\expa}$                                            & \coHavoc{x}                      & $\squant{x}{\expa}$                                            \\
					\Validate                                          & $\expvalidate{\expa}$                                          & \coValidate                      & $\expcovalidate{\expa}$                                        \\
					\bottomrule
				\end{tabular}
			}
		\end{center}
	\end{minipage}
	\Description{\HeyVL{} semantics.}
\end{table}

\paragraph{Semantics of Statements.}
The semantics of \HeyVL{} is based on weakest pre-expectation-style reasoning~\cite{DBLP:series/mcs/McIverM05,DBLP:phd/dnb/Kaminski19}.
The \emph{verification pre-expectation transformer} $\vc{\stmt} \colon \Expectations \to \Expectations$ transforms an expectation $\expa$ by going backwards through the statement $\stmt$.
We define $\symVc$ by induction on $\stmt$ in \cref{fig:heyvl-semantics}.
Given a state $\State$, the result $\vc{\stmt}(\expa)(\State)$ can be intuitively understood as the expected value of $\expa$ on termination of $\stmt$ when starting in $\State$.

Let $\expa$ be an expectation.
We call it the \emph{post(-expectation)}.
The expectation $\vc{\stmt}(\expa)$ is the verification pre-expectation with respect to post $\expa$.
For a random assignment $\stmt = \stmtDeclInit{x}{\typevar}{\mu}$, $\vc{\stmt}(\expa)$ is the weighted sum $p_1 \cdot \expa\substBy{x}{\termvar_1} + \ldots + p_n \cdot \expa\substBy{x}{\termvar_n}$, corresponding to the expected value of $\expa$ after assigning $x$ sampled from $\mu$.
A statement $\stmtTick{a}$ increases the post $\expa$ by $a$, therefore $\vc{\stmtTick{a}}(\expa) = \expa + a$.
For $\stmt = \stmtDemonic{\stmtOne}{\stmtTwo}$, $\vc{\stmt}(\expa)$ is the minimum of $\vc{\stmtOne}(\expa)$ and $\vc{\stmtTwo}(\expb)$.
For $\symAngelic$, we take the maximum.

The semantics of $\stmt = \stmtAssert{\expb}$ generalizes the classical definition with a Boolean conjunction.
In our quantitative setting, $\vc{\stmt}(\expa) = \expb \sqcap \expa$.
This corresponds to the \emph{least truth value} between the post $\expa$ and the assertion $\expb$.
The statement $\stmt = \stmtAssume{\expb}$ introduces a quantitative implication: $\vc{\stmt}(\expa) = \expb \impl \expa$.
Following the definition of $\impl$, introducing an $\symAssume$ statement lowers the threshold at which an expected value is considered absolutely true.
In a state $\State$ where $\expb(\State) \leq \expa(\State)$, we have $\vc{\stmtAssume{\expb}}(\expa)(\State) = \infty$.
Otherwise, the semantics is a no-op.
The semantics of the statement $\stmtHavoc{x}$ generalizes the classical $\forall$ semantics with an infimum: $\vc{\stmtHavoc{x}}(\expa) = \iquant{x}{\expa}$.
The validation statement $\stmtValidate$ has semantics that translate the post into either $0$ or $\infty$ using the $\expvalidate{\cdot}$ operator, such that $\vc{\stmtValidate}(\expa)(\State) = \infty$ iff $\expa(\State) = \infty$.

The dual statements have a dual semantics.
The $\symUp\symAssert$ statement generates the maximum of the post and the assertion instead of a minimum.
For $\symUp\symAssume$, we use the coimplication $\coimpl$ so that $\vc{\coAssume{\expb}}(\expa)(\State) = 0$ if and only if $\expb(\State) \geq \expa(\State)$ holds.
Otherwise, the semantics is a no-op.
For the $\symUp\symHavoc$ statement, we get the supremum.
Finally, $\coValidate$ has semantics such that $\vc{\coValidate}(\expa)(\State) = 0$ if and only if $\expa(\State) = 0$, and $\infty$ otherwise.

\paragraph{Specifications.}\label{par:heyvl-specifications}
We define a shorthand notation resembling Hoare triples for quantitative specifications.
Because we reason both about lower and upper bounds, we define two kinds of triples:
\(
\models \lowerTriple{\expa}{\stmt}{\expb} \text{ if and only if } \expa \expleq \vc{\stmt}(\expb),
\)
and dually
\(
\models \upperTriple{\expa}{\stmt}{\expb} \text{ if and only if } \expa \expgeq \vc{\stmt}(\expb).
\)
We say a statement $\stmt$ \emph{verifies} if $\models \lowerTriple{\infty}{\stmt}{\infty}$.
Dually, we say $\stmt$ \emph{co-verifies} if $\models \upperTriple{0}{\stmt}{0}$.

\paragraph{Conservativity.}
\HeyVL{} is a conservative extension of Boolean IVLs.
This means that we can embed Boolean reasoning in \HeyVL{} using the embedding operator $\embed{\cdot}$ and obtain the same results as in the Boolean setting.
For example, the classical encoding of the conditional choice statement $\stmtIf{\bexpr}{\stmtOne}{\stmtTwo}$ can be generalized: $\stmtDemonic{\stmtAssume{\embed{\bexpr}}\symSemi \stmtOne}{\stmtAssume{\embed{\neg \bexpr}}\symSemi \stmtTwo}$.
We will use the former as shorthand for the encoding in the rest of this paper.

\paragraph{Monotonicity.}
The $\symVc$ transformer is monotonic, which means that if the post $\expa$ is replaced by $\expaP$ such that $\expa \expleq \expaP$, then we retain the same inequality on the $\symVc$ semantics.
This property is crucial for some of our soundness theorems.

\begin{theorem}[Monotonicity of $\symVc$, {{\cite[Thm.~3.2]{DBLP:journals/pacmpl/SchroerBKKM23}}}]
	\label{thm:monotonicity} For all \HeyVL{} statements $\stmt$ and $\expa,\expaP \in \Expectations$, $\expa \expleq \expaP$ implies $\vc{\stmt}(\expa) \expleq \vc{\stmt}(\expaP)$.
\end{theorem}

\section{Error Localization, Certificates, and Hints}\label{sec:slicing-heyvl}

We will now introduce the different kinds of slices of \HeyVL{} programs that enable the identification of error locations when verification fails (\cref{sec:slicing-error-reporting}), the construction of witnesses for the validity of a verifying program (\cref{sec:slicing-proof-simplification}), and the tailoring of probabilistic programs to fulfil a given specification (\cref{sec:slicing-verification-preserving}).
For simplicity, we will focus this exposition on \HeyVL{} programs encoding lower bound verification tasks.
Similar results can be stated for the dual case of upper bound verification.
\marknewtext{Further, the slice notions defined for \HeyVL programs can be lifted to any (probabilistic) language encodable into \HeyVL.}

\subsection{Error Localization}\label{sec:slicing-error-reporting}

In nonprobabilistic systems, safety properties can be refuted by a single path in the system.
Providing counterexamples for specifications for probabilistic systems is inherently more complex.
In the area of probabilistic model checking, the extraction of counterexamples for Markov models has been extensively studied, e.g.\ by~\cite{DBLP:phd/dnb/Jansen15,DBLP:conf/sfm/AbrahamBDJKW14,DBLP:conf/fmcad/JantschHFB20}.
There, counterexamples are sub-Markov chains which violate the specification at hand.
For probabilistic programs, a natural definition arises from the negation of the verification task, i.e.\ a counterexample is a state where the specification is violated.

\begin{definition}
	A \emph{counterexample} of \HeyVL{} program $\stmt$ w.r.t.\ $\expa, \expb \in \Expectations$ is a state $\State \in \States$ with $\expa(\State) \not\leq \vc{\stmt}(\expb)(\State)$.
	We write $\State \not\models \lowerTriple{\expa}{\stmt}{\expb}$.
\end{definition}

Inspired by~\cite{DBLP:journals/corr/abs-1305-5055,DBLP:conf/atva/DehnertJWAK14}, we consider \emph{sub-programs} as the counterpart of sub-Markov chains at the program level.
Here, we formalize this as a slicing problem.
Slicing aims to find a subprogram $\slice$ that preserves certain semantic properties, in this case, that $\slice$ fails to satisfy the specification \emph{just like the original program $\stmt$} did.
$\slice$ witnesses the error.

\begin{definition}[Error-witnessing slice]\label{def:slice:error_witnessing}
	A subprogram $\slice$ of program $\stmt$ is an \emph{error-witnessing slice w.r.t.\ $\expa, \expb \in \Expectations$} if
	\begin{enumerate}
		\item\label{def:slice:error_witnessing:counterexample_exists}

		      $\not\models \lowerTriple{X}{\slice}{Y}$, and
		\item\label{def:slice:error_witnessing:counterexamples_transfer_to_original}
		      $\forall \sigma'.~ (\sigma' \not\models \lowerTriple{\expa}{\slice}{\expb} \implies \sigma' \not\models \lowerTriple{\expa}{\stmt}{\expb})$.

	\end{enumerate}
\end{definition}

The first requirement ensures that the slice $\slice$ has at least one counterexample state.
The second requirement states that all counterexamples of $\slice$ are counterexamples of $\stmt$ too.

\begin{figure}
	\begin{minipage}[t]{0.3\textwidth}
		\centering
		\begin{spreadlines}{0ex}
			\begin{align*}
				\gutter{1} & \stmtAssume{0.75}               \\
				\gutter{2} & \stmtRasgn{x}{\exprFlip{0.5}}   \\
				\gutter{3} & \stmtAssert{\iverson{x \geq 0}} \\
				\gutter{4} & \stmtAssert{\iverson{x = 1}}
			\end{align*}
		\end{spreadlines}
		\vspace*{-0.52em} %
		\caption{\HeyVL{} program which does not verify. Removing \lineNumber{3} results in an error-witnessing slice. However, removing \lineNumber{4} would result in a verifying program.}
		\label{fig:example:introductory_counterexample_slice}
	\end{minipage}
	\hfill
	\begin{minipage}[t]{0.3\textwidth}
		\centering
		\begin{spreadlines}{0ex}
			\begin{align*}
				\gutter{1} & \stmtAssume{1}                \\
				\gutter{2} & \stmtAsgn{x}{\exprTrue}       \\
				\gutter{3} & \stmtRasgn{x}{\exprFlip{0.5}} \\
				\gutter{4} & \stmtAssert{\iverson{x}}
			\end{align*}
		\end{spreadlines}
		\vspace*{-0.52em} %
		\caption{\HeyVL{} program which does not verify. Removing \lineNumber{2} results in an error-witnessing slice, even though the assignment is not reductive.}
		\label{fig:example:non_reductive_counterexample_slice}
	\end{minipage}
	\hfill
	\begin{minipage}[t]{0.3\textwidth}
		\centering
		\begin{spreadlines}{0ex}
			\begin{align*}
				\gutter{1} & \stmtAssume{\nicefrac{1}{2}}     \\
				\gutter{2} & \stmtAssume{\iverson{x \geq 1}}  \\
				\gutter{3} & \stmtAssume{\iverson{x \leq 10}} \\
				\gutter{4} & \stmtRasgn{y}{\exprFlip{0.5}}    \\
				\gutter{5} & \stmtAsgn{y}{y + x}              \\
				\gutter{6} & \stmtAssert{\iverson{y \geq 2}}
			\end{align*}
		\end{spreadlines}
		\caption{\HeyVL{} program which verifies. Removing \lineNumber{3} results in a verification-witnessing slice.}
		\label{fig:example:introductory_verification_slice}
	\end{minipage}
\end{figure}

\begin{example}

	Consider the \HeyVL{} program $\stmt$ in~\cref{fig:example:introductory_counterexample_slice} which flips a fair coin and tries to establish that the probability of the result being non-negative and exactly $1$ is at least $75\%$.
	We have $\vc{\stmt}(\infty) = 0.75 \rightarrow (0.5 \cdot (1 \sqcap 0 \sqcap \infty) + 0.5 \cdot (1 \sqcap 1 \sqcap \infty)) = 0.75 \rightarrow 0.5 = 0.5 \neq \infty$.
	Therefore, the program fails to verify.
	In fact, the set of counterexamples consists of all states.
	The subprogram $\stmtOneP$ obtained by erasing \lineNumber{3}~($\stmtAssert{\iverson{x \geq 0}}$) is an error-witnessing slice of $\stmt$ w.r.t.\ $(\infty, \infty)$, because any state is a valid counterexample: $\vc{\stmtOneP}(\infty) = 0.75 \rightarrow 0.5 = 0.5 \neq \infty$.
	The subprogram $\stmtTwoP$ obtained by erasing \lineNumber{4}~($\stmtAssert{\iverson{x = 1}}$) verifies, as \(\vc{\stmtTwoP}(\infty) = 0.75 \rightarrow 1 = \infty\).
	Hence, it is not an error-witnessing slice.
	We conclude that the requirement that $x = 1$ holds at the end of program $\stmt$ with at least $75\%$ probability is violated.
\end{example}
The example showed that removing assertions can help to turn a non-verifying program into a verifying one.
Conversely, adding $\symAssert$-statements to a program that fails to verify will not yield fewer counterexamples.
Statements that \emph{always} strengthen proof obligations are called \emph{reductive}.\footnote{This terminology is based on nomenclature for closure operator definitions~\cite[Section~11.7]{cousot_aa}.}
\begin{definition}\label{def:reductive}
	A statement $\stmt$ is \emph{reductive} if $\vc{\stmt}(\expa) \expleq \expa$ for all $\expa \in \Expectations$.
\end{definition}
Equivalently, one may use \emph{upper} bound Hoare triples, saying that $\models \upperTriple{\expa}{\stmt}{\expa}$ holds for all $\expa \in \Expectations$.
Various verification statements in \HeyVL{} are reductive, most notably \symAssert-statements.
We will also call such statements \emph{\symAssert-like}.

\begin{restatable}{lemma}{lemmaReductive}
	\label{lemma:reductive}
	Statements $\symAssert$, $\symUp\symAssume$, $\symHavoc$, and $\symValidate$ are reductive.\footnote{These are all atomic \HeyVL{} statements that are reductive in their general form. Specific instances of other statements, e.g.\ $\Assume{\infty}$, are also reductive. }
\end{restatable}

From \cref{lemma:reductive} and the monotonicity of $\symVc$, it follows that erasing reductive statements in a program $\stmtP$ can never decrease the $\symVc$ of program $\stmt$ in which $\stmtP$ is embedded.
This suggests:

\begin{restatable}{theorem}{theoremReductiveSlice}\label{theorem:reductive-slice}
	Let $\stmtP$ be a subprogram obtained from program $\stmt$ by only erasing reductive statements.
	If $\not\models \lowerTriple{\expa}{\stmtP}{\expb}$, then $\stmtP$ is an error-witnessing slice of $\stmt$ w.r.t.\ $(\expa, \expb)$ for all $\expa, \expb \in \Expectations$.
	In particular, we have $\vc{\stmt}(\expa) \expleq \vc{\stmtP}(\expa)$ for all $\expa \in \Expectations$.
\end{restatable}

The next example shows that error-witnessing slices can also result from erasing non-reductive statements.

\begin{example}

	Consider the \HeyVL{} program $\stmt$ in~\cref{fig:example:non_reductive_counterexample_slice}.
	It does not verify, as the final assertion is required to always hold but is only fulfilled in $50\%$ of the runs.
	Let $\stmtP$ be obtained by removing \lineNumber{2} ($\stmtAsgn{x}{\exprTrue}$).
	Even though $\stmtAsgn{x}{\exprTrue}$ is not reductive, e.g.\ $\vc{\stmtAsgn{x}{\exprTrue}}(\embed{x}) = \embed{\exprTrue} = \infty \not\expleq \embed{x}$, it holds that $\vc{\stmt}(\infty) \expleq \vc{\stmtP}(\infty)$.
	As $\stmtP$ also fails to verify, $\stmtP$ is an error-witnessing slice of $\stmt$ w.r.t.\ $(\infty, \infty)$.
\end{example}

Removing \symAssert-like statements allows pinpointing the error location: after removal of the irrelevant \symAssert-like statements, only the ones necessary for the error remain.
We do want to trace down a path leading to the error; e.g., if a certain branch of an \symIf-\symElse-statement results in an error, this can help to identify the program's problem.
Observe that a demonic choice $\stmtDemonic{\stmtOne}{\stmtTwo}$ exhibits a property similar to the reductivity of \symAssert-like statements: replacing $\stmtOne$ (or $\stmtTwo$) by the demonic choice between $\stmtOne$ and $\stmtTwo$ introduces more constraints.

\begin{restatable}{lemma}{lemmaSlicingDemonic}\label{lemma:slicing-demonic}
	For all $\stmtOne$ and $\stmtTwo$, the demonic choice $\stmt = \stmtDemonic{\stmtOne}{\stmtTwo}$ fulfills: $\vc{\stmt}(\expa) \expleq \vc{\stmtOne}(\expa)$ and $\vc{\stmt}(\expa) \expleq \vc{\stmtTwo}(\expa)$, for all $\expa \in \Expectations$.
\end{restatable}

\begin{example}
	Consider~\cref{fig:intro-error-preserving}, which
		{\makeatletter
			\let\par\@@par
			\par\parshape0
			\everypar{}
			\begin{wrapfigure}[17]{r}{.5\textwidth}
    \vspace*{-4em}
    \begin{spreadlines}{0ex}
        \begin{align*}
            \gutter{1} & \coAssume{\nicefrac{1}{3}}                                                       \\
            \gutter{2} & \stmtRasgn{c}{\exprFlip{0.5}}                                                    \\
            \gutter{3} & \stmtAngelicStart{\coAssume{\embed{\neg c}} \symSemi \stmtAsgn{b0}{0}}\}         \\
            \gutter{4} & \stmtElseStart{\coAssume{\embed{c}} \symSemi \stmtAsgn{b0}{1}} \}                \\
            \gutter{5} & \stmtRasgn{c}{\exprFlip{0.5}}                                                    \\
            \gutter{6} & \sliced{\stmtAngelicStart~\coAssume{\embed{\neg c}} \symSemi \stmtAsgn{b1}{0}\}} \\
            \gutter{7} & \sliced{\symElse}~\{\coAssume{\embed{c}} \symSemi \stmtAsgn{b1}{1}\}             \\
            \gutter{8} & \stmtAsgn{r}{b0 + 2 * b1}                                                        \\
            \gutter{9} & \coAssert{\iverson{r \geq 2}}
        \end{align*}
    \end{spreadlines}
    \caption{Error-witnessing slice of the uniform 2-bit integer sampling using two coins.}
    \label{fig:example:error-trace-slice}
\end{wrapfigure}

			\noindent
			uses upper bounds.
			We desugar the conditional choices into explicit angelic choices.\footnote{As is standard in \HeyVL{}, the Boolean embeddings inside the $\symUp$-verification statements are negated (c.f.~\cite[Example~2.3]{DBLP:journals/pacmpl/SchroerBKKM23}).}
			The result and error-witnessing slice are shown in \cref{fig:example:error-trace-slice}, applying the dual version of \cref{lemma:slicing-demonic}.
			In the error-witnessing slice, the angelic choice in \lineNumber{5} is replaced by the second branch.
			Thus, the resulting counterexample slice consists of both branches in \lineNumber{3} and the second branch of \lineNumber{5}.
			These correspond to the two traces that are needed to show the violation of the probability bound.
			\par}%
\end{example}

\subsection{Certificates}\label{sec:slicing-proof-simplification}

Dual to finding a subprogram that witnesses an error, we consider the problem of finding a subprogram that witnesses the validity of a verifying program.
Instead of identifying the source of an error, we are now interested in identifying the parts of the program that are essential for the verification to succeed.

Again, we phrase this as a slicing problem: We want to find a subprogram $\slice$ whose verification guarantees that the program $\stmt$ verifies.
In that sense, $\slice$ contains the necessary parts to \emph{witness the verification} of $\stmt$.
Any parts that are non-essential to the verification are erased.

\begin{definition}[Verification-witnessing slice]\label{def:slice:verification_witnessing}
	A subprogram $\slice$ of a program $\stmt$ is called a \emph{verification-witnessing slice w.r.t.\ $\expa, \expb \in \Expectations$} if
	\[
		\models \lowerTriple{\expa}{\slice}{\expb} \qimplies \models \lowerTriple{\expa}{\stmt}{\expb}~.
	\]
\end{definition}
Subprograms that fail to verify trivially become verification-witnessing slices.
As these slices are not very helpful, we exclude them in the following.

\begin{example}
	Consider the verifying \HeyVL{} program $\stmt$ in~\cref{fig:example:introductory_verification_slice} and the subprogram $\stmtP$ obtained by removing \lineNumber{3} ($\stmtAssume{\iverson{x \leq 10}}$).
	Due to the assumption $\iverson{x \geq 1}$, the final assertion holds in $\geq 50\%$ of the states.
	Thus, $\stmtP$ is a verification-witnessing slice of $\stmt$, showing that $\iverson{x \leq 10}$ is unnecessary for $\infty \expleq \vc{\stmt}(\infty)$.
	\end{example}

As in the previous example, proof assumptions are commonly expressed through \symAssume-statements.
Indeed, introducing additional assumptions can only simplify the proof.
Formally, these statements have a dual property to reductiveness, called \emph{extensiveness}:
\begin{definition}\label{def:extensive}
	A statement $\stmt$ is \emph{extensive} if $\vc{\stmt}(\expa) \expgeq \expa$ for all $\expa \in \Expectations$.
\end{definition}

In terms of \emph{lower} bound Hoare triples, one may require $\models \lowerTriple{\expa}{\stmt}{\expa}$ for all $\expa \in \Expectations$.
$\symAssume$ statements are extensive, and we will call extensive verification statements \emph{$\symAssume$-like}.
\begin{restatable}{lemma}{lemmaExtensive}\label{lemma:extensive}
	$\symAssume$, $\symUp\symAssert$, $\symUp\symHavoc$, $\symUp\symValidate$, $\symTick$ are extensive.
\end{restatable}

We obtain the dual of \cref{theorem:reductive-slice}:

\begin{restatable}{theorem}{theoremExtensiveSlice}\label{theorem:extensive-slice}
	Let $\stmtP$ be a subprogram obtained from $\stmt$ by only removing extensive statements.
	Then, $\stmtP$ is a verification-witnessing slice of $\stmt$ w.r.t.\ $(\expa, \expb)$ for all $\expa,\expb \in \Expectations$.
	In particular, we have $\vc{\stmt}(\expa) \expgeq \vc{\stmtP}(\expa)$ for all $\expa \in \Expectations$.
\end{restatable}

\subsection{Hints}\label{sec:slicing-verification-preserving}

The verification-witnessing slices are based on the idea to obtain a slice which relies on fewer assumptions.
However, often the complications in a proof do not lie in the too specific assumptions, but rather in the sequence of steps taken to reach the conclusion.
These steps usually correspond to intermediate statements like assignments.

This view motivates a third type of slicing, aimed at shrinking a program while still being able to verify it.
Intuitively, this can be seen as tailoring the program to the specification.
The result is a slice that only contains statements required to \emph{preserve verification}.

\begin{definition}[Verification-preserving slice]\label{def:slice:verification_preserving}
	A subprogram $\slice$ of a program $\stmt$ is called a \emph{verification-preserving slice w.r.t.\ $\expa, \expb \in \Expectations$} if
	\[
		\models \lowerTriple{\expa}{\stmt}{\expb} \qimplies \models \lowerTriple{\expa}{\slice}{\expb}~.
	\]
\end{definition}

Such slices are known as \emph{specification-based} slices.
They were introduced by~\cite{DBLP:journals/fac/BarrosCHP12} and subsequently extended to probabilistic programs by~\cite{DBLP:journals/scp/NavarroO22}.
The contraposition of being a verification-preserving slices reads: $\not\models \lowerTriple{\expa}{\slice}{\expb}$ implies $\not\models \lowerTriple{\expa}{\stmt}{\expb}$.
This resembles our definition of error-witnessing slice as expressed by the following lemma.

\begin{restatable}{lemma}{lemmaErrorWitnessing}\label{lemma:error-witnessing}
	For subprogram $\slice$ of program $\stmt$:
	If $\slice$ is an error-witnessing slice of $\stmt$ w.r.t.\ $(\expa, \expb)$, then $\slice$ is a verification-preserving slice of $\stmt$ w.r.t.\ $(\expa, \expb)$.
\end{restatable}

\marknewtext{In particular, it follows from~\cref{theorem:reductive-slice} that slicing reductive statements results in a verification-preserving slice. But as shown in the following example, we can also slice sampling statements to obtain a verification-preserving slice.}

\begin{example}
	\begin{figure}[t]
		\begin{subfigure}[c]{.45\textwidth}
			\begin{spreadlines}{0ex}
				\begin{align*}
					\gutter{1}    & \sliced{\stmtProb{0.5}{\stmtAsgn{b0}{0}}{\stmtAsgn{b0}{1}}} \\
					              &                                                             \\
					\gutter{2}    & \stmtProb{0.5}{\sliced{\stmtAsgn{b1}{0}}}{\stmtAsgn{b1}{1}} \\
					              &                                                             \\
					\gutter{3}    & \stmtAsgn{r}{b0 + 2 * b1}                                   \\
					\gutter{Goal} & \Pr(r \geq 2) \geq \nicefrac{1}{2}
				\end{align*}
			\end{spreadlines}
			\caption{\pGCL{} program.}
			\label{fig:example:specification_based_slice:pgcl}
		\end{subfigure}
		\hfill
		\begin{subfigure}[c]{.45\textwidth}
			\begin{spreadlines}{0ex}
				\begin{align*}
					\gutter{1}    & \sliced{\stmtRasgn{c}{\exprFlip{0.5}}}                           \\
					\gutter{2}    & \stmtIf{c}{\sliced{\stmtAsgn{b0}{0}}}{\sliced{\stmtAsgn{b0}{1}}} \\
					\gutter{3}    & \stmtRasgn{c}{\exprFlip{0.5}}                                    \\
					\gutter{4}    & \stmtIf{c}{\sliced{\stmtAsgn{b1}{0}}}{\stmtAsgn{b1}{1}}          \\
					\gutter{5}    & \stmtAsgn{r}{b0 + 2 * b1}                                        \\
					\gutter{Goal} & \Pr(r \geq 2) \geq \nicefrac{1}{2}
				\end{align*}
			\end{spreadlines}
			\caption{\HeyVL{} encoding.}
			\label{fig:example:specification_based_slice:heyvl}
		\end{subfigure}
		\caption{Uniform 2-bit integer sampling using two fair coin flips and its verification-preserving slice w.r.t.\ $(\nicefrac{1}{2}, \iverson{r \geq 2})$.}
		\label{fig:example:specification_based_slice}
	\end{figure}

	\Cref{fig:example:specification_based_slice} shows a program in the \pGCL{} language and its \HeyVL{} encoding $\stmt$, with corresponding slices.
	The original program in~\cref{fig:example:specification_based_slice:pgcl} samples a two-bit integer by performing two coin flips, inspired by~\cite[Example~4.3]{DBLP:journals/scp/NavarroO22}.
	It should generate an integer that is at least 2 with probability at least $50\%$.
	Therefore, we fix pre $\expa = \frac{1}{2}$ and post $\expb = \iverson{r \geq 2}$.
	A verification-preserving slice of $\stmt$ can now be tailored to $(\expa, \expb)$, e.g.\ by removing any sampling of the least significant bit $b0$, yielding the slice in~\cref{fig:example:specification_based_slice:heyvl}.
	Notice that $\stmt$ has no error-witnessing slices w.r.t.\ $(\expa, \expb)$, as such a slice would contradict $\models \lowerTriple{\expa}{\stmt}{\expb}$.
\end{example}

\section{Slicing on Proof Rule Encodings}\label{sec:slicing_on_encodings}

\newcommand{\varInvariant}{\ensuremath{\colheylo{I}}}
\newcommand{\varBody}{\ensuremath{\color{heyvlColor}{C}}}

High-level language constructs are either encoded or reasoned about using proof rules.
For example, loops are typically analyzed using proof rules that provide sufficient conditions for lower or upper bounds of the meaning of the loop.
In this section, we consider the \HeyVL{} encodings to embed specifications, as well as \emph{specification statements}~\cite{DBLP:journals/toplas/Morgan88} that encode ``placeholders'' and can be used to encode procedure calls.

We will apply slicing to the \HeyVL{} encodings of high-level program language features to obtain diagnostic information.
We demonstrate how a \HeyVL{} slice can be interpreted to obtain information about the original program, by applying slicing to different proof rules encoded in \HeyVL{}.
To indicate which part of the proof rule has been violated when verification fails, our implementation includes support for annotations on statements.\footnote{The user can write e.g. $\stmtAnnotate{\texttt{error\_msg}}{msg}{\stmtAssert{\expa}}$ to add a custom error message to an assertion, or write $\stmtAnnotate{\texttt{success\_msg}}{msg}{\stmtAssume{\expa}}$ to add a custom hint to an assumption, which are based on an erroring or verifying slice, respectively.}
We explain the diagnostics for the verification of specifications, specification statements, and the Park induction rule to reason about loops.
nostics.

\subsection{Verifying Specifications}
\label{sec:verifying-specifications}

A standard code verification approach~\cite{Mueller19} is to convert a triple of pre, program, and post into one IVL program.
This program is then checked against a fixed pre and post pair (classically $\exprTrue$ and $\exprTrue$) and is said to \emph{verify} if this check succeeds.
A similar approach can be taken with quantitative pre and post in \HeyVL{}~\cite{DBLP:journals/pacmpl/SchroerBKKM23}.
Let $\stmt$ be a \HeyVL{} program and $\expa, \expb \in \Expectations$. Then:
\begin{align*}
	 & \models \lowerTriple{\expa}{\stmt}{\expb} \qiff \Assume{\expa}\symSemi \stmt\symSemi \Assert{\expb} \text{ verifies.} %
\end{align*}
A dual encoding works for upper bounds using the dual statements $\symUp\symAssume$ and $\symUp\symAssert$.
We can now slice the specification itself by slicing the encoding w.r.t. pre $\infty$ and post $\infty$, using that $\symAssert$ is reductive and $\symAssume$ is extensive:

\begin{restatable}{theorem}{thmVerificationEncodingComponents}\label{thm:verification-encoding-components}
	Let $\stmt = \stmtAssume{\expa}\symSemi \stmtP\symSemi \stmtAssert{\expb}$.
	\begin{enumerate}
		\item \textbf{Removing the $\symAssert$ statement shows unverifiability with any post.} \\ If $\stmtAssume{\expa}\symSemi \stmtP$ is an error-witnessing slice of $\stmt$, then $\not\models \lowerTriple{\expa}{\stmtP}{\expc}$ for $\expc \in \Expectations$.
		\item \textbf{Removing the $\symAssume$ statement shows verifiability with any pre.} \\ If $\stmtP\symSemi \stmtAssert{\expb}$ is a verifying verification-witnessing slice of $\stmt$, then $\models \lowerTriple{\expc}{\stmtP}{\expb}$ for $\expc \in \Expectations$.
	\end{enumerate}
\end{restatable}

Both results indicate possible issues within the program or specification.
An always failing verification hints at contradicting requirements, while an always successful verification might indicate that the program is ineffective in establishing the post, which is possibly too weak.

\subsection{Specification Statements}\label{sec:specification-statements}

\emph{Specification statements}~\cite{DBLP:journals/toplas/Morgan88} are useful during incremental development.
They are used as placeholders in a program, and explicitly state the requirements and guarantees required from the parts that are omitted (program \say{holes}).
Additionally, one can specify a collection of variables $\vec{v}$ which is allowed to change in the \say{hole}.
All other variables are assumed to remain unchanged.

The specification statement $\stmtSpec{\vec{v}}{\expa}{\expb}$ is a placeholder for all statements which have pre of at least $\expa(\State)$ when executed with post $\expb(\State)$ in state $\State$, i.e., it satisfies $\models \lowerTriple{\expa}{(\stmtSpec{\vec{v}}{\expa}{\expb})}{\expb}$.
The specification statement $\stmtSpec{\vec{v}}{\expa}{\expb}$ is encoded in \HeyVL{} as~\cite[Section 3.5]{DBLP:journals/pacmpl/SchroerBKKM23}:
\[
	\stmtSpec{\vec{v}}{\expa}{\expb} \eeq \stmtAssert{\expa}\symSemi \stmtHavoc{\vec{v}}\symSemi \stmtValidate\symSemi \stmtAssume{\expb}~.
\]
The encoding closely mirrors the intuition of the classical specification statement: we ensure (\emph{assert}) that the pre $\expa$ holds.
Then, we abstract from the modifying variables (with the $\symHavoc$ statement) and assume the post $\expb$ holds.
In \HeyVL{}, we are not limited to predicates and can reason about expectations.
The only difference is that an additional $\stmtValidate$ statement is necessary for correctness.

\begin{restatable}{theorem}{thmSpecStatementMessages}\label{theorem:spec-statement-messages}
	Let $C = \stmtSpec{\vec{v}}{\expa}{\expb}$ and its \HeyVL{} encoding be $\stmtC$.

	\begin{enumerate}
		\item \textbf{Pre does not entail specification's pre.} If $\stmtP$ is a minimal error-witnessing slice of $\stmtC$ w.r.t.\ $(\expaP, \expbP)$ which contains $\Assert{\expa}$, then $\expaP(\State') \not\leq \expa(\State')$ for some $\State'$.
		\item \textbf{Specification statement can be removed.} If $\stmtP$ is a verification-witnessing slice of $\stmtC$ w.r.t.\ $(\expaP, \expbP)$ which verifies and does not contain $\Assume{\expb}$, then $\expaP \expleq \expbP$ holds.
	\end{enumerate}
\end{restatable}

\subsection{Park Induction for Loop Reasoning}\label{sec:park-induction}

\newcommand{\park}[3]{Park({#1},{#2},{#3})}

\emph{Park induction} is an effective proof rule in probabilistic program verification~\cite{parkFixpointInductionProofs1969,DBLP:phd/dnb/Kaminski19}.
It can be seen as a generalization of the classical Hoare logic proof rule for loops: an \emph{invariant} is \emph{inductive}, if whenever it holds at the start of the loop, it also holds at its end.
For weakest liberal pre semantics $(\symWlp$)\footnote{For an expectation $\expa \in \Expectations$ where $\expa \expleq 1$, the \emph{weakest liberal pre-expectation} $\wlp{\stmt}(\expa)(\State)$ is the expected value of $\expa$ after executing $\stmt$ when starting in state $\State \in \States$, assigning $1$ to those executions which do not terminate. Dually, the \emph{weakest pre-expectation semantics} $\wp{\stmt}(\expa)(\State)$ assigns $0$ to diverging executions, and is defined for all expectations $\expa \in \Expectations$. A detailed exposition is given in \cite{DBLP:phd/dnb/Kaminski19}.}, the quantitative version of Park induction for a loop $\stmtWhile{b}{\varBody}$ and invariant $\varInvariant$ for post $\expb$ reads:
\begin{align*}
	\underbrace{
		\varInvariant \expleq \left( \embed{b} \impl \wlp{\varBody}(\varInvariant) \right)
		\sqcap
		\left(\embed{\neg b} \impl \expb\right)
	}_{\varInvariant\text{ is an inductive invariant}}
	\quad \text{implies}
	\quad
	\underbrace{
		\varInvariant \expleq \wlp{\stmtWhile{b}{\varBody}}(\expb)
	}_{
		\varInvariant\text{ underapproximates the loop's $\symWlp$}
	}
	.
\end{align*}
Thus, it can be used to establish lower bounds on the weakest liberal pre-expectation of a loop.
Dually, for weakest pre-expectation semantics it enables establishing upper bounds.
Again, we focus on lower bounds.
\Cref{fig:heyvl-park-induction} shows the encoding $\park{b}{\varBody}{\varInvariant}$ of the proof rule in \HeyVL{}.
Its correctness was shown in~\cite{DBLP:journals/pacmpl/SchroerBKKM23}.
We use slicing to justify each statement in this \HeyVL{} encoding.

\begin{restatable}{theorem}{theoremParkErrorMessages}\label{theorem:park-error-messages}
	Let $\stmtP$ be a slice of $\stmt = \park{b}{\varBody}{\varInvariant}$ as in~\cref{lemma:park-error-components} and be minimal such that $\State \not\models \lowerTriple{\expa}{\stmtP}{\expb}$ for some state $\State$.
	\begin{enumerate}
		\item \textbf{Pre does not entail invariant.} If $\stmtP$ includes \stmtNo{1}, then $\expa(\State') \not\leq \varInvariant(\State')$ for some state $\State'$.
		\item \textbf{Invariant not inductive.} If $\stmtP$ includes the assertion \stmtNo{2}, then there is a state $\State'$ where the loop guard $b$ is true and $\State' \not\models \lowerTriple{\varInvariant}{\varBody}{\varInvariant}$.
		\item \textbf{Invariant does not entail post.} If $\stmtP$ includes neither assertion \stmtNo{1} nor \stmtNo{2}, and $\varBody$ contains neither $\symAssert$ nor $\symUp\symAssume$ statements, then there is a state $\State'$ for which the loop guard $b$ does not hold and $\varInvariant(\State') \not\leq \expb(\State')$.
	\end{enumerate}
\end{restatable}

\begin{figure}[t]
	\begin{minipage}[t]{0.45\textwidth}
		\centering
		\begin{spreadlines}{0ex}
			\begin{align*}
				                                          &                                                           \\
				{\scriptsize\stmtNo{1}} \quad \gutter{1}  & \stmtAssert{\varInvariant}\symSemi                        \\
				\gutter{2}                                & \stmtHavoc{variables}\symSemi                             \\
				\gutter{3}                                & \stmtValidate\symSemi                                     \\
				{\scriptsize\stmtNo{I}} \quad \gutter{4}  & \stmtAssume{\varInvariant}\symSemi                        \\
				\gutter{5}                                & \stmtIfStart{b}                                           \\
				\gutter{6}                                & \quad \varBody\symSemi                                    \\
				{\scriptsize\stmtNo{2}} \quad \gutter{7}  & \quad \stmtAssert{\varInvariant} \symSemi                 \\
				{\scriptsize\stmtNo{II}} \quad \gutter{8} & \quad \stmtAssume{\embed{\exprFalse}}                     \\
				\gutter{9}                                & \blockEnd~\stmtElseStart~\blockEnd \quad \intersem{\expb} \\
				                                          &
			\end{align*}
		\end{spreadlines}
		\vspace*{-0.7em} %
		\Description{Encoding of Park induction.}
		\caption{Encoding $\park{b}{\varBody}{\varInvariant}$ of Park induction to underapproximate $\wlp{\stmtWhile{b}{\varBody}}(\expb)$.}
		\label{fig:heyvl-park-induction}
	\end{minipage}
	\hfill
	\begin{minipage}[t]{0.45\textwidth}
		\centering
		\begin{spreadlines}{0ex}
			\begin{align*}
				\gutter{1}  & \stmtAsgn{x}{init\_x} \symSemi                          \\
				\gutter{2}  & \stmtAsgn{cont}{\exprTrue} \symSemi                     \\
				\gutter{3}  & \stmtAnnotate{invariant}{\varInvariant}{}               \\
				\gutter{4}  & \headerWhile{cont}~\blockStart                          \\
				\gutter{5}  & \quad \stmtRasgn{prob\_choice}{\exprFlip{0.5}} \symSemi \\
				\gutter{6}  & \quad \stmtIfStart{prob\_choice}                        \\
				\gutter{7}  & \quad \quad \stmtAsgn{cont}{\exprFalse}                 \\
				\gutter{8}  & \quad \blockEnd~\stmtElseStart                          \\
				\gutter{9}  & \quad \quad \stmtAsgn{x}{x + 1}                         \\
				\gutter{10} & \quad \blockEnd                                         \\
				\gutter{11} & \blockEnd
			\end{align*}
		\end{spreadlines}
		\caption{Geometric loop $\stmt$ with different loop invariants $\varInvariant$ to verify $\models \upperTriple{init\_x + 1}{\stmt}{x}$.}
		\label{fig:geometric_loop:invariant_diagnostics}
	\end{minipage}
\end{figure}

These three error messages closely correspond to error messages in qualitative verification.
Dually, assumptions can be analyzed to find out how they influence the final proof obligation.
Here, we obtain diagnostics from \emph{removed} assumptions, as opposed to the \emph{remaining} assertions in the above theorem.

\begin{restatable}{theorem}{theoremParkVerificationMessages}\label{theorem:park-verification-messages}
	Let $\stmt = \park{b}{\varBody}{\varInvariant}$ as in \Cref{fig:heyvl-park-induction} and $\stmtP$ be obtained by removing any subset of the assumptions \stmtNo{I} or \stmtNo{II}, and let $\stmtP$ verify, i.e.\ $\models \lowerTriple{\expa}{\stmtP}{\expb}$.
	Then:
	\begin{enumerate}
		\item\label{item:invariant-not-necessary} \textbf{Assuming the invariant is not necessary.} If $\stmtP$ does not include \stmtNo{I}, then $\stmt$ verifies with invariant $\varInvariant = \infty$, i.e. $\models \lowerTriple{\expa}{\park{b}{\varBody}{\infty}}{\expb}$.
		\item\label{item:while-could-be-if} \textbf{While loop could be an if statement.} If $\stmtP$ does not include the assumption \stmtNo{II}, then $\models \lowerTriple{\expa}{\stmtIf{b}{\varBody}{}}{\expb}$ holds.
	\end{enumerate}
\end{restatable}

In case (\labelcref{item:invariant-not-necessary}), the loop verification can be generalized to not require the invariant so that the loop can be used in a larger context.
In case (\labelcref{item:while-could-be-if}), the specification might be too weak to require a loop, indicating a possible bug in the specification or the loop.
\marknewtext{An example is the buggy reservoir sampling algorithm from \Cref{fig:intro-verification-preserving}, annotated with the invariant $\nicefrac{1}{n} \cdot \iverson{i \leq X \land X \leq n} + \iverson{c = X}$.}

\begin{example}\label{example:park-induction:all-diagnostics}
	Consider the \HeyVL{} program $\stmt$ in~\cref{fig:geometric_loop:invariant_diagnostics}.
	It models a geometric loop counting the number of coin flips until the first success in $x$.
	To upper bound the expected value of $x$ at the end by the initial value of $x$ increased by one, we use pre $init\_x + 1$ and post $x$, and annotate the loop using the Park induction proof rule with invariant $\varInvariant$.
	The choice of $\varInvariant$ gives different diagnostic messages from~\cref{theorem:park-error-messages,theorem:park-verification-messages}, dualized for the upper bound scenario.
	With $\varInvariant = x + 1$ verification fails with 'Invariant not inductive.'
	With $\varInvariant = \exprIte{cont}{x + 1}{x}$ verification succeeds with 'While loop could be an if statement,' i.e., performing a single flip establishes the upper bound, as this increases $x$ by at most one.
\end{example}

\section{Implementation and Evaluation}\label{sec:implementation}

We realized our three types of slices in the quantitative deductive verifier \emph{Caesar}~\cite{DBLP:journals/pacmpl/SchroerBKKM23}.
The slicer is called \emph{Brutus}.\footnote{Caesar with slicing support is available online at \ifAnon{\emph{(link removed for double-blind review)}}{\url{https://www.caesarverifier.org}}.}
The results of \emph{Brutus} are used in \emph{Language Server Protocol}~\cite{lsp_protocol} support of Caesar and are shown as errors or suggestions in the Visual Studio Code editor.

In this section, we will briefly explain how we implemented the search for suitable slices, \marknewtext{compare the approaches on a set of selected benchmarks, and explore how helpful the resulting user diagnostics are during the verification process.}
More extensive information can be found in~\Cref{sec:detailed-benchmarks}.

At a high level, Brutus operates in three steps.
First, it selects suitable statements that can be sliced.
\marknewtext{Following \cref{fig:slicing-flowchart}~on~\cpageref{fig:slicing-flowchart}, if the original program does not verify, Brutus selects reductive statements.
	If the original program verifies, Brutus selects extensive statements to obtain verification-witnessing slices.
	The user can manually mark statements for slicing, which can be used to obtain verification-preserving slices.}
The second step is a \HeyVL{} program transformation that instruments the program with Boolean input variables to enable or disable the selected statements.
Finally, Brutus uses an SMT solver to search for suitable or optimal slices.\footnote{In general, the SMT queries to check whether a \HeyVL{} program verifies are not necessarily decidable. It embeds \HeyLo{}, which contains a language expressive enough to describe the termination probability of arbitrary programs (\cite{DBLP:journals/pacmpl/BatzKKM21}), and termination with probability 1 is undecidable (c.f. \cite{DBLP:conf/mfcs/KaminskiK15}).}
Let us explain the latter two steps in more detail.

\paragraph{Program Transformation for Slicing.}

We first transform the input \HeyVL{} program $\stmt$ with a set of sub-statements $\stmtN{1},\ldots,\stmtN{n}$ that we try to slice.
Intuitively, a Boolean input variable $\varEnabled{\stmtI}$ is inserted for each potentially sliceable statement $\stmtI$ and we replace $\stmtI$ by $\stmtIf{\varEnabled{\stmtI}}{\stmtI}{}$.
To avoid an exponential blow-up of the $\symVc$ size, we use a transformation which ensures that for each statement $\stmtI$, its transformation $\stmtITransformed$ contains the post $\expa$ in its $\vc{\stmtITransformed}(\expa)$ exactly as often as in $\vc{\stmtI}(\expa)$.
For example, $\stmtAsgn{x}{a}$ and $\stmtAssert{\expb}$ are transformed to $\stmtAsgn{x}{\ite{\varEnabled{\stmt}}{a}{x}}$ and $\stmtAssert{\ite{\varEnabled{\stmt}}{\expb}{\infty}}$, respectively.
The full list of transformations is given in~\Cref{fig:heyvl-transformations-extended} in~\Cref{appendix:proofs}.

The generated verification conditions of the \HeyVL{} encoding are similar to the ones in Boogie~\cite{DBLP:journals/scp/LeinoMS05}.
However, our encoding is done on the program level instead of by defining a modified program semantics.
Furthermore, our encoding is not limited to error messages, but is also used in the same way when we slice for verification.

\paragraph{Searching for Erroring Slices.}
After the program transformation, Brutus finds slices that fail verification with a counterexample.
We find error-witnessing slices by selecting reductive statements in a program that does not verify.
The task of finding (minimal) slices involves considering (all) combinations of enabled and disabled statements and checking whether the modified \HeyVL{} program verifies.
We solve this task using SMT solving.
To find an erroring slice, we check:
\begin{align}
	\exists \varEnabled{\stmtN{1}},\ldots,\varEnabled{\stmtN{n}}.~ \left( \exists \State \in \States.~ \vc{\stmt}(\infty)(\State) \neq \infty\right)~. \label{eq:smt-query-error}
\end{align}
As this query does not contain a quantifier alternation, it is well-suited for SMT solving.
For our strategy \sliceMethod{first}, we simply ask the SMT solver for a model.

However, the result may not be minimal.
Trivially, assigning $\exprTrue$ to all $\varEnabled{\stmt}$ variables also yields a slice.
Our implementation uses the Z3 SMT solver~\cite{z3solver} which can sometimes return models where some variables are marked as irrelevant for the counterexample.
We found that with irrelevant variables set to $\exprFalse$, the counterexamples for our benchmarks are almost minimal slices.

To find minimal error-preserving slices, we have the strategy \sliceMethod{opt} which runs a modified binary search in which we repeatedly ask the SMT solver for solutions with at most $n$ enabled statements.
Due to incompleteness of the SMT solver, not all queries have a yes/no answer.
Our algorithm marks these cases unacceptable and checks the other values.
We observed that the search needs a lot fewer steps when irrelevant variables information is used.

\paragraph{Searching for Verifying Slices.}

For both verification-witnessing and verification-preserving slices, we search for slices that verify.
The SMT query is:
\begin{align}
	\exists \varEnabled{\stmtOne},\ldots,\varEnabled{\stmtN{n}} \left( \forall \State \in \States.~ \vc{\stmt}(\infty)(\State) = \infty \right)~. \label{eq:smt-query-verify}
\end{align}
In contrast to \cref{eq:smt-query-error}, this query contains a quantifier alternation.
We implemented four strategies to solve this query.
The \sliceMethod{exists-forall} strategy simply encodes the query directly as a
quantified formula and discharges it to the SMT solver.
The \sliceMethod{core} strategy uses \emph{unsatisfiable cores} as provided by Z3.
In our setting, unsatisfiable cores represent a subset of the $\varEnabled{\stmt}$ variables that are set to $\exprTrue$ such that there is no counterexample (i.e.\ the program verifies).
Our third strategy (\sliceMethod{mus}) is based on enumerating minimal unsatisfiable subsets (MUS) of assertions~\cite{DBLP:conf/cpaior/LiffitonM13}.
Each such MUS corresponds to a set of enabled statements that are necessary for the program to verify.
The resulting slice is minimal, \marknewtext{i.e.\ there is no subprogram of the slice which is also a valid slice}.
The fourth strategy (\sliceMethod{sus}) finds the globally smallest unsatisfiable subset (SUS) of the query by enumerating all MUS and thus finds a verifying slice \marknewtext{with the least number of statements}.

\paragraph{Evaluation and Setup.}\label{sec:evaluation}
Our goal is to provide useful diagnostics to a user during the verification process.
From a theoretical perspective, we have argued in \Cref{sec:slicing-heyvl,sec:slicing_on_encodings} how to obtain useful diagnostics with formal guarantees, by finding (minimal) slices.
To compare the different slicing methods we seek to answer the following evaluation questions (EQ):
\begin{enumerate}[label={EQ \arabic*},ref={EQ~\arabic*}, font=\bfseries, leftmargin=*]
	\item\label{evaluation:question:is_interactive} Is the time needed for slicing acceptable for an interactive setting, i.e. $\leq 1$ second?
	\item\label{evaluation:question:how_good_is_the_result} What is the ratio between what is actually sliced and what
	      is sliced in a minimal slice?
\end{enumerate}
We collected benchmarks for both error-witnessing and verifying slices based on the literature and extended them with new examples.
Our evaluation focuses on benchmarks verifiable with automated deductive probabilistic program verification, as we focus on their diagnostics reporting.
Testing our algorithms implemented inside classical verifiers like Boogie~\cite{leino2008boogie} is out of scope of this paper.
In absence of a benchmark suite for slicing of probabilistic programs, our benchmarks capture a variety of language features.
We have 11 purely Boolean programs and 21 probabilistic programs.
They have explicit nondeterminism (1 program), loops (24 programs), recursion (1 program), conditioning (1 program), and continuous sampling (3 programs).
The experiments were conducted on a 2021 Apple MacBook with an M1 Pro chip and 16 GB of RAM, and a timeout of 30 s.
The detailed results from our benchmarks are included in \cref{sec:detailed-benchmarks}.
In the following, we summarize the results of our evaluation separately for error-witnessing and verifying slices.

\paragraph{Evaluation for Error-Witnessing Slices.}
We considered 16 example programs (\cref{tbl:erroring-results}~in~\Cref{sec:detailed-benchmarks}).
The benchmark set includes a broader set of examples aimed to demonstrate that error localization through slicing effectively identifies the problematic statement in verification tasks, including \Cref{fig:intro-error-preserving} and examples producing every diagnostic for Park induction (\Cref{fig:geometric_loop_template,tbl:geometric_loop_template_instances}~in~\Cref{sec:detailed-benchmarks}).
From the literature on error localization, we include a (non-probabilistic) simple list access error from~\cite{DBLP:journals/scp/LeinoMS05}.
We also include 7 modified probabilistic examples from~\cite{DBLP:journals/pacmpl/SchroerBKKM23} where we added errors, failing verification of expected runtimes of loops.
One example involves reasoning about conditioning, utilizing the decomposition of conditioning into $\symWp$ and $\symWlp$ as described by~\cite{DBLP:journals/entcs/0001KKOGM15}\footnote{For an expectation $\expa \in \Expectations$, the conditional expected value $\cwp{\stmt}(\expa)$ is given by $\cwp{\stmt}(\expa) = \nicefrac{\wp{\stmt}(\expa)}{\wlp{\stmt}(1)}$ (when $\wlp{\stmt}(1) \neq 0$). Proving an upper bound $\cwp{\stmt}(\expa) \expleq \nicefrac{\expb}{\expc}$ is split into two tasks, $\wp{\stmt}(\expa) \expleq \expb$ and $\expc \expleq \wlp{\stmt}(1)$. An $\symObserve$ statement is encoded into $\symAssert$ statements. Details are provided in \cite[Section 4.1]{DBLP:journals/pacmpl/SchroerBKKM23}}.
We did the same for 3 programs modeling sampling from continuous distributions~\cite{continuous25}.
The benchmarks range from 7 to 57 lines of code, and our theory narrows this to at most 5 sliceable statements, 3 on average.

Search for error slices is fast with both \sliceMethod{first} and \sliceMethod{opt}, usually with times below 20 milliseconds (\labelcref{evaluation:question:is_interactive}).
For 10 (out of 16) benchmarks, the first found slice is already optimal (\labelcref{evaluation:question:how_good_is_the_result}).
In general, the time to find a guaranteed optimal slice is not significantly higher ($\leq$ +20 ms) than using the first slice.
As minimal slices are required to accurately report errors (cf.\ e.g.\ \Cref{theorem:park-error-messages}), this makes our implementation \sliceMethod{opt} viable to provide user diagnostics with formal guarantees in an interactive environment.

\marknewtext{
	One limitation of our current implementation is that it requires the SMT solver being able to produce a consistent model (returning \enquote{SAT}), which is not always possible due to incompleteness.
	We found that Boogie intentionally uses potentially inconsistent models (from an \enquote{unknown} result of the SMT solver) to generate error diagnostics\footnote{\ifAnon{\emph{(link removed for double-blind review)}}{\url{https://github.com/boogie-org/boogie/issues/1008}}}.
	As these models do not carry any guarantees, we did not follow this approach.
}

\begin{figure}[t]
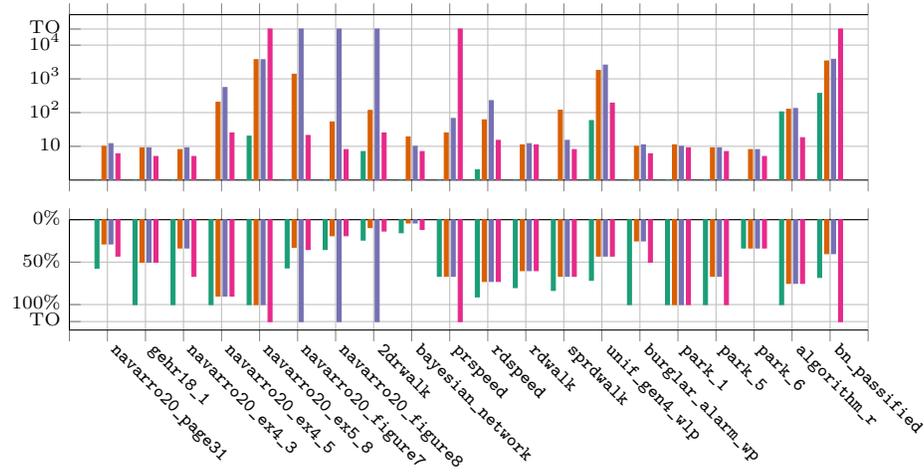

	\centering
	\newcommand{\legendcols}{2}
	\newcommand{\legendstyle}{at={(0.45,1.05)},anchor=south}
	\relativeslicesizeruntimeplot{sections/benchmarks/result-verifying.csv}%
	{core,
		mus,
		sus,
		exists_forall}%
	{{\sliceMethod{core}},
		{\sliceMethod{mus}},
		{\sliceMethod{sus}},
		{\sliceMethod{exists-forall}}}%
	{relative removal}%
	{time in milliseconds}%
	{states}%
	{30}%
	\caption{Dual bar chart showing run-time (positive y-axis, log-scaled) and percentage of removed sliceable statements (negative y-axis, linear scale). TO indicates timeout (30s). We show only 20 of our 30 verifying slices for readability. The bars are grouped by benchmark. Methods: \sliceMethod{core} (\textcolor{barplotcolora}{$\blacksquare$}), \sliceMethod{mus} (\textcolor{barplotcolorb}{$\blacksquare$}), \sliceMethod{sus} (\textcolor{barplotcolorc}{$\blacksquare$}), \sliceMethod{exists-forall} (\textcolor{barplotcolord}{$\blacksquare$}).}
	\label{fig:evaluation:verifying_slices}

\end{figure}

\paragraph{Evaluation for Verifying Slices.}
For verifying slices, we have \marknewtext{30} examples (\cref{tbl:verifying-results}~in~\Cref{sec:detailed-benchmarks}).
\marknewtext{Of these, 25} are for verification-preserving slicing and \marknewtext{five are} for verification-witnessing slicing.
Eight are nonprobabilistic examples from~\cite{DBLP:conf/sefm/BarrosCHP10}, the others are probabilistic examples based on the literature~\cite{DBLP:journals/toms/Vitter85,DBLP:conf/pldi/GehrMTVWV18,DBLP:journals/pacmpl/SchroerBKKM23,DBLP:journals/scp/NavarroO22}.
We include all applicable examples from the above sources.
The number of potentially sliceable statements ranges from 3 to 75.

Regarding \labelcref{evaluation:question:is_interactive}, the results \marknewtext{in \Cref{fig:evaluation:verifying_slices}} show that the \sliceMethod{core} strategy is very fast, with \marknewtext{instant results ($\leq$ 5 ms) in most cases and never exceeding 105 ms} for any benchmark.
The runtimes of \sliceMethod{exists-forall} have an average of 16 ms and low variance ($96 \%$ in $\leq$ 25 ms), while \sliceMethod{mus} and \sliceMethod{sus} are slower and exhibit a wide variance in the runtime.
The slowest slice search is the \sliceMethod{sus} strategy, \marknewtext{reaching the time-out of 30 seconds on some benchmarks.}
The search for a slice of a while loop \marknewtext{whose successful termination probability is verified (\cite[Example 5.8]{DBLP:journals/scp/NavarroO22})} uses the SMT theory of \emph{uninterpreted functions} with axioms to specify an exponential function and takes \marknewtext{3.7} seconds.
On average, \sliceMethod{sus} is $20\%$ slower than \sliceMethod{mus}.
We conclude that only the \sliceMethod{core} and \sliceMethod{exists-forall} strategies are usable in interactive scenarios.

For \labelcref{evaluation:question:how_good_is_the_result}, we empirically found that the \sliceMethod{core} slices are often far from optimal, as the unsatisfiable cores are not minimal.
On \marknewtext{13 of the 30} benchmarks, \sliceMethod{core} does not slice any statements.
In general, the other strategies are able to find (much) smaller slices.
For the benchmarks from the literature, our tool is able to find optimal slices.
Notably, on some examples from~\cite{DBLP:journals/scp/NavarroO22,DBLP:journals/fac/BarrosCHP12} we find smaller slices than their methods.
Even though \sliceMethod{mus} only does a local minimization, the resulting slices match the globally optimal slices obtained by \sliceMethod{sus} in size.
\sliceMethod{exists-forall} does not explicitly minimize the slice, yet also finds results as good as the optimal \sliceMethod{sus} strategy in 12 cases.
The figure shows that any minimization often improves the results compared to the simple \sliceMethod{core} strategy.

The results indicate that using the \sliceMethod{core} strategy has almost no performance impact on Caesar.
As minimal slices are not required for diagnostics from verification-witnessing slices (c.f. \Cref{sec:slicing_on_encodings}), it is planned to enable \sliceMethod{core} by default.
For optimal slices, the \sliceMethod{mus} or \sliceMethod{sus} strategies are needed, but require a lot of extra time.
While the \sliceMethod{exists-forall} strategy is often competitive with the \sliceMethod{mus} and \sliceMethod{sus} strategies, it is limited to reasoning about verification tasks without uninterpreted functions.
When applicable, the \sliceMethod{exists-forall} strategy is a good alternative to \sliceMethod{core} to greatly improve the results.

\section{Related Work}\label{sec:related-work}

\paragraph{Syntactic Slicing of Probabilistic Programs.}
A recent overview of slicing techniques for probabilistic programs is~\cite{Olmedo2025,slicing_olmedo_26}.
The techniques are either \emph{syntactic} or \emph{semantic}.
There are two main lines of research for syntactic slicing of probabilistic programs.
\cite{DBLP:conf/pldi/HurNRS14} incorporate stochastic (in-)dependencies inspired by d-separation criteria from Bayesian networks.
\cite{DBLP:conf/fossacs/AmtoftB16,DBLP:journals/toplas/AmtoftB20} generalize similar ideas to probabilistic control flow graphs.

\paragraph{Semantic Slicing of Probabilistic Programs.}
Lifting the concepts from classical programs by~\cite{DBLP:journals/fac/BarrosCHP12}, the only prior work for semantic slicing on probabilistic programs is~\cite{DBLP:journals/scp/NavarroO22}.
They consider specification-based slices, i.e.\ verification-preserving slices in our terminology.
To obtain minimal slices, they describe an algorithm that checks between every pair of statements $(\stmtN{i}, \stmtN{j})$ in a sequence $\stmtN{1}\symSemi \cdots \symSemi \stmtN{n}$ whether one can jump from one to the other and still establish the specification with pre $\expa$ and post $\expb$.
The check can be seen as a \emph{local reductiveness} check for the specific propagated weakest pre: $\vc{\stmtN{i} \symSemi \cdots \symSemi \stmtN{j-1}}(\vc{\stmtN{j} \symSemi \cdots \symSemi \stmtN{n}}(\expb)) \overset{?}{\expleq} \vc{\stmtN{j} \symSemi \cdots \symSemi \stmtN{n}}(\expb)$.
These results are assembled into a slice graph in which a weighted shortest path is searched to appropriately handle nested statements.

We do not require reductiveness to obtain verification-preserving slices.
Consider the program $\stmtAsgn{x}{1}\symSemi \stmtAsgn{x}{x \cdot x}$ w.r.t. pre $1$ and post $x$.
In contrast to ours, their approach cannot remove any statement.
We are able to find the same minimal slices or even smaller slices (c.f. \Cref{sec:evaluation}).
For future work, local reductiveness checks could be explored for finding error-witnessing slices.

As their fairly intricate algorithm has not been implemented, experimental comparison is not feasible.
Scaling in the number of statements, our binary search approaches require only a logarithmic amount of solver calls, while their algorithm requires a quadratic number of calls to obtain the slice graph.
We achieve this by leveraging SMT solvers to explore many options efficiently.

Finally, \cite{DBLP:journals/scp/NavarroO22} only supports the verification of a loop's partial correctness through Park induction, and for its total correctness via a variant-based rule~\cite[Lemma 7.5.1]{DBLP:series/mcs/McIverM05}.
We slice on \HeyVL{}, into which different proof rules can be encoded.
Thus, our method is independent of the proof rules used.

\paragraph{Forwards- and Backwards Reasoning.}
We disagree that \say{specification-based slicing techniques require a combination of both backward propagation of post-conditions and forward propagation of pre-conditions to yield minimal slices}~\cite[Sec.~8]{DBLP:journals/scp/NavarroO22}.
A similar claim is made for classical program slicing in~\cite{DBLP:conf/sefm/BarrosCHP10}.
Our approach \emph{only} uses backwards reasoning and finds minimal slices for their claimed counterexamples (c.f. \cref{sec:benchmarks-verification-preserving}).
As our approach takes a global view of all statements, including \symAssume{} statements encoding the pre, the solver slices using that as well.
However, Caesar also supports forwards reasoning while generating verification conditions as a performance optimization.\footnote{For example, assignments are applied lazily, as in KeY~\cite{DBLP:series/lncs/10001}, c.f. \cite[Section 6.3.4]{DBLP:books/mc/22/Hahnle22}.}
In fact, this optimization is necessary to handle the Bayesian network example, which suffers from exponential blow-up due to branching.

\paragraph{Slicing of Intermediate Languages.}
For classical programs, slicing on IVLs is tackled in~\cite{DBLP:journals/toplas/WardZ07}.
Strengthening the definition of a specification-based slice as a program refinement, they define a \emph{semi-refinement} which must preserve the exact behavior of the original program on all terminating states.
However, it allows slicing all assertions~\cite[Lemma~6.6]{DBLP:journals/toplas/WardZ07}, and they only present a syntactic slicing algorithm.
Our notions of error-witnessing and verification-witnessing slices are less restrictive than the equivalence imposed by a semi-refinement.

\paragraph{Error Reporting.}
\cite{DBLP:journals/fac/LechenetKG18} extends slicing via program dependence graphs to classical programs with $\symAssert$ statements, enabling error detection from failed assertions~\cite[Thm.~6.2]{DBLP:journals/fac/LechenetKG18}.
Localization of the erroring assertion is not done.

The problem of error reporting for an IVL has been tackled in~\cite{DBLP:journals/scp/LeinoMS05,DBLP:journals/jlp/DaillerHMM18} by modifying the verification condition generation.
In~\cite{DBLP:journals/scp/LeinoMS05}, labels (auxiliary Boolean variables) are attached to subformulas and error locations can be extracted from the counterexample produced by an SMT solver.
The approach relies on the SMT solver returning suitably small counterexamples.
As shown in our evaluation, this is not always the case without further minimization, which we implement.
An advantage of our approach is that it performs labeling directly on the level of program statements, enabling further program optimizations.

\cite{DBLP:phd/basesearch/Ruskiewicz12} defines syntactic slicing on an intermediate language.
A failed verification attempt yields a program trace which is used to find the statement causing the failure.
Then, a syntactic slice is constructed which is sufficient to produce the same error.
However, correctness guarantees on the obtained subprogram w.r.t.\ a slicing definition are not given.
Finally, \cite{DBLP:phd/dnb/Liu18} searches for verification-witnessing slices using counterexamples of an SMT solver.
For minimization, minimal unsatisfiable cores are used.
Statements which are deemed irrelevant are not removed, but are replaced by abstractions to keep the program's structure.
Hence, the result is not a subprogram.

\paragraph{Program Refinement.}\label{par:related-work:refinement}
\cite{DBLP:series/mcs/McIverM05} define a program refinement relation for probabilistic programs.
Verification-witnessing and verification-preserving slices can be seen as refinements.
The original program refines a verification-witnessing slice and a verification-preserving slice refines the original program.
Slices are restricted to be subprograms of the original program, while refinement does not impose any syntactic constraints.

\paragraph{Counterexamples in Probabilistic Model Checking.}
Searching for (minimal) counterexamples in probabilistic model checking has been studied extensively~\cite{DBLP:phd/dnb/Jansen15,DBLP:conf/sfm/AbrahamBDJKW14,DBLP:conf/fmcad/JantschHFB20}.
\citeauthor{DBLP:conf/tacas/HanK07} defined \emph{evidences} for property violation; for reachability as a set of paths whose sum of probabilities exceed a given threshold~\cite{DBLP:conf/tacas/HanK07}.
Alternatively, critical subsystems represented by \emph{high-level counterexamples} on the level of a PRISM program are used.
Minimal counterexamples can be obtained by solving a MaxSAT problem~\cite{DBLP:conf/atva/DehnertJWAK14} or an MILP~\cite{DBLP:journals/corr/abs-1305-5055}.

\section{Conclusion and Future Work}

\marknewtext{
	In this paper, we formally defined the notions of \emph{error-witnessing}, \emph{verification-witnessing}, and \emph{verification-preserving} slices for probabilistic program verification.
	These are aimed at localizing verification errors, extracting verification certificates, and tailoring programs to a specification, respectively.
	We applied error-witnessing and verification-witnessing slicing to encodings of proof rules for probabilistic programs in the \HeyVL{} intermediate verification language and gave formal guarantees for user diagnostics for these encodings based on the slicing results.
	Finally, we presented \emph{Brutus}, the first tool for specification-based slicing of probabilistic programs, which is integrated in the Caesar deductive verifier.
	Our implementation supports different algorithms to compute these slices, and we compared their trade-offs using a set of benchmarks.
}

Future work includes the addition of slicing-based error messages and hints to other proof rules, such as latticed $k$-induction~\cite{DBLP:conf/cav/BatzCKKMS20}, and the integration of probabilistic model checking into the slicing process.
We also want to investigate making use of SMT solver results with fewer guarantees (\enquote{unknown} results) to enable better diagnostics in the presence of SMT solver incompleteness.
Further, our implementation is limited by selecting only assert-like statements for error-witnessing slicing.
This could be improved by additional syntactic slicing, or further investigating \emph{local reductiveness}.

\ifAnon{}{
	\section*{Acknowledgments}
	This work was partially supported by the ERC Advanced Research Grant FRAPPANT (grant no.\ 787914), the Research Training Group 2236 \emph{UnRAVeL} (project number 282652900), and by the European Union's Horizon 2020 research and innovation programme under the Marie Sk\l{}odowska-Curie grant agreement MISSION (grant no.\ 101008233). We thank Alexander Bork for suggesting the name \emph{Brutus}.
}

\printbibliography

\newpage
\appendix
\section{Proofs}\label{appendix:proofs}

\paragraph{Notation.}
We denote by $\stmtSkip$ the effectless statement with $\vc{\stmtSkip}(\expa) = \expa$ for all $\expa \in \Expectations$.
Although not part of our formal syntax, we can express it via e.g. $\stmtSkip \equiv \stmtAssert{\infty} \equiv \stmtAssume{\infty}$.

\subsection{Slices from Refinement}

We first formally state and prove how to obtain slices from refinement relations, which will be used later in the proofs below.
The second and third lemmas correspond to our claims in the paragraph \hyperref[par:related-work:refinement]{\emph{Program Refinement}} in \Cref{sec:related-work}:
\Cref{lemma:refinement-verification-witnessing} states that \enquote{the original program refines a verification-witnessing slice} and \Cref{lemma:refinement-verification-preserving} states that \enquote{a verification-preserving slice refines the original program}.

\begin{lemma}[Error-witnessing slice from refinement]\label{lemma:refinement-error-witnessing}
	Let $\slice$ be a subprogram of $\stmt \in \HeyVL$.
	Assume $\stmt$ is refined by $\slice$, i.e. $\vc{\stmt}(\expc) \expleq \vc{\slice}(\expc)$ for all $\expc \in \Expectations$.
	If $\slice$ does not verify w.r.t. $\expa,\expb \in \Expectations$, i.e. $\expa \not\expleq \vc{\slice}(\expb)$, then $\slice$ is an error-witnessing slice of $\stmt$ w.r.t. $\expa,\expb$.
\end{lemma}

\begin{proof}
	Since $\slice$ does not verify w.r.t. $\expa,\expb$, we satisfy condition (1) of \Cref{def:slice:error_witnessing}.

	Let $\State'$ be a state such that $\State' \not\models \lowerTriple{\expa}{\slice}{\expb}$ holds.
	This means $\expa(\State) \not\leq \vc{\slice}(\expb)(\State)$.
	However, since $\stmt$ is refined by $\slice$, we have $\vc{\stmt}(\expb)(\State) \leq \vc{\slice}(\expb)(\State)$.
	Therefore, $\State' \not\models \lowerTriple{\expa}{\stmt}{\expb}$ holds, satisfying condition (2) of \Cref{def:slice:error_witnessing}.
\end{proof}

\begin{lemma}[Verification-witnessing slice from refinement]\label{lemma:refinement-verification-witnessing}
	Let $\slice$ be a subprogram of $\stmt \in \HeyVL$.
	Assume $\slice$ is refined by $\stmt$, i.e. $\vc{\slice}(\expc) \expleq \vc{\stmt}(\expc)$ for all $\expc \in \Expectations$.
	Then, for all $\expa,\expb \in \Expectations$, $\slice$ is a verification-witnessing slice of $\stmt$ w.r.t. $\expa,\expb$.
\end{lemma}

\begin{proof}
	Assume $\models \lowerTriple{\expa}{\slice}{\expb}$, i.e. $\expa \expleq \vc{\slice}(\expb)$.
	Since $\slice$ is refined by $\stmt$, we have $\expa \expleq \vc{\slice}(\expb) \expleq \vc{\stmt}(\expb)$.
	Therefore, $\models \lowerTriple{\expa}{\stmt}{\expb}$ holds.
	On the other hand, if $\models \lowerTriple{\expa}{\slice}{\expb}$ does not hold, the definition is satisfied trivially.
	Therefore, \Cref{def:slice:verification_witnessing} is satisfied.
\end{proof}

\begin{lemma}[Verification-preserving slice from refinement]\label{lemma:refinement-verification-preserving}
	Let $\slice$ be a subprogram of $\stmt \in \HeyVL$.
	Assume $\stmt$ is refined by $\slice$, i.e. $\vc{\stmt}(\expc) \expleq \vc{\slice}(\expc)$ for all $\expc \in \Expectations$.
	Then, for all $\expa,\expb \in \Expectations$, $\slice$ is a verification-preserving slice of $\stmt$ w.r.t. $\expa,\expb$.
\end{lemma}

\begin{proof}
	Assume $\models \lowerTriple{\expa}{\stmt}{\expb}$, i.e. $\expa \expleq \vc{\stmt}(\expb)$.
	Since $\stmt$ is refined by $\slice$, we have $\expa \expleq \vc{\stmt}(\expb) \expleq \vc{\slice}(\expb)$.
	Therefore, $\models \lowerTriple{\expa}{\slice}{\expb}$ holds.
	On the other hand, if $\models \lowerTriple{\expa}{\stmt}{\expb}$ does not hold, the definition is satisfied trivially.
	Therefore, \Cref{def:slice:verification_preserving} is satisfied.
\end{proof}

\subsection[Proofs for Section~\ref{sec:slicing-heyvl}]{Proofs for \Cref{sec:slicing-heyvl}}

\noindent%
Recall \Cref{lemma:reductive} from \cpageref{lemma:reductive}:

\lemmaReductive*

\begin{proof}
	By \Cref{def:reductive}, a statement $\stmt$ is reductive if for all $\expa \in \Expectations$, $\vc{\stmt}(\expa) \expleq \expa$ holds.

	\begin{itemize}
		\item Let $\stmt = \stmtAssert{\expb}$. Then, for all $\expa \in \Expectations$,
		      \[
			      \vc{\stmtAssert{\expb}}(\expa) = \expb \sqcap \expa \morespace{\expleq} \expa.
		      \]
		\item Let $\stmt = \symUp\stmtAssume{\expb}$. Then, for all $\expa \in \Expectations$,
		      \[
			      \vc{\symUp\stmtAssume{\expb}}(\expa) = \expb \coimpl \expa = \lam{\State}{\ifThenElse{\expb(\State) \geq \expa(\State)}{0}{\expa(\State)}} \morespace{\expleq} \expa~.
		      \]
		\item Let $\stmt = \stmtValidate$. Then, for all $\expa \in \Expectations$,
		      \[
			      \vc{\stmtValidate}(\expa) = \expvalidate{\expa} = \lam{\State}{\ifThenElse{\expa(\State) = \infty}{\infty}{0}} \morespace{\expleq} \expa~.
		      \]
	\end{itemize}
\end{proof}

\noindent%
Recall \Cref{theorem:reductive-slice} from \cpageref{theorem:reductive-slice}:

\theoremReductiveSlice*

\begin{proof}\label{proof:reductive-slice}
	Let $\slice$ be a subprogram of $\stmt$ obtained by only removing reductive statements.
	We show by induction on the structure of $\stmt$ that the following refinement holds for all $\expa \in \Expectations$:
	\[
		\vc{\stmt}(\expa) \expleq \vc{\slice}(\expa)~.
	\]

	Let $\stmt \in \HeyVL$ and $\slice$ be a subprogram of $\stmt$.
	For each case, when $\stmt = \slice$, the property holds trivially and we will omit them in the following.
	Assume therefore that $\stmt \neq \slice$.

	\emph{Base cases.}
	These are the cases $S \in \{ \stmtDeclInit{x}{\typevar}{\mu},~ \stmtAssert{\expb},~\allowbreak \stmtAssume{\expb},~\allowbreak \stmtHavoc{x},~\allowbreak \stmtValidate,~\allowbreak  \stmtTick{\aexpr},~\allowbreak \coAssert{\expb},~\allowbreak \coAssume{\expb},~\allowbreak\coHavoc{x},~\allowbreak \coValidate \}$.
	Either $\stmt$ is not reductive (in which case $\stmt = \slice$ holds) or it is reductive.
	If removed, we have $\slice = \stmtSkip$.
	By reductiveness, we obtain $\vc{\stmt}(\expa) \expleq \vc{\stmtSkip}(\expa)$.

	\emph{Induction hypothesis.}
	Assume the property holds for all subprograms $\stmtOneP,\stmtTwoP$ of $\stmtOne,\stmtTwo$, respectively.

	\emph{Induction step.}
	\begin{itemize}
		\item Case $\stmt = \stmtSeq{\stmtOne}{\stmtTwo}$: We have $\slice = \stmtSeq{\stmtOneP}{\stmtTwoP}$ for subprograms $\stmtOneP,\stmtTwoP$ of $\stmtOne,\stmtTwo$, respectively. Then,
		      \begin{align*}
			      \vc{\stmtSeq{\stmtOne}{\stmtTwo}}(\expa) & = \vc{\stmtOne}(\vc{\stmtTwo}(\expa))                                                    \\
			                                               & \expleq \vc{\stmtOneP}(\vc{\stmtTwo}(\expa)) \tag{IH}                                    \\
			                                               & \expleq \vc{\stmtOneP}(\vc{\stmtTwoP}(\expa)) \tag{IH, monotonicity of $\vc{\stmtOneP}$} \\
			                                               & = \vc{\slice}(\expa)~.
		      \end{align*}
		\item Case $\stmt = \stmtDemonic{\stmtOne}{\stmtTwo}$: We have $\slice = \stmtDemonic{\stmtOneP}{\stmtTwoP}$ for subprograms $\stmtOneP,\stmtTwoP$ of $\stmtOne,\stmtTwo$, respectively.
		      \begin{align*}
			      \vc{\stmtDemonic{\stmtOne}{\stmtTwo}}(\expa) & = \vc{\stmtOne}(\expa) \sqcap \vc{\stmtTwo}(\expa)                  \\
			                                                   & \expleq \vc{\stmtOneP}(\expa) \sqcap \vc{\stmtTwo}(\expa) \tag{IH}  \\
			                                                   & \expleq \vc{\stmtOneP}(\expa) \sqcap \vc{\stmtTwoP}(\expa) \tag{IH} \\
			                                                   & = \vc{\slice}(\expa)~.
		      \end{align*}
		\item Case $\stmt = \stmtAngelic{\stmtOne}{\stmtTwo}$: We have $\slice = \stmtAngelic{\stmtOneP}{\stmtTwoP}$ for subprograms $\stmtOneP,\stmtTwoP$ of $\stmtOne,\stmtTwo$, respectively.
		      \begin{align*}
			      \vc{\stmtAngelic{\stmtOne}{\stmtTwo}}(\expa) & = \vc{\stmtOne}(\expa) \sqcup \vc{\stmtTwo}(\expa)                  \\
			                                                   & \expleq \vc{\stmtOneP}(\expa) \sqcup \vc{\stmtTwo}(\expa) \tag{IH}  \\
			                                                   & \expleq \vc{\stmtOneP}(\expa) \sqcup \vc{\stmtTwoP}(\expa) \tag{IH} \\
			                                                   & = \vc{\slice}(\expa)~.
		      \end{align*}
	\end{itemize}

	Thus, $\vc{\stmt}(\expa) \expleq \vc{\slice}(\expa)$ holds for all subprograms $\slice$ of $\stmt$ where $\slice$ was obtained by only removing reductive statements.
	By \Cref{lemma:refinement-error-witnessing}, $\slice$ is an error-witnessing slice of $\stmt$ w.r.t. $(\expa,\expb)$.
\end{proof}

\noindent%
Recall \Cref{lemma:slicing-demonic} from \cpageref{lemma:slicing-demonic}:

\lemmaSlicingDemonic*

\begin{proof}
	Let $\expa \in \Expectations$.
	We have
	\[
		\vc{\stmtDemonic{\stmtOne}{\stmtTwo}}(\expa) = \vc{\stmtOne}(\expa) \sqcap \vc{\stmtTwo}(\expa).
	\]
	By definition of $\sqcap$, we obtain both
	\[
		\vc{\stmt}(\expa) \expleq \vc{\stmtOne}(\expa) \quad\text{and}\quad \vc{\stmt}(\expa) \expleq \vc{\stmtTwo}(\expa)~. \qedhere
	\]
\end{proof}

Although not stated explicitly in the main text, a dual statement to \Cref{lemma:slicing-demonic} exists for the angelic choice $\stmt = \stmtAngelic{\stmtOne}{\stmtTwo}$, and can be shown similarly.
For all $\expa \in \Expectations$, we have
\[
	\vc{\stmt}(\expa) \expgeq \vc{\stmtOne}(\expa) \quad\text{and}\quad \vc{\stmt}(\expa) \expgeq \vc{\stmtTwo}(\expa)~.
\]

\noindent%
Recall \Cref{lemma:extensive} from \cpageref{lemma:extensive}:

\lemmaExtensive*

\begin{proof}
	By \Cref{def:extensive}, a statement $\stmt$ is extensive if for all $\expa \in \Expectations$, $\vc{\stmt}(\expa) \expgeq \expa$ holds.

	\begin{itemize}
		\item Let $\stmt = \stmtAssume{\expb}$. Then, for all $\expa \in \Expectations$,
		      \[
			      \vc{\stmtAssume{\expb}}(\expa) = \expb \impl \expa = \lam{\State}{\ifThenElse{\expb(\State) \leq \expa(\State)}{\infty}{\expa(\State)}} \morespace{\expgeq} \expa~.
		      \]
		\item Let $\stmt = \coAssert{\expb}$. Then, for all $\expa \in \Expectations$,
		      \[
			      \vc{\coAssert{\expb}}(\expa) = \expb \sqcup \expa \morespace{\expgeq} \expa~.
		      \]
		\item Let $\stmt = \coValidate$. Then, for all $\expa \in \Expectations$,
		      \[
			      \vc{\coValidate}(\expa) = \expcovalidate{\expa} = \lam{\State}{\ifThenElse{\expa(\State) = 0}{0}{\infty}} \morespace{\expgeq} \expa~.
		      \]
		\item Let $\stmt = \stmtTick{\expb}$. Then, for all $\expa \in \Expectations$,
		      \[
			      \vc{\stmtTick{\expb}}(\expa) = \expb + \expa \morespace{\expgeq} \expa~. \tag{$\expb$ is non-negative}
		      \]
	\end{itemize}
\end{proof}

\noindent%
Recall \Cref{theorem:extensive-slice} from \cpageref{theorem:extensive-slice}:

\theoremExtensiveSlice*

\begin{proof}
	Let $\slice$ be a subprogram of $\stmt$ obtained by only removing extensive statements.
	It can be shown that the following holds for all $\expa \in \Expectations$:
	\[
		\vc{\stmt}(\expa) \expgeq \vc{\slice}(\expa)~.
	\]
	The proof is completely dual to the structural induction done for \Cref{theorem:reductive-slice} (c.f. \cpageref{proof:reductive-slice}): one simply replaces $\expleq$ by $\expgeq$.
	By \Cref{lemma:refinement-verification-witnessing}, we obtain the desired result.
\end{proof}

\noindent%
Recall \Cref{lemma:error-witnessing} from \cpageref{lemma:error-witnessing}:

\lemmaErrorWitnessing*

\begin{proof}
	If $\slice$ is an error-witnessing slice of $\stmt$ w.r.t. $(\expa,\expb)$, then by condition (2) of \Cref{def:slice:error_witnessing}:
	\[
		\forall \sigma'.~ (\sigma' \not\models \lowerTriple{\expa}{\slice}{\expb} \implies \sigma' \not\models \lowerTriple{\expa}{\stmt}{\expb})~.
	\]
	Rewriting the implication:
	\[
		\forall \sigma'.~ (\sigma' \models \lowerTriple{\expa}{\slice}{\expb} \lor \sigma' \not\models \lowerTriple{\expa}{\stmt}{\expb})~.
	\]
	Rewriting back to an implication, we get the desired result, equivalent to \Cref{def:slice:verification_preserving}:
	\[
		\forall \sigma'.~ (\sigma' \models \lowerTriple{\expa}{\stmt}{\expb} \implies \sigma' \models \lowerTriple{\expa}{\slice}{\expb})~.
		\qedhere
	\]
\end{proof}

\subsection[Proofs for Section~\ref{sec:slicing_on_encodings}]{Proofs for \Cref{sec:slicing_on_encodings}}

\noindent%
Recall \Cref{thm:verification-encoding-components} from \cpageref{thm:verification-encoding-components}.
In this context, we consider slices w.r.t. pre $\infty$ and post $\infty$.

\thmVerificationEncodingComponents*

\begin{proof}
	(1) Let $\stmtAssume{\expa}\symSemi \stmtP$ be an error-witnessing slice of $\stmt$ w.r.t. $\infty,\infty$.
	By definition, we get $\infty \not\expleq \vc{\stmtAssume{\expa}\symSemi \stmtP}(\infty)$.
	Applying adjointness of $\symAssume$/$\impl$, we get $\expa \not\expleq \vc{\stmtP}(\infty)$.
	By monotonicity, we have $\vc{\stmtP}(\infty \sqcap \expc) \expleq \vc{\stmtP}(\infty)$ for all $\expc \in \Expectations$.
	By transitivity of $\expleq$, we get $\expa \not\expleq \vc{\stmtP}(\expc)$, i.e. $\not\models \lowerTriple{\expa}{\stmtP}{\expc}$.

	(2) Let $\stmtP\symSemi \stmtAssert{\expb}$ be a verifying verification-witnessing slice of $\stmt$ w.r.t. $\infty,\infty$.
	As program $\stmtP\symSemi \stmtAssert{\expb}$ verifies, we know $\infty \expleq \vc{\stmtP\symSemi \stmtAssert{\expb}}(\infty)$.
	For all $\expc \in \Expectations$, we have $\expc \expleq \infty$.
	Thus, $\expc \expleq \vc{\stmtP\symSemi \stmtAssert{\expb}}(\infty) = \vc{\stmtP}(\expb)$, i.e. $\models \lowerTriple{\expc}{\stmtP}{\expb}$ for $\expc \in \Expectations$.
\end{proof}

\noindent%
Recall \Cref{theorem:spec-statement-messages} from \cpageref{theorem:spec-statement-messages}:

\thmSpecStatementMessages*

\begin{proof}
	\begin{enumerate}
		\item Let $\stmtP = \Assert{\expb} \symSemi \bar{\stmtP}$.
		      Since $\stmtP$ is an error-witnessing slice and minimal, and the only removed statement $\Assert{\expb}$ is reductive, $\bar{\stmtP}$ verifies.
		      Hence,
		      \[
			      \expaP \expleq \vc{\bar{\stmtP}}(\expbP).
		      \]
		      If also $\expaP \expleq \expb$ holds, we get
		      \[
			      \expaP \expleq \expb \sqcap \vc{\bar{\stmtP}}(\expbP) = \vc{\stmtP}(\expbP),
		      \]
		      thus contradicting that $\stmtP$ is an error-witnessing slice with a counterexample state.
		      Hence, we must have $\expaP \not\expleq \expb$.
		\item Notice that the only statements that might be present in $\stmtP$ are reductive.
		      As $\stmtP$ verifies, we thus have
		      \[
			      \expaP \expleq \vc{\stmtP}(\expbP) \expleq \expbP. \qedhere
		      \]
	\end{enumerate}
\end{proof}

\noindent%

The following lemma states how the assertions in the encoding contribute to the final proof obligation.
It lets us infer the error messages of \Cref{theorem:park-error-messages,theorem:park-verification-messages}.

\begin{restatable}[Park Induction Error Components]{lemma}{lemmaParkErrorComponents}\label{lemma:park-error-components}
	Let $\stmt = \park{b}{\varBody}{\varInvariant}$ as in \Cref{fig:heyvl-park-induction} and $\stmtP$ be obtained by removing any subset of the assertions \stmtNo{1} and \stmtNo{2} in $\stmt$.
	Then:
	\begin{align*}
		\vc{\stmtP}(\expb) & {}= \ParkIndCondPre(\stmtP) \sqcap \ParkIndCondInd(\stmtP) \sqcap \ParkIndCondPost(\stmtP),
	\end{align*}

	\begin{minipage}[c]{0.3\textwidth}
		where for $i \in \Set{1,2}$:
		\begin{align*}
			A_{\stmtNo{i}} & {}=
			\begin{cases}
				\varInvariant, & \text{if $\stmtP$ includes \stmtNo{i}}, \\
				\infty,        & \text{else},
			\end{cases}
		\end{align*}
	\end{minipage}%
	\begin{minipage}[c]{0.07\textwidth}
		\centering
		\hspace{0.9em}\emph{and}
	\end{minipage}%
	\begin{minipage}[c]{0.6\textwidth}
		\begin{align*}
			\hspace{0.9em}\text{a)}\hspace{-0.9em} &  & \ParkIndCondPre(\stmtP)  & {}= A_{\stmtNo{1}},                                                                                       \\
			\text{b)}\hspace{-0.9em}               &  & \ParkIndCondInd(\stmtP)  & {}= \iquant{\vec{v}}{\expvalidate{(\varInvariant \sqcap \embed{b}) \impl \vc{\varBody}(A_{\stmtNo{2}})}}, \\
			\text{c)}\hspace{-0.9em}               &  & \ParkIndCondPost(\stmtP) & {}= \iquant{\vec{v}}{\expvalidate{(\varInvariant \sqcap \embed{\neg b}) \impl \expb}}.
		\end{align*}
	\end{minipage}
\end{restatable}

\begin{proof}
	First recall that for any $\expc \in \Expectations$,
	\[
		\expc = \infty \sqcap \expc = \vc{\stmtAssert{\infty}}(\expc).
	\]
	Hence, omitting assertion \stmtNo{1} or \stmtNo{2} in $\stmtP$ corresponds to replacing the respective statement by $\stmtAssert{\infty}$.
	We obtain the result by simple rewriting:
	\begin{align*}
		\vc{\stmtP}(\expb) & {}= \vc{\stmtAssert{A_{\stmtNo{1}}}}\left(\iquant{\vec{v}}{\triangle\!\left(I \impl \left( \left(\embed{b} \impl \vc{\stmtSeq{\stmtSeq{C}{\stmtAssert{A_{\stmtNo{2}}}}}{\stmtAssume{\embed{\exprFalse}}}}(\expb)\right)\right.\right.}\right.                                                                                                                                                         \\
		                   & \hspace{13.5em} \sqcap \left(\embed{\neg b} \impl \expb\right) \Big)\Big)                                                                                                                                                                                                                                                                                                                             \\
		                   & {}= {A_{\stmtNo{1}}} \sqcap \left(\iquant{\vec{v}}{\expvalidate{I \impl \left( \left(\embed{b} \impl \vc{C}(\vc{\stmtAssert{A_{\stmtNo{2}}}}(\infty))\right) \sqcap \left(\embed{\neg b} \impl \expb\right) \right)}}\right)                                                                                                                                                                          \\
		                   & {}= {A_{\stmtNo{1}}} \sqcap \left(\iquant{\vec{v}}{\expvalidate{I \impl \left( \left(\embed{b} \impl \vc{C}(A_{\stmtNo{2}})\right) \sqcap \left(\embed{\neg b} \impl \expb\right) \right)}}\right)                                                                                                                                                                                                    \\
		                   & {}= {A_{\stmtNo{1}}} \sqcap \left(\iquant{\vec{v}}{\expvalidate{ \left(I \impl \left(\embed{b} \impl \vc{C}(A_{\stmtNo{2}})\right)\right) \sqcap \left(I \impl \left(\embed{\neg b} \impl \expb\right) \right)}}\right)                   \tag{$\impl$ and $\sqcap$ commute in antecedent}                                                                                                            \\
		                   & {}= {A_{\stmtNo{1}}} \sqcap \left(\iquant{\vec{v}}{\expvalidate{ \left(\left(I \sqcap \embed{b}\right) \impl \vc{C}(A_{\stmtNo{2}})\right) \sqcap \left(\left(I \sqcap \embed{\neg b}\right) \impl \expb \right)}}\right)                 \tag{collapse sequence of $\impl$}                                                                                                                          \\
		                   & {}= \underbrace{{A_{\stmtNo{1}}}}_{\ParkIndCondPre(\stmtP)} \sqcap \underbrace{\iquant{\vec{v}}{\expvalidate{ \left(I \sqcap \embed{b}\right) \impl \vc{C}(A_{\stmtNo{2}}) }}}_{\ParkIndCondInd(\stmtP)} \sqcap \underbrace{\iquant{\vec{v}}{\expvalidate{ \left(I \sqcap \embed{\neg b}\right) \impl \expb }}}_{\ParkIndCondPost(\stmtP)}.         \tag{$\Inf$/$\symValidate$ commute with $\sqcap$} \\
	\end{align*}
\end{proof}

\noindent%
We introduce an auxiliary lemma that will be used for the proof of \Cref{theorem:park-error-messages}.

\begin{lemma}
	\label{lemma:heyvl:preserve_infty_without_assert}
	If $\stmt$ is a \HeyVL{} program that contains neither $\symAssert{}$ nor $\symUp{}\symAssume{}$ statements, then $\stmt$ verifies, i.e. $\vc{\stmt}(\infty) = \infty$.
	Dually, if $\stmt$ is a \HeyVL{} program that contains neither $\symUp\symAssert{}$ nor $\symAssume{}$ nor $\symTick$ statements, then $\stmt$ co-verifies, i.e. $\vc{\stmt}(0) = 0$.
\end{lemma}
\begin{proof}
	Let $\stmt$ be a \HeyVL{} program without $\symAssert$ nor $\symUp\symAssume$ statements.
	A simple induction on the structure of $\stmt$ shows that $\vc{\stmt}(\infty) = \infty$ holds.

	\emph{Base cases.}
	\begin{itemize}
		\item Case $\stmt = \stmtDeclInit{x}{\typevar}{\mu}$: Let $\mu = \pexp{p_1}{\termvar_1} \pexpand \ldots \pexpand \pexp{p_n}{\termvar_n}$. Since we consider distributions, $\sum_{i=1}^n p_i = 1$. Then we have
		      \[
			      \vc{\stmt}(\infty) = p_1 \cdot \infty\substBy{x}{\termvar_1} + \ldots + p_n \cdot \infty\substBy{x}{\termvar_n} = \infty.
		      \]
		\item Case $\stmt = \stmtAssert{\expb}$: Impossible.
		\item Case $\stmt = \stmtAssume{\expb}$: $\vc{\stmt}(\infty) = \expb \impl \infty = \infty$.
		\item Case $\stmt = \stmtHavoc{x}$ $\vc{\stmt}(\infty) = \iquant{x}{\infty} = \infty$.
		\item Case $\stmt = \stmtValidate$: $\vc{\stmt}(\infty) = \expvalidate{\infty} = \infty$.
		\item Case $\stmt = \stmtTick{\aexpr}$: $\vc{\stmt}(\infty) = \infty + \aexpr = \infty$.
		\item Case $\stmt = \coAssert{\expb}$: $\vc{\stmt}(\infty) = \infty \sqcup \expb = \infty$.
		\item Case $\stmt = \coAssume{\expb}$: Impossible.
		\item Case $\stmt = \coHavoc{x}$: $\vc{\stmt}(\infty) = \squant{x}{\infty} = \infty$.
		\item Case $\stmt = \coValidate$: $\vc{\stmt}(\infty) = \expcovalidate{\infty} = \infty$.
	\end{itemize}

	\emph{Inductive step.}
	Assume the property holds for $\stmtOne,\stmtTwo$.
	\begin{itemize}
		\item Case $\stmt = \stmtSeq{\stmtOne}{\stmtTwo}$: $\vc{\stmt}(\infty) = \vc{\stmtOne}(\vc{\stmtTwo}(\infty)) = \vc{\stmtOne}(\infty) = \infty$.
		\item Case $\stmt = \stmtDemonic{\stmtOne}{\stmtTwo}$: $\vc{\stmt}(\infty) = \vc{\stmtOne}(\infty) \sqcap \vc{\stmtTwo}(\infty) = \infty \sqcap \infty = \infty$.
		\item Case $\stmt = \stmtAngelic{\stmtOne}{\stmtTwo}$: $\vc{\stmt}(\infty) = \vc{\stmtOne}(\infty) \sqcup \vc{\stmtTwo}(\infty) = \infty \sqcup \infty = \infty$.
	\end{itemize}

	Therefore, $\stmt$ verifies.
	The proof for the dual case is analogous.
\end{proof}

\noindent%
Recall \Cref{theorem:park-error-messages} from \cpageref{theorem:park-error-messages}:

\theoremParkErrorMessages*

Note that the dual version of this theorem requires a modification in condition (3) that $\slice$ contains neither $\symUp\symAssert$ nor $\symAssume$ nor $\symTick$ statements, corresponding to the dual case in \cref{lemma:heyvl:preserve_infty_without_assert}.

\begin{proof}
	\begin{enumerate}
		\item Assume that $\stmtP$ includes assertion \stmtNo{1}.
		      By minimality of $\stmtP$, $\stmtP \ominus \stmtNo{1}$ does not have a counterexample, i.e.\ $\expa \expleq \vc{\stmtP \ominus \stmtNo{1}}(\expb)$.
		      Hence, by~\cref{lemma:park-error-components}
		      \[
			      \expa \expleq \vc{\stmtP \ominus \stmtNo{1}}(\expb) = \ParkIndCondPre(\stmtP \ominus \stmtNo{1}) \sqcap \ParkIndCondInd(\stmtP \ominus \stmtNo{1}) \sqcap \ParkIndCondPost(\stmtP \ominus \stmtNo{1}).
		      \]
		      From this we can conclude
		      \[
			      \expa \expleq \ParkIndCondInd(\stmtP \ominus \stmtNo{1}) = \ParkIndCondInd(\stmtP) \quad\text{and}\quad \expa \expleq \ParkIndCondPost(\stmtP \ominus \stmtNo{1}) = \ParkIndCondPost(\stmtP).
		      \]

		      If also $\expa \expleq \ParkIndCondPre(\stmtP)$ holds, we use~\cref{lemma:park-error-components} to get
		      \[
			      \expa \expleq \ParkIndCondPre(\stmtP) \sqcap \ParkIndCondInd(\stmtP) \sqcap \ParkIndCondPost(\stmtP) = \vc{\stmtP}(\expb),
		      \]
		      thus contradicting that $\stmtP$ is an error-witnessing slice with a counterexample state.

		      We must thus have $\expa \not\expleq \ParkIndCondPre(\stmtP) = A_{\stmtNo{1}} = I$, i.e.\ there exists a state $\State'$ with $\expa(\State') \not\leq I(\State')$.

		\item Assume that $\stmtP$ includes assertion \stmtNo{2}.
		      Similar to the previous case, we can use minimality of $\stmtP$ to conclude that
		      \[
			      \expa \not\expleq \ParkIndCondInd(\stmtP) = \iquant{\vec{v}}{\expvalidate{(I \sqcap \embed{b}) \impl \vc{C}(I)}}.
		      \]
		      If for every state $\State'$, ${\expvalidate{(I \sqcap \embed{b}) \impl \vc{C}(I)}}\substBy{\vec{v}}{\vec{\nu}}(\State') = \infty$ for all $\vec{\nu} \in \vec{\tau}$, then $\expa \expleq \infty = \ParkIndCondInd(\stmtP)$, a contradiction.
		      Hence, there must exist a state $\State'$ and valuation $\vec{\nu} \in \vec{\tau}$, such that ${\expvalidate{(I \sqcap \embed{b}) \impl \vc{C}(I)}}\substBy{\vec{v}}{\vec{\nu}}(\State') \neq \infty$.
		      Applying the definition of $\symValidate$ yields
		      \begin{align*}
			      \infty & {}\not\leq \left((I \sqcap \embed{b}) \impl \vc{C}(I)\right)\substBy{\vec{v}}{\vec{\nu}}(\State')                            \\
			             & {}=\left((I \sqcap \embed{b}) \impl \vc{C}(I)\right)(\State'\substBy{\vec{v}}{\vec{\nu}}). \tag{apply substitution to state}
		      \end{align*}
		      Writing $\State'' = \State'\substBy{\vec{v}}{\vec{\nu}}$ and using the adjointness of $\impl$ and $\sqcap$ we immediately get
		      \[
			      \infty \sqcap (I \sqcap \embed{b})(\State'') \not\leq \vc{C}(I)(\State'').
		      \]
		      We have that $\State''$ fulfills $b$, or else $(I \sqcap \embed{b})(\State'') = 0 \leq \vc{C}(I)(\State'')$ leads to a contradiction.
		      Hence, $\State''$ is the state witnessing that the loop invariant is not inductive:
		      \[
			      I(\State'') = I(\State'') \sqcap \infty =  (I \sqcap \embed{b})(\State'') \not\leq \vc{C}(I)(\State'').
		      \]

		\item Assume that $\stmtP$ includes neither assertion \stmtNo{1} nor \stmtNo{2}, and $C$ contains neither $\symAssert$ nor $\symUp\symAssume$ statements.
		      By~\cref{lemma:park-error-components}
		      \begin{align*}
			      \expa & {}\not\expleq \ParkIndCondPre(\stmtP) \sqcap \ParkIndCondInd(\stmtP) \sqcap \ParkIndCondPost(\stmtP)                                                                                                      \\
			            & {}= \infty \sqcap \ParkIndCondInd(\stmtP) \sqcap \ParkIndCondPost(\stmtP)                                             \tag{assertion \stmtNo{1} not in $\stmtP$}                                          \\
			            & {}= \ParkIndCondInd(\stmtP) \sqcap \ParkIndCondPost(\stmtP)                                                                                                                                               \\
			            & {}= (\iquant{\vec{v}}{\expvalidate{(I \sqcap \embed{b}) \impl \vc{C}(\infty)}}) \sqcap \ParkIndCondPost(\stmtP)       \tag{assertion \stmtNo{2} not in $\stmtP$}                                          \\
			            & {}= (\iquant{\vec{v}}{\expvalidate{(I \sqcap \embed{b}) \impl \infty}}) \sqcap \ParkIndCondPost(\stmtP)               \tag{$\vc{C}(\infty) = \infty$ by~\cref{lemma:heyvl:preserve_infty_without_assert}} \\
			            & {}= \infty \sqcap \ParkIndCondPost(\stmtP)                                                                            \tag{simplification}                                                                \\
			            & {}= \ParkIndCondPost(\stmtP) = \iquant{\vec{v}}{\expvalidate{(I \sqcap \embed{\neg b}) \impl \expb}}.                                                                                                     \\
		      \end{align*}
		      Reasoning analogous to the previous case there exists a state $\State'$ and valuation $\vec{\nu} \in \vec{\tau}$ with
		      \begin{align*}
			      \infty & {}\not\leq \left((I \sqcap \embed{\neg b}) \impl \expb\right)\substBy{\vec{v}}{\vec{\nu}}(\State').
		      \end{align*}
		      Again, writing $\State'' = \State'\substBy{\vec{v}}{\vec{\nu}}$ and using the adjointness of $\impl$ and $\sqcap$ we immediately get $(I \sqcap \embed{\neg b})(\State'') \not\leq \expb(\State'')$.
		      $\State''$ does not fulfill $b$, or else $(I \sqcap \embed{\neg b})(\State'') = 0 \leq \expb(\State'')$ leads to a contradiction.
		      Hence, $\State''$ is the state witnessing that the loop invariant does not entail the post:
		      \[
			      I(\State'') = I(\State'') \sqcap \infty =  (I \sqcap \embed{\neg b})(\State'') \not\leq \expb(\State'').\qedhere
		      \]
	\end{enumerate}
\end{proof}

\noindent%
Recall \Cref{theorem:park-verification-messages} from \cpageref{theorem:park-verification-messages}:

\theoremParkVerificationMessages*

\begin{proof}
	First recall that for any $\expc \in \Expectations$,
	\[
		\expc = \infty \impl \expc = \vc{\stmtAssume{\infty}}(\expc) \qmorespace{\text{and}} \expc = \infty \sqcap \expc = \vc{\stmtAssert{\infty}}(\expc) .
	\]
	Hence, omitting assumption \stmtNo{I} or \stmtNo{II} in $\stmtP$ corresponds to replacing the respective statement by $\stmtAssume{\infty}$.
	And $\stmtAssume{\infty}$/$\Assert{\infty}$ statements can be added without changing the program's semantics.
	We write
	\[
		A_{\stmtNo{I}}           =
		\begin{cases}
			\varInvariant, & \text{if $\stmtP$ includes \stmtNo{I}}, \\
			\infty,        & \text{else}.
		\end{cases}
		\qqand
		A_{\stmtNo{II}}           =
		\begin{cases}
			\embed{\exprFalse}, & \text{if $\stmtP$ includes \stmtNo{II}}, \\
			\infty,             & \text{else}.
		\end{cases}
	\]

	\begin{enumerate}
		\item Assume that $\stmtP$ does not include assumption \stmtNo{I}.
		      Then
		      \begin{align*}
			      \expa & {}\expleq \vc{\stmtP}(\expb)                                                                                                                                                                                                                          \\
			            & {}= \vc{\stmtSeq{\stmtAssert{\varInvariant}}{\stmtSeq{\Havoc{variables}}{\stmtSeq{\Validate}{{\stmtIf{b}{\stmtSeq{C}{\stmtSeq{\Assert{\varInvariant}}{\stmtAssume{A_{\stmtNo{II}}}}}}{}}}}}}(\expb)                                                   \\
			            & {}\expleq \vc{{\stmtSeq{\Havoc{variables}}{\stmtSeq{\Validate}{{\stmtIf{b}{\stmtSeq{C}{{\stmtAssume{A_{\stmtNo{II}}}}}}{}}}}}}(\expb)                                                                 \tag{remove reductive statements}               \\
			            & {}= \vc{\stmtSeq{\stmtAssert{\infty}}{\stmtSeq{\Havoc{variables}}{\stmtSeq{\Validate}{{\stmtIf{b}{\stmtSeq{C}{\stmtSeq{\Assert{\infty}}{\stmtAssume{A_{\stmtNo{II}}}}}}{}}}}}}(\expb)                                                                 \\
			            & {}\expleq \vc{\park{b}{C}{\infty}}(\expb).                                                                                                                                                            \tag{add extensive statement $\Assume{\infty}$}
		      \end{align*}
		      Hence, $\lowerTriple{\expa}{\park{b}{C}{\infty}}{\expb}$ is valid.

		\item Assume that $\stmtP$ does not include assumption \stmtNo{II}.
		      We have
		      \begin{align*}
			      \expa & {}\expleq \vc{\stmtP}(\expb)                                                                                                                                                                           \\
			            & {}\expleq \vc{\stmtSeq{\stmtAssert{I}}{\stmtSeq{\stmtAssume{A_{\stmtNo{I}}}}{\stmtIf{b}{\stmtSeq{C}{\stmtAssume{A_{\stmtNo{II}}}}}{}}}}(\expb)    \tag{removing reductive statements}                  \\
			            & {}= \vc{\stmtSeq{\stmtAssert{I}}{\stmtSeq{\stmtAssume{A_{\stmtNo{I}}}}{\stmtIf{b}{C}{}}}}(\expb)                                                  \tag{$A_{\stmtNo{II}} = \infty$}                     \\
			            & {}= I \sqcap \left( A_{\stmtNo{I}} \impl \vc{\stmtIf{b}{C}{}}(\expb) \right)                                                                                                                           \\
			            & {}= I \sqcap A_{\stmtNo{I}} \sqcap \left( A_{\stmtNo{I}} \impl \vc{\stmtIf{b}{C}{}}(\expb) \right)                                                \tag{$I \expleq A_{\stmtNo{I}} \in \Set{I, \infty}$} \\
			            & {}\expleq I \sqcap \vc{\stmtIf{b}{C}{}}(\expb)                                                                                                    \tag{Modus ponens}                                   \\
			            & {}\expleq \vc{\stmtIf{b}{C}{}}(\expb).
		      \end{align*}
		      Hence, $\lowerTriple{\expa}{\stmtIf{b}{C}{}}{\expb}$ is valid. \qedhere
	\end{enumerate}
\end{proof}

\subsection{Proofs for Section~\ref{sec:implementation}}

\begin{table}[t]
	\caption{\HeyVL{} program transformations of statements for slicing.}
	\label{fig:heyvl-transformations-extended}
	\small{}
	\begin{tabular}{l|l}
		\toprule
		$\stmt$                            & $\stmtTransformed$                                                                                                                                 \\
		\midrule
		$\stmtAsgn{x}{a}$                  & $\stmtAsgn{x}{\ite{\varEnabled{\stmt}}{a}{x}}$                                                                                                     \\
		$\stmtAssert{\expb}$               & $\stmtAssert{\ite{\varEnabled{\stmt}}{\expb}{\infty}}$                                                                                             \\
		$\stmtAssume{\expb}$               & $\stmtAssume{\ite{\varEnabled{\stmt}}{\expb}{\infty}}$                                                                                             \\
		$\stmtTick{\expb}$                 & $\stmtTick{\ite{\varEnabled{\stmt}}{\expb}{0}}$                                                                                                    \\
		$\coAssert{\expb}$                 & $\coAssert{\ite{\varEnabled{\stmt}}{\expb}{0}}$                                                                                                    \\
		$\coAssume{\expb}$                 & $\coAssume{\ite{\varEnabled{\stmt}}{\expb}{0}}$                                                                                                    \\
		$\stmtDemonic{\stmtOne}{\stmtTwo}$ & $\stmtDemonic{\stmtSeq{\stmtAssume{\embed{\varEnabled{\stmtOne}}}}{\stmtOne}}{\stmtSeq{\stmtAssume{\embed{\varEnabled{\stmtTwo}}}}{\stmtTwo}}$     \\
		$\stmtAngelic{\stmtOne}{\stmtTwo}$ & $\stmtAngelic{\stmtSeq{\coAssume{\embed{\neg\varEnabled{\stmtOne}}}}{\stmtOne}}{\stmtSeq{\coAssume{\embed{\neg\varEnabled{\stmtTwo}}}}{\stmtTwo}}$ \\
		\bottomrule
	\end{tabular}
\end{table}

For the implementation (\Cref{sec:implementation}), we introduced a program transformation that introduces a new variable $\varEnabled{\stmt}$ for each potentially sliceable statement $\stmt$ in the \HeyVL{}, executing $\stmt$ only if $\varEnabled{\stmt}$ is true (and otherwise skipping it).
This transformation is logically equivalent to a conditional execution of $\stmt$:
\[
	\stmtIf{enabled_{\stmtI}}{\stmtI}{}
\]
In the implementation, we make use of better encodings to avoid potential exponential blow-up of the verification pre-expectation.
As an example, consider a program consisting of two assignments $\stmt = \stmtAsgn{x}{a} \symSemi \stmtAsgn{y}{b}$.
The naive transformation would yield
\[
	\stmtTransformed = \stmtIf{\varEnabled{\stmtOne}}{\stmtAsgn{x}{a}}{} \symSemi \stmtIf{\varEnabled{\stmtTwo}}{\stmtAsgn{y}{b}}{}~,
\]
so that the verification pre-expectation w.r.t. to post $\expa$ would evaluate to the following:
\begin{align*}
	\vc{\stmtTransformed}(\expa) = & \quad \embed{\varEnabled{\stmtOne}} \impl (                                                                                                              \\
	                               & \quad\quad (\embed{\varEnabled{\stmtTwo}} \impl \expa\substBy{y}{b}\substBy{x}{a}) \sqcap (\embed{\neg \varEnabled{\stmtTwo}} \impl \expa\substBy{x}{a}) \\
	                               & \quad ) \sqcap \embed{\neg \varEnabled{\stmtOne}} \impl (                                                                                                \\
	                               & \quad\quad (\embed{\varEnabled{\stmtTwo}} \impl \expa\substBy{y}{b}) \sqcap (\embed{\neg \varEnabled{\stmtTwo}} \impl \expa)                             \\
	                               & \quad )~.
\end{align*}
This expression contains four occurrences of $\expa$, one for each possible combination of the two boolean variables $\varEnabled{\stmtOne}$ and $\varEnabled{\stmtTwo}$.
Instead, we use improved transformations that avoid this blow-up by pushing the conditionals into the statements themselves.
The complete set of transformations is shown in \Cref{fig:heyvl-transformations-extended}.
For the above example, we obtain
\[
	\vc{\stmtTransformed}(\expa) = \expa\substBy{y}{\ite{\varEnabled{\stmtTwo}}{b}{y}}\substBy{x}{\ite{\varEnabled{\stmtOne}}{a}{x}}~.
\]
This expression contains only one occurrence of $\expa$, compared to four in the naive transformation.

The following theorems state that these transformations, as well as the respective duals (added explicitly here), are correct.
The theorems are stated for non-recursive transformations, and soundness of recursive transformations holds by induction on the structure of the program.

We begin with the soundness for the atomic statements.
\begin{theorem}[Soundness of Atomic Program Transformations for Slicing]
	Let $\stmt$ be an atomic \HeyVL{} program, and let $\stmtTransformed$ be the program obtained by applying the transformations from \Cref{fig:heyvl-transformations-extended}.
	For all $\expa \in \Expectations$, the following holds:
	\[
		\vc{\stmtIf{enabled_\stmt}{\stmt}{\stmtSkip}}(\expa) \quad=\quad \vc{\stmtTransformed}(\expa)~.
	\]
\end{theorem}

\begin{proof}
	\emph{Case $\stmt = \stmtAsgn{x}{a}$.}
	\begin{align*}
		\vc{\stmtTransformed}(\expa) & = \vc{\stmtAsgn{x}{\ite{\varEnabled{\stmt}}{a}{x}}}(\expa)                                 \\
		                             & = \expa\substBy{x}{\ite{\varEnabled{\stmt}}{a}{x}}                                         \\
		                             & = \ite{\varEnabled{\stmt}}{\expa\substBy{x}{a}}{\expa\substBy{x}{x}}                       \\
		                             & = \ite{\varEnabled{\stmt}}{\expa\substBy{x}{a}}{\expa} \tag{$\expa\substBy{x}{x} = \expa$} \\
		                             & = \vc{\stmtIf{enabled_\stmt}{\stmt}{\stmtSkip}}(\expa)
	\end{align*}

	\emph{Case $\stmt = \stmtAssert{\expb}$.}
	\begin{align*}
		\vc{\stmtTransformed}(\expa) & = \vc{\stmtAssert{\ite{\varEnabled{\stmt}}{\expb}{\infty}}}(\expa)                        \\
		                             & = \ite{\varEnabled{\stmt}}{\expb}{\infty} \sqcap \expa                                    \\
		                             & = \ite{\varEnabled{\stmt}}{\expb \sqcap \expa}{\infty \sqcap \expa}                       \\
		                             & = \ite{\varEnabled{\stmt}}{\expb \sqcap \expa}{\expa} \tag{$\infty \sqcap \expa = \expa$} \\
		                             & = \vc{\stmtIf{enabled_\stmt}{\stmt}{\stmtSkip}}(\expa)
	\end{align*}

	\emph{Case $\stmt = \stmtAssume{\expb}$.}
	\begin{align*}
		\vc{\stmtTransformed}(\expa) & = \vc{\stmtAssume{\ite{\varEnabled{\stmt}}{\expb}{\infty}}}(\expa)                      \\
		                             & = \ite{\varEnabled{\stmt}}{\expb}{\infty} \impl \expa                                   \\
		                             & = \ite{\varEnabled{\stmt}}{\expb \impl \expa}{\infty \impl \expa}                       \\
		                             & = \ite{\varEnabled{\stmt}}{\expb \impl \expa}{\expa} \tag{$\infty \impl \expa = \expa$} \\
		                             & = \vc{\stmtIf{enabled_\stmt}{\stmt}{\stmtSkip}}(\expa)
	\end{align*}

	\emph{Case $\stmt = \stmtTick{\aexpr}$.}
	\begin{align*}
		\vc{\stmtTransformed}(\expa) & = \vc{\stmtTick{\ite{\varEnabled{\stmt}}{\aexpr}{0}}}(\expa)                \\
		                             & = \expa + \ite{\varEnabled{\stmt}}{\aexpr}{0}                               \\
		                             & = \ite{\varEnabled{\stmt}}{\expa + \aexpr}{\expa + 0}                       \\
		                             & = \ite{\varEnabled{\stmt}}{\expa + \aexpr}{\expa} \tag{$\expa + 0 = \expa$} \\
		                             & = \vc{\stmtIf{enabled_\stmt}{\stmt}{\stmtSkip}}(\expa)
	\end{align*}

	\emph{Case $\stmt = \coAssert{\expb}$.}
	\begin{align*}
		\vc{\stmtTransformed}(\expa) & = \vc{\coAssert{\ite{\varEnabled{\stmt}}{\expb}{0}}}(\expa)                          \\
		                             & = \ite{\varEnabled{\stmt}}{\expb}{0} \sqcup \expa                                    \\
		                             & = \ite{\varEnabled{\stmt}}{\expb \sqcup \expa}{0 \sqcup \expa}                       \\
		                             & = \ite{\varEnabled{\stmt}}{\expb \sqcup \expa}{\expa} \tag{$0 \sqcup \expa = \expa$} \\
		                             & = \vc{\stmtIf{enabled_\stmt}{\stmt}{\stmtSkip}}(\expa)
	\end{align*}

	\emph{Case $\stmt = \coAssume{\expb}$.}
	\begin{align*}
		\vc{\stmtTransformed}(\expa) & = \vc{\coAssume{\ite{\varEnabled{\stmt}}{\expb}{0}}}(\expa)                            \\
		                             & = \ite{\varEnabled{\stmt}}{\expb}{0} \coimpl \expa                                     \\
		                             & = \ite{\varEnabled{\stmt}}{\expb \coimpl \expa}{0 \coimpl \expa}                       \\
		                             & = \ite{\varEnabled{\stmt}}{\expb \coimpl \expa}{\expa} \tag{$0 \coimpl \expa = \expa$} \\
		                             & = \vc{\stmtIf{enabled_\stmt}{\stmt}{\stmtSkip}}(\expa)	\qedhere
	\end{align*}
\end{proof}

For the non-atomic statements, i.e. just the two nondeterminstic branches, the tranformation introduces one variable for each branch.
However, the case where both branches are disabled (i.e. in a state $\State$ where $\State(\varEnabled{\stmtOne}) = \false$ and $\State(\varEnabled{\stmtTwo}) = \false$) does not correspond to a sub-program of the original program.
Therefore, our implementation explicitly keeps track of the additional constraint that at least one branch must be enabled.

\begin{theorem}[Soundness of Nondeterministic Choice Transformations for Slicing]
	Let $\stmtTransformed$ be the tranformed version of a \HeyVL{} program $\stmt$ according to \Cref{fig:heyvl-transformations-extended}.
	Let $\expa \in \Expectations$ and $\State$ be a program state.
	For the demonic choice $\stmt = \stmtDemonic{\stmtOne}{\stmtTwo}$, we have:
	\begin{align*}
		\vc{\stmtTransformed}(\expa)(\State) \quad & =\quad \begin{cases}
			                                                    \vc{\stmtOne}(\expa)(\State) & \text{\small{}if $\State(\varEnabled{\stmtOne}) = \true$ and $\State(\varEnabled{\stmtTwo}) = \false$,}  \\
			                                                    \vc{\stmtTwo}(\expa)(\State) & \text{\small{}if $\State(\varEnabled{\stmtOne}) = \false$ and $\State(\varEnabled{\stmtTwo}) = \true$,}  \\
			                                                    \vc{\stmt}(\expa)(\State)    & \text{\small{}if $\State(\varEnabled{\stmtOne}) = \true$ and $\State(\varEnabled{\stmtTwo}) = \true$,}   \\
			                                                    \infty                       & \text{\small{}if $\State(\varEnabled{\stmtOne}) = \false$ and $\State(\varEnabled{\stmtTwo}) = \false$.}
		                                                    \end{cases}
		\intertext{Dually, for the angelic choice $\stmt = \stmtAngelic{\stmtOne}{\stmtTwo}$, we have:}
		\vc{\stmtTransformed}(\expa)(\State) \quad & =\quad \begin{cases}
			                                                    \vc{\stmtOne}(\expa)(\State) & \text{\small{}if $\State(\varEnabled{\stmtOne}) = \true$ and $\State(\varEnabled{\stmtTwo}) = \false$,}  \\
			                                                    \vc{\stmtTwo}(\expa)(\State) & \text{\small{}if $\State(\varEnabled{\stmtOne}) = \false$ and $\State(\varEnabled{\stmtTwo}) = \true$,}  \\
			                                                    \vc{\stmt}(\expa)(\State)    & \text{\small{}if $\State(\varEnabled{\stmtOne}) = \true$ and $\State(\varEnabled{\stmtTwo}) = \true$,}   \\
			                                                    0                            & \text{\small{}if $\State(\varEnabled{\stmtOne}) = \false$ and $\State(\varEnabled{\stmtTwo}) = \false$.}
		                                                    \end{cases}
	\end{align*}
\end{theorem}

\begin{proof}
	\emph{Case $\stmt = \stmtDemonic{\stmtOne}{\stmtTwo}$.}
	\begin{align*}
		\vc{\stmtTransformed}(\expa) & = \vc{\stmtDemonic{\stmtSeq{\stmtAssume{\embed{\varEnabled{\stmtOne}}}}{\stmtOne}}{\stmtSeq{\stmtAssume{\embed{\varEnabled{\stmtTwo}}}}{\stmtTwo}}}(\expa)     \\
		                             & = \vc{\stmtSeq{\stmtAssume{\embed{\varEnabled{\stmtOne}}}}{\stmtOne}}(\expa) \sqcap \vc{\stmtSeq{\stmtAssume{\embed{\varEnabled{\stmtTwo}}}}{\stmtTwo}}(\expa) \\
		                             & = (\embed{\varEnabled{\stmtOne}} \impl \vc{\stmtOne}(\expa)) \sqcap (\embed{\varEnabled{\stmtTwo}} \impl \vc{\stmtTwo}(\expa))                                 \\
		                             & = \ite{\varEnabled{\stmtOne}}{\vc{\stmtOne}(\expa)}{\infty} \sqcap \ite{\varEnabled{\stmtTwo}}{\vc{\stmtTwo}(\expa)}{\infty}~.
	\end{align*}

	\emph{Case $\stmt = \stmtAngelic{\stmtOne}{\stmtTwo}$.}
	\begin{align*}
		\vc{\stmtTransformed}(\expa) & = \vc{\stmtAngelic{\stmtSeq{\coAssume{\embed{\neg\varEnabled{\stmtOne}}}}{\stmtOne}}{\stmtSeq{\coAssume{\embed{\neg\varEnabled{\stmtTwo}}}}{\stmtTwo}}}(\expa)     \\
		                             & = \vc{\stmtSeq{\coAssume{\embed{\neg\varEnabled{\stmtOne}}}}{\stmtOne}}(\expa) \sqcup \vc{\stmtSeq{\coAssume{\embed{\neg\varEnabled{\stmtTwo}}}}{\stmtTwo}}(\expa) \\
		                             & = (\embed{\neg\varEnabled{\stmtOne}} \coimpl \vc{\stmtOne}(\expa)) \sqcup (\embed{\neg\varEnabled{\stmtTwo}} \coimpl \vc{\stmtTwo}(\expa))                         \\
		                             & = \ite{\varEnabled{\stmtOne}}{\vc{\stmtOne}(\expa)}{0} \sqcup \ite{\varEnabled{\stmtTwo}}{\vc{\stmtTwo}(\expa)}{0}~. \qedhere
	\end{align*}
\end{proof}

\clearpage

\pagebreak
\section{Benchmarks}
\label{sec:detailed-benchmarks}

This section contains the detailed results of the benchmarks presented in \cref{sec:evaluation}, and we also give sources and descriptions of the programs used in the benchmarks.
The benchmarks were executed on a 2021 Apple MacBook with an M1 Pro chip and 16 GB of RAM.

\newcommand{\featB}{\textsf{B}}
\newcommand{\featP}{\textsf{P}}
\newcommand{\featC}{\textsf{C}}
\newcommand{\featS}{\textsf{S}}
\newcommand{\featL}{\textsf{L}}
\newcommand{\featN}{\textsf{N}}
\newcommand{\featO}{\textsf{O}}

\paragraph{Overview and Features.}
Erroring slices are handled in \cref{sec:results-error-slices} and verifying (verification-witnessing and -preserving slices) are grouped in \cref{sec:results-verifying-slices}.
\Cref{tbl:erroring-results,tbl:verifying-results} show the results of different strategies for finding slices.
For each example program, we show a list of features that are present in the example.
Programs are either boolean (\featB) or probabilistic (\featP).
We have programs that reason about continuous sampling (\featC) taken from and using the technique of~\cite{continuous25}.
Programs with recursion use the specification statement (\featS) as described in \Cref{sec:specification-statements}.
Programs with loops are marked by \featL, and here the Park induction proof rule (\cref{sec:park-induction}) is used.
Some examples make use of explicit non-deterministic choices (\featN).
Finally, we include programs that encode conditioning/observations (\featO) using the definition of the conditional expected value $\cwp{\stmt}(\expa) = \nicefrac{\wp{\stmt}(\expa)}{\wlp{\stmt}(1)}$ as described in \cite{DBLP:journals/entcs/0001KKOGM15}.

\paragraph{Procs and Coprocs.}
In the examples below, we use Caesar's syntax to declare two kinds of procedures in HeyVL, $\symProc$ and $\symcoProc$.
Both kinds of declarations consist of a name, a list of input and output variables, and a pre $\expa$, a post $\expb$, and a body $\stmt$.
The only difference is that $\symProc$s are required to \emph{verify} ($\models \lowerTriple{\expa}{\stmt}{\expb}$) and that $\symcoProc$s are required to \emph{co-verify} ($\models \upperTriple{\expa}{\stmt}{\expb}$).
These dual kinds of specifications were introduced in \Cref{par:heyvl-specifications}.

According to our theory, for slicing in $\symProc$s, we select e.g. \emph{reductive} statements for slicing by default when slicing for errors.
Dually, in $\symcoProc$s, we select \emph{extensive} statements by default when slicing for errors.
Similarly, we select \emph{extensive} statements in $\symProc$s by default when searching for verification-witnessing slices, while dually in $\symcoProc$ we select \emph{reductive} statements.

\begin{figure}[t]

    \centering
    \begin{spreadlines}{0ex}
        \begin{align*}
            \gutter{1}  & \stmtAsgn{x}{init\_x} \symSemi                          \\
            \gutter{2}  & \stmtAsgn{cont}{\exprTrue} \symSemi                     \\
            \gutter{3}  & \stmtAnnotate{invariant}{\varInvariant}{}               \\
            \gutter{4}  & \headerWhile{cont}~\blockStart                          \\
            \gutter{5}  & \quad \stmtRasgn{prob\_choice}{\exprFlip{0.5}} \symSemi \\
            \gutter{6}  & \quad \stmtIfStart{prob\_choice}                        \\
            \gutter{7}  & \quad \quad \stmtAsgn{cont}{\exprFalse}                 \\
            \gutter{8}  & \quad \blockEnd~\stmtElseStart                          \\
            \gutter{9}  & \quad \quad \stmtAsgn{x}{A}                             \\
            \gutter{10} & \quad \blockEnd                                         \\
            \gutter{11} & \blockEnd
        \end{align*}
    \end{spreadlines}
    \caption{Loop template $\stmt(A,\varInvariant)$ with loop invariant $I$ and assigning $A$ to $x$ in the loop, which is verified against different upper bound specifications.}
    \label{fig:geometric_loop_template}
\end{figure}

\begin{table}[t]
    \caption{Diagnostic messages on verifying different instantiations of the loop template in~\cref{fig:geometric_loop_template}.
        The first two columns show whether the verification succeeds and the user diagnostics.
        Columns three and four contain the pre and post.
        The final two columns give the invariant, and the assignment to $x$ that is applied in every but the final loop iteration.
        The first case produces no diagnostic messages as the program verifies and all assumptions are used.
    }
    \label{tbl:geometric_loop_template_instances}
    \begin{minipage}{1\textwidth}
        \begin{center}
            \adjustbox{width=\textwidth}{%
                \begin{tabular}{*{2}{ l } | *{4}{ l }}
                    \toprule
                                 & Diagnostic                               & Pre $\expa$                                 & Post $\expb$        & Invariant $\varInvariant$                           & $A$     \\
                    \midrule
                    $\tikzcmark$ & Verified.                                & $\embed{init\_x \neq 0} \sqcup \frac{1}{4}$ & $\iverson{x == 1}$  & $\mathrm{ite}(cont, \frac{1}{2} \iverson{x == 1} +$ & $x+1$   \\
                                 &                                          &                                             &                     & $\frac{1}{4} \iverson{x == 0}, {\iverson{x == 1}})$ &         \\
                    $\tikzxmark$ & Post is part of the error.               & $init\_x + 1$                               & $x + 1$             & $\exprIte{cont}{x + 1}{x}$                          & $x+1$   \\
                    $\tikzxmark$ & Pre might not entail the invariant.      & $0$                                         & $x$                 & $x$                                                 & $1$     \\
                    $\tikzxmark$ & Invariant might not be inductive.        & $init\_x + 1$                               & $x$                 & $x + 1$                                             & $x+1$   \\
                    $\tikzcmark$ & While could be an if statement.          & $init\_x + 1$                               & $x$                 & $\exprIte{cont}{x + 1}{x}$                          & $x + 1$ \\
                    $\tikzcmark$ & Invariant not necessary for inductivity. & $\embed{true}$                              & $\embed{\neg cont}$ & $\embed{cont}$                                      & $x+1$   \\
                    \bottomrule
                \end{tabular}
            }
        \end{center}
    \end{minipage}%
\end{table}

\subsection{Results for Error-Witnessing Slices}\label{sec:results-error-slices}

\begin{table}[htbp]
    \caption{Results of the two strategies for finding error-witnessing slices.
        The column \enquote{LoC} shows the total lines of code in the program and $|S|$ the number of potentially sliceable statements.
        For each strategy, the columns $|S|$ and $t$ show the number of statements in the slice and the runtime in milliseconds, respectively.
        Features: Boolean (\featB), Probabilistic (\featP), Continuous Sampling (\featC), Specification statement/recursion (\featS), Loops (\featL), Non-determinism (\featN), and Observations/conditioning (\featO).}
    \label{tbl:erroring-results}
    \centering
    \begin{tabular}{>{\footnotesize}llrr *{2}{crr}}
        \toprule
        \multirow{2}{*}{Program}     & \multirow{2}{*}{Feat.}       & \multirow{2}{*}{LoC} & \multirow{2}{*}{$|S|$} &  & \multicolumn{2}{c}{\sliceMethod{first}} &       & \multicolumn{2}{c}{\sliceMethod{opt}}                 \\
        \cmidrule(lr){6-7}\cmidrule(lr){9-10}
                                     &                              &                      &                        &  & $|S|$                                   & $t$   &                                       & $|S|$ & $t$   \\
        \texttt{list\_access\_error} & \featB{}                     & 7                    & 3                      &  & 1                                       & 2     &                                       & 1     & 1     \\
        \texttt{intro\_error}        & \featP{}, \featN{}           & 16                   & 5                      &  & 4                                       & 1     &                                       & 4     & 2     \\
        \texttt{leino05\_1}          & \featB{}                     & 13                   & 2                      &  & 1                                       & 1     &                                       & 1     & 1     \\
        \texttt{park\_2}             & \featP{}, \featL{}           & 16                   & 3                      &  & 2                                       & 1     &                                       & 1     & 1     \\
        \texttt{park\_3}             & \featP{}, \featL{}           & 16                   & 3                      &  & 2                                       & 1     &                                       & 1     & 1     \\
        \texttt{park\_4}             & \featP{}, \featL{}           & 16                   & 3                      &  & 1                                       & 1     &                                       & 1     & 1     \\
        \texttt{condand}             & \featP{}, \featL{}           & 19                   & 3                      &  & 1                                       & 1     &                                       & 1     & 2     \\
        \texttt{fcall}               & \featP{}, \featL{}           & 21                   & 3                      &  & 2                                       & 1     &                                       & 1     & 2     \\
        \texttt{hyper}               & \featP{}, \featL{}           & 26                   & 3                      &  & 1                                       & 1     &                                       & 1     & 2     \\
        \texttt{linear01}            & \featP{}, \featL{}           & 18                   & 3                      &  & 1                                       & 1     &                                       & 0     & 1     \\
        \texttt{prdwalk}             & \featP{}, \featL{}           & 57                   & 3                      &  & 1                                       & 2     &                                       & 1     & 3     \\
        \texttt{six\_sided\_die}     & \featP{}, \featO{}           & 10                   & 2                      &  & 2                                       & 1     &                                       & 2     & 1     \\
        \texttt{geo\_recursive}      & \featP{}, \featS{}           & 17                   & 2                      &  & 1                                       & 2     &                                       & 1     & 3     \\
        \texttt{irwin\_hall16}       & \featP{}, \featL{}, \featC{} & 13                   & 3                      &  & 1                                       & 22    &                                       & 1     & 24    \\
        \texttt{monte\_carlo\_pi4}   & \featP{}, \featL{}, \featC{} & 16                   & 3                      &  & 2                                       & 11    &                                       & 1     & 24    \\
        \texttt{tortoise\_hare12}    & \featP{}, \featL{}, \featC{} & 31                   & 4                      &  & 2                                       & 10400 &                                       & 1     & 10400 \\
        \bottomrule
    \end{tabular}
\end{table}

The results for error-witnessing slices are shown in \Cref{tbl:erroring-results}.
\texttt{list\_access\_error} is a simple boolean program to locate a possible out-of-bounds list access.

The probabilistic program \texttt{intro\_error} corresponds to \Cref{fig:intro-error-preserving} and is based on \cite[Example 4.5]{DBLP:journals/scp/NavarroO22}.
The encoding has nondeterminism built-in to enable slicing the branches of a probabilistic choice.
We show the actual HeyVL code used for our tool in \Cref{fig:intro-error-heyvl}.

The program \texttt{leino05\_1} is based on \cite[Page 6]{DBLP:journals/scp/LeinoMS05}.
The examples \texttt{park\_2}, \texttt{park\_3}, and \texttt{park\_4} correspond to the first three rows of \Cref{tbl:geometric_loop_template_instances}, all erroring instances of a geometric loop instantiated from the template~\Cref{fig:geometric_loop_template}.

The examples \texttt{condand}, \texttt{fcall}, \texttt{hyper}, \texttt{linear01}, \texttt{prdwalk}, \texttt{six\_sided\_die} and \texttt{geo\_recursive} are examples of Park induction, conditioning, and recursion from~\cite{DBLP:journals/pacmpl/SchroerBKKM23} and have been modified such that they do not verify anymore.
The example \texttt{six\_sided\_die} is the erroring part of a (failing) two-part verification of a program that simulates a fair six-sided die using conditioning, shown in \Cref{fig:six-sided-die-heyvl}.
We show the example \texttt{geo\_recursive} in \Cref{fig:geo-recursive-heyvl}, whose recursive procedure call is translated using the specification statement encoding (\Cref{sec:specification-statements}).

Finally, the examples \texttt{irwin\_hall16}, \texttt{monte\_carlo\_pi4}, and \texttt{tortoise\_hare12} are HeyVL encodings of probabilistic programs with sampling from a continuous distribution based on~\cite{continuous25}, modified to fail verification.
All use the proof rule of Park induction for a loop, and encode sampling statements from continuous distributions as described in the above source.
For the example \texttt{irwin\_hall16}, we show the HeyVL file used for our tool in \Cref{fig:irwin-hall-heyvl}.

In all of the above examples, our tool only selects reductive statements in $\symProc$s and dually extensive statements in $\symcoProc$s for slicing (column 4, $|S|$).
By \Cref{theorem:reductive-slice}, we ensure that we obtain an error-witnessing slice.
This behavior is the default for programs that fail verification.

\begin{figure}[htbp]
\begin{lstlisting}[language=HeyVL,firstline=3]
coproc binary_sample_error() -> (r: EUReal)
    pre 1/3
    post [r >= 2]
{
    var b0: UInt
    var b1: UInt
    var b2: UInt
    var b3: UInt

    var choice: Bool = flip(0.5)
    if \cup {
        coassume !?(choice); b0 = 0
    } else {
        coassume !?(!choice); b0 = 1
    }

    choice = flip(0.5)
    if \cup {
        coassume !?(choice); b1 = 0 // the coassume is sliced
    } else {
        coassume !?(!choice); b1 = 1
    }

    r = b0 + 2 * b1
}
\end{lstlisting}
    \caption{HeyVL code for the probabilistic program \texttt{intro\_error}, corresponding to \Cref{fig:intro-error-preserving}.
        It does not verify, as $\nicefrac{1}{3}$ is not an upper bound for the probability of $\texttt{r} \geq 2$ on termination.
        Our tool slices the $\symUp\symAssume$ statement in line 19, and therefore it indicates that only the $\symUp\symAssume$ statements in lines 12, 14, and 21 are relevant for the verification failure.}
    \label{fig:intro-error-heyvl}
\end{figure}

\begin{figure}[htbp]
    \centering
    \begin{subfigure}[t]{0.48\textwidth}
\begin{lstlisting}[language=HeyVL,firstline=15]
coproc die_wp() -> (r: EUReal)
    pre 2.625
    post r
{
    var a: Bool = flip(0.5)
    var b: Bool = flip(0.5)
    var c: Bool = flip(0.5)
    r = 4 * [a] + 2 * [b] + [c] + 1
    // changed from r <= 6
    assert ?(r <= 5)
}
\end{lstlisting}
        \caption{HeyVL code to prove that $\wp{\stmt}(r) \expleq 2.625$. This verification succeeds, and nothing is sliced.}
    \end{subfigure}
    \hfill
    \begin{subfigure}[t]{0.48\textwidth}
\begin{lstlisting}[language=HeyVL,firstline=14]
proc die_wlp() -> (r: EUReal)
    pre 0.75
    post 1
{
    var a: Bool = flip(0.5)
    var b: Bool = flip(0.5)
    var c: Bool = flip(0.5)
    r = 4 * [a] + 2 * [b] + [c] + 1
     // changed from r <= 6
    assert ?(r <= 5)
}
\end{lstlisting}
        \caption{HeyVL code to prove that $\wlp{\stmt}(r) \expgeq 0.75$. This verification fails, and both the $\symPost$ and the $\symAssert$ statement are reported as part of the error.}
    \end{subfigure}
    \caption{HeyVL code to check that $\cwp{\stmt}(r) \expleq 3.5 = \nicefrac{2.625}{0.75}$ for the probabilistic program \texttt{six\_sided\_die}, based on the decomposition of conditioning by~\cite{DBLP:journals/entcs/0001KKOGM15}. The last $\symAssert$ statement models a conditioning statement, i.e. $\stmtObserve{r \leq 5}$ (changed from the correct $r \leq 6$). The error report points to a failure of the $\symWlp$ part, and notices that both the post and the $\symObserve$ statement are relevant for the failure.}
    \label{fig:six-sided-die-heyvl}
\end{figure}

\begin{figure}[htbp]
\begin{lstlisting}[language=HeyVL, firstline=4]
coproc geo1(init_c: UInt, init_f: UInt) -> (c: UInt, f: UInt)
    pre [init_f == 1] * (init_c + 1) + [init_f != 1] * init_c
    post c
{
    var prob_choice: Bool
    c = init_c
    f = init_f
    if f != 1 { // it should be f == 1 in the verifying program
        prob_choice = flip(0.5)
        if prob_choice {
            f = 0
        } else {
            c = c + 1
        }
        c, f = geo1(c, f) // pre might not hold
    } else {}
}
\end{lstlisting}
    \caption{HeyVL code for the probabilistic program \texttt{geo\_recursive}, based on the example from \cite{DBLP:journals/pacmpl/SchroerBKKM23}.
        The procedure call in line 15 is encoded using a specification statement, as described in \Cref{sec:specification-statements}.
        Our tool slices all but the $\symUp\symAssume$ statement in the specification statement encoding and therefore prints the error message \enquote{pre might not hold}, located in line 15.
        In this case, slicing reveals that the precondition of the recursive call is not satisfied, allowing the user to identify a wrong conditional (line 8).}
    \label{fig:geo-recursive-heyvl}
\end{figure}

\begin{figure}[htbp]
\begin{lstlisting}[language=HeyVL,firstline=2]
coproc irwin_hall(M : UReal) -> (x : UReal)
    pre 0.531*M // changed pre so that pre might not entail invariant
    post x
{
    x = 0
    var i : UReal = 1
    @invariant(ite(i <= M, (x + 0.532*((M-i) + 1)), x))
    while i <= M {
        // var inc : UReal = unif[0,1]
        var inc : UReal; var N : UInt = 16; var j : UInt = unif(0, 15); cohavoc inc; coassume ?!(j / N <= inc && inc <= (j+1) / N)
        x = x + inc
        i = i + 1
    }
}
\end{lstlisting}
    \caption{HeyVL code for the probabilistic program \texttt{irwin\_hall16}, based on the encoding of continuous sampling by~\cite{continuous25}. Our tool reports that the \enquote{pre does not entail the invariant}, based on slicing the encoding of Park induction (generated by the \texttt{@invariant} annotation) and application of \Cref{theorem:park-error-messages}.
        This hints that the invariant is correct, but that the pre is wrong (should be $0.532 \cdot M$ instead of $0.531 \cdot M$).}
    \label{fig:irwin-hall-heyvl}
\end{figure}

\subsection{Results for Verifying Slices}\label{sec:results-verifying-slices}

\begin{table}[hbtp]
    \caption{Results of the different strategies for finding verifying slices.
        The column \enquote{LoC} shows the total lines of code in the program and $|S|$ the number of potentially sliceable statements.
        For each strategy, the columns $|S|$ and $t$ show the number of statements in the slice and the runtime in milliseconds, respectively.
        Features: Boolean (\featB), Probabilistic (\featP), Continuous Sampling (\featC), Specification statement/recursion (\featS), Loops (\featL), Non-determinism (\featN).}
    \label{tbl:verifying-results}
    \centering

    \begin{tabular}{c>{\footnotesize}lllr *{4}{crr}}
        \toprule
                                                            & \multirow{2}{*}{Program}    & \multirow{2}{*}{Feat.} & \multirow{2}{*}{LoC} & \multirow{2}{*}{$|S|$} &  & \multicolumn{2}{c}{\sliceMethod{core}} &     & \multicolumn{2}{c}{\sliceMethod{mus}} &       & \multicolumn{2}{c}{\sliceMethod{sus}} &  & \multicolumn{2}{c}{\sliceMethod{e-f}}                         \\
        \cmidrule(lr){7-8}\cmidrule(lr){10-11}\cmidrule{13-14} \cmidrule{16-17}
                                                            &                             &                        &                      &                        &  & $|S|$                                  & $t$ &                                       & $|S|$ & $t$                                   &  & $|S|$                                 & $t$  &  & $|S|$ & $t$ \\
        \midrule
        \\
                                                            & \texttt{barros\_program1}   & \featB{}               & 11                   & 5                      &  & 5                                      & 0   &                                       & 3     & 10                                    &  & 3                                     & 11   &  & 3     & 7   \\
                                                            & \texttt{barros\_program2}   & \featB{}               & 11                   & 4                      &  & 4                                      & 0   &                                       & 2     & 9                                     &  & 2                                     & 10   &  & 3     & 6   \\
                                                            & \texttt{barros\_program3}   & \featB{}               & 11                   & 5                      &  & 5                                      & 0   &                                       & 3     & 9                                     &  & 3                                     & 11   &  & 5     & 7   \\
                                                            & \texttt{barros\_program4}   & \featB{}               & 17                   & 8                      &  & 5                                      & 0   &                                       & 2     & 10                                    &  & 2                                     & 14   &  & 4     & 6   \\
        \multirow{-5}{*}{\rotatebox{90}{\emph{preserving}}} & \texttt{barros\_program5}   & \featB{}               & 16                   & 7                      &  & 4                                      & 0   &                                       & 2     & 10                                    &  & 2                                     & 13   &  & 3     & 7   \\
                                                            & \texttt{barros\_program6}   & \featB{}               & 10                   & 4                      &  & 3                                      & 0   &                                       & 1     & 9                                     &  & 1                                     & 10   &  & 3     & 6   \\
                                                            & \texttt{barros\_program7}   & \featB{}               & 9                    & 3                      &  & 2                                      & 0   &                                       & 2     & 9                                     &  & 2                                     & 9    &  & 2     & 6   \\
                                                            & \texttt{barros\_program8}   & \featB{}               & 14                   & 4                      &  & 3                                      & 0   &                                       & 3     & 9                                     &  & 3                                     & 10   &  & 4     & 7   \\
                                                            & \texttt{navarro20\_page31}  & \featB{}               & 18                   & 7                      &  & 4                                      & 0   &                                       & 2     & 10                                    &  & 2                                     & 12   &  & 3     & 6   \\
                                                            & \texttt{gehr18\_1}          & \featP{}               & 14                   & 4                      &  & 4                                      & 0   &                                       & 2     & 9                                     &  & 2                                     & 9    &  & 2     & 5   \\
                                                            & \texttt{gehr18\_2}          & \featP{}               & 14                   & 4                      &  & 4                                      & 0   &                                       & 2     & 9                                     &  & 2                                     & 12   &  & 2     & 6   \\
                                                            & \texttt{gehr18\_3}          & \featP{}               & 14                   & 3                      &  & 2                                      & 0   &                                       & 1     & 9                                     &  & 1                                     & 8    &  & 2     & 5   \\
                                                            & \texttt{navarro20\_ex4\_3}  & \featP{}               & 12                   & 3                      &  & 3                                      & 0   &                                       & 1     & 8                                     &  & 1                                     & 9    &  & 2     & 5   \\
                                                            & \texttt{navarro20\_ex4\_5}  & \featP{}               & 20                   & 10                     &  & 10                                     & 1   &                                       & 9     & 203                                   &  & 9                                     & 552  &  & 9     & 25  \\
                                                            & \texttt{navarro20\_ex5\_8}  & \featP{}, \featL{}     & 32                   & 5                      &  & 5                                      & 20  &                                       & 5     & 3700                                  &  & 5                                     & 3690 &  & -     & -   \\
                                                            & \texttt{navarro20\_figure7} & \featP{}               & 73                   & 37                     &  & 21                                     & 0   &                                       & 12    & 1370                                  &  & -                                     & -    &  & 13    & 21  \\
                                                            & \texttt{navarro20\_figure8} & \featP{}               & 73                   & 37                     &  & 13                                     & 0   &                                       & 7     & 53                                    &  & -                                     & -    &  & 7     & 8   \\
                                                            & \texttt{2drwalk}            & \featP{}, \featL{}     & 219                  & 75                     &  & 18                                     & 7   &                                       & 7     & 116                                   &  & -                                     & -    &  & 10    & 25  \\
                                                            & \texttt{bayesian\_network}  & \featP{}, \featL{}     & 102                  & 26                     &  & 4                                      & 0   &                                       & 1     & 19                                    &  & 1                                     & 10   &  & 3     & 7   \\
                                                            & \texttt{prspeed}            & \featP{}, \featL{}     & 40                   & 9                      &  & 6                                      & 1   &                                       & 6     & 25                                    &  & 6                                     & 67   &  & -     & -   \\
                                                            & \texttt{rdspeed}            & \featP{}, \featL{}     & 43                   & 11                     &  & 10                                     & 2   &                                       & 8     & 61                                    &  & 8                                     & 226  &  & 8     & 15  \\
                                                            & \texttt{rdwalk}             & \featP{}, \featL{}     & 19                   & 5                      &  & 4                                      & 0   &                                       & 3     & 11                                    &  & 3                                     & 12   &  & 3     & 11  \\
                                                            & \texttt{sprdwalk}           & \featP{}, \featL{}     & 23                   & 6                      &  & 5                                      & 0   &                                       & 4     & 118                                   &  & 4                                     & 15   &  & 4     & 8   \\
                                                            & \texttt{unif\_gen4\_wlp}    & \featP{}, \featL{}     & 123                  & 7                      &  & 5                                      & 58  &                                       & 3     & 1770                                  &  & 3                                     & 2560 &  & 3     & 191 \\
                                                            & \texttt{burglar\_alarm\_wp} & \featP{}               & 35                   & 4                      &  & 4                                      & 0   &                                       & 1     & 10                                    &  & 1                                     & 11   &  & 2     & 6   \\
        \\
        \midrule
        \\
                                                            & \texttt{park\_1}            & \featP{}, \featL{}     & 16                   & 3                      &  & 3                                      & 0   &                                       & 3     & 11                                    &  & 3                                     & 10   &  & 3     & 9   \\
                                                            & \texttt{park\_5}            & \featP{}, \featL{}     & 16                   & 3                      &  & 3                                      & 0   &                                       & 2     & 9                                     &  & 2                                     & 9    &  & 3     & 7   \\
                                                            & \texttt{park\_6}            & \featP{}, \featL{}     & 16                   & 3                      &  & 1                                      & 0   &                                       & 1     & 8                                     &  & 1                                     & 8    &  & 1     & 5   \\
                                                            & \texttt{algorithm\_r}       & \featP{}, \featL{}     & 16                   & 4                      &  & 4                                      & 105 &                                       & 3     & 125                                   &  & 3                                     & 132  &  & 3     & 18  \\
        \multirow{-5}{*}{\rotatebox{90}{\emph{witnessing}}} & \texttt{bn\_passified}      & \featP{}, \featN{}     & 47                   & 25                     &  & 17                                     & 370 &                                       & 10    & 3390                                  &  & 10                                    & 3810 &  & -     & -   \\
        \\
        \bottomrule
    \end{tabular}
\end{table}

\Cref{tbl:erroring-results} shows the results for verifying slices, separated by whether we are searching for verification-\emph{preserving} or -\emph{witnessing} slices.

\paragraph{Verification-preserving slices.}\label{sec:benchmarks-verification-preserving}
The first eight examples are from \cite{DBLP:journals/fac/BarrosCHP12}.
In particular, \texttt{barros\_program6} refutes the claim of theirs that forwards and backwards reasoning is necessary to find minimal slices (their Program 6), as we explain in \Cref{sec:related-work}.
The HeyVL program used for the tool is shown in \Cref{fig:barros6-heyvl}.
With the same goal, program \texttt{navarro20\_page31} is based on the example from \cite[Page 31]{DBLP:journals/scp/NavarroO22} and the HeyVL version is shown in \Cref{fig:navarro20-heyvl}.

The examples \texttt{gehr18\_1}, \texttt{gehr18\_2}, and \texttt{gehr18\_3} are based on Bayesian networks encoded as probabilistic programs, from \cite{DBLP:conf/pldi/GehrMTVWV18}.

The following examples \texttt{2drwalk}, \texttt{bayesian\_network}, \texttt{prspeed}, \texttt{rdspeed}, \texttt{rdwalk}, \texttt{sprdwalk}, and \texttt{unif\_gen4\_wlp} are probabilistic programs from \cite{DBLP:journals/pacmpl/SchroerBKKM23}.

All of the above examples have not been modified from their source, except to add the necessary \texttt{@slice\_verify} annotations to let our tool select the assignments for slicing.

\paragraph{Verification-witnessing slices.}
The examples \texttt{park\_1}, \texttt{park\_5}, and \texttt{park\_6} correspond to the instantiations of~\Cref{fig:geometric_loop_template} by the respective lines in \Cref{tbl:geometric_loop_template_instances}.
\texttt{algorithm\_r} corresponds to our \Cref{fig:intro-verification-preserving} and its HeyVL version is shown in \Cref{fig:algorithm-r-heyvl}.
Finally, \texttt{bn\_passified} is a simplified version of \cite[Figure 7]{DBLP:journals/scp/NavarroO22} with the \emph{passification} transformation applied, turning assignments into assumptions.

For all of these examples, our tool automatically selects extensive statements in $\symProc$s and dually reductive statements in $\symcoProc$s when slicing verifying programs.
By \Cref{theorem:extensive-slice}, this ensures we obtain verification-witnessing slices.
Their number is reported in column 4, $|S|$.

\begin{figure}[htbp]
\begin{lstlisting}[language=HeyVL,firstline=4]
proc barros_program6(init_x: Int) -> (x: Int)
    pre ?(true)
    post ?(x >= 100)
{
    x = init_x
    @slice_verify {
        x = x * x
        x = x + 100
        x = x + 50 // this is sliced
    }
}
\end{lstlisting}
    \caption{HeyVL code for the program \texttt{barros\_program6}, corresponding to \Cref{fig:barros6-heyvl} (\cite[Program 6]{DBLP:journals/fac/BarrosCHP12}).
        The \texttt{@slice\_verify} annotation lets our tool select the assignments for slicing.
        Our method can slice the last assignment, while the forwards- or backwards-reasoning methods by \citeauthor{DBLP:journals/fac/BarrosCHP12} cannot.}
    \label{fig:barros6-heyvl}
    \vspace{-1\baselineskip} %
\end{figure}

\begin{figure}[htbp]
\begin{lstlisting}[language=HeyVL,firstline=2]
proc navarro_page31(init_x: Int, init_y: Int) -> (x: Int, y: Int)
    pre ?(init_y > 0)
    post ?(x >= 0)
{
    x = init_x
    y = init_y
    @slice_verify {
        if y > 0 {
            x = 100
            x = x + 50  // this is sliced
            x = x - 100 // this is sliced
        } else {
            x = x - 150 // this is sliced
            x = x - 100 // this is sliced
            x = x + 100 // this is sliced
        }
    }
}
\end{lstlisting}
    \caption{HeyVL code for the program \texttt{navarro20\_page31} corresponding to the example from \cite[Page 31]{DBLP:journals/scp/NavarroO22}.
        The \texttt{@slice\_verify} annotation has Brutus select the assignments for slicing.
        Our method can slice all marked statements, while \citeauthor{DBLP:journals/scp/NavarroO22} claim it needs both forwards and backwards reasoning.}
    \label{fig:navarro20-heyvl}
\end{figure}

\begin{figure}[htbp]
\begin{lstlisting}[language=HeyVL]
coproc algorithm_r_wrong(n: UInt, x: UInt) -> (c: UInt)
    pre !?(n > 0 && n >= x && x >= 1)
    pre 1/n
    post [c == x]
{
    c = 0
    var i: UInt = 1
    @invariant((1/n) * [i <= x && x <= n] + [x == c]) // while could be an if statement
    while i <= n {
        var prob_choice: Bool = flip(1/n)
        if prob_choice {
            c = i
        } else {}
        i = i + 1
    }
}
\end{lstlisting}
    \caption{HeyVL code for the program \texttt{algorithm\_r}, corresponding to \Cref{fig:intro-verification-preserving}. After desugaring the annotated while loop according to the schema of \Cref{fig:heyvl-park-induction}, our tool determines the assumption (II) can be sliced.
        By \Cref{theorem:park-verification-messages}, the tool will print the message \enquote{While loop could be an if statement} to the user.}
    \label{fig:algorithm-r-heyvl}
\end{figure}

\begin{figure}[htbp]
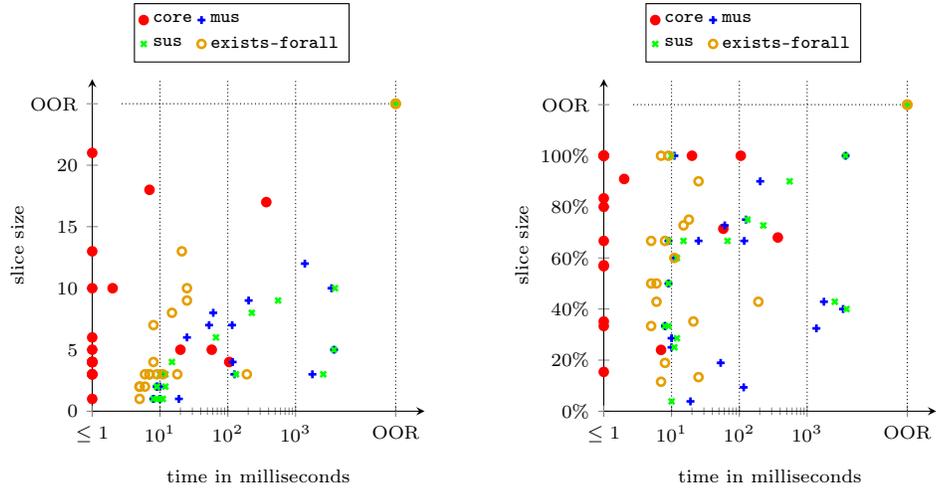

    \centering
    \newcommand{\legendcols}{2}
    \newcommand{\legendstyle}{at={(0.45,1.05)},anchor=south}
    \begin{subfigure}[t]{.45\textwidth}
        \numberruntimeplot{sections/benchmarks/result-verifying.csv}%
        {core,
            mus,
            sus,
            exists_forall}%
        {{\sliceMethod{core}},
            {\sliceMethod{mus}},
            {\sliceMethod{sus}},
            {\sliceMethod{exists-forall}}}%
        {slice size}%
        {time in milliseconds}%
        {states}%
        {27}%
        \caption{\marknewtext{Absolute size of slices vs.\ computation time.}}
        \label{fig:evaluation:scatter_plots:verifying_slices:absolute}
    \end{subfigure}
    \hfill
    \begin{subfigure}[t]{.45\textwidth}
        \relativenumberruntimeplot{sections/benchmarks/result-verifying.csv}%
        {core,
            mus,
            sus,
            exists_forall}%
        {{\sliceMethod{core}},
            {\sliceMethod{mus}},
            {\sliceMethod{sus}},
            {\sliceMethod{exists-forall}}}%
        {slice size}%
        {time in milliseconds}%
        {states}%
        {27}%
        \caption{\marknewtext{Relative size of slices, compared to total sliceable statements, vs.\ computation time.}}
        \label{fig:evaluation:scatter_plots:verifying_slices:relative}
    \end{subfigure}
    \caption{\marknewtext{Scatter plots showing the size of verifying slices computed by different methods with the time needed (log scale). OOR represents computations that did not finish within the timeout of 30 seconds.}}
    \label{fig:evaluation:scatter_plots:verifying_slices}
\end{figure}
\begin{figure}[t]
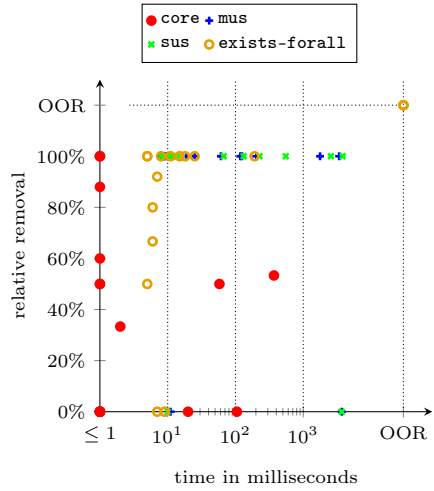

    \centering
    \newcommand{\legendcols}{2}
    \newcommand{\legendstyle}{at={(0.45,1.05)},anchor=south}

    \relativeslicesizeruntimeplotscatter{sections/benchmarks/result-verifying.csv}%
    {core,
        mus,
        sus,
        exists_forall}%
    {{\sliceMethod{core}},
        {\sliceMethod{mus}},
        {\sliceMethod{sus}},
        {\sliceMethod{exists-forall}}}%
    {relative removal}%
    {time in milliseconds}%
    {states}%
    {27}%
    \caption{\marknewtext{Scatter plot showing the ratio of statements removed from verifying slices computed by different methods compared to the optimal slice computed by \sliceMethod{sus}, with the time needed (log scale). OOR represents computations that did not finish within the timeout of 30 seconds.}}
    \label{fig:evaluation:scatter_plots:verifying_slices:relative_removal:scatter}
\end{figure}

\end{document}